\documentclass[12pt]{article}
\usepackage{amsmath}
\usepackage{graphicx}
\usepackage{enumerate}
\usepackage{natbib}
\usepackage{url} % not crucial - just used below for the URL

\newcommand{\blind}{0}

\usepackage[utf8]{inputenc} 
\usepackage{amssymb}
\usepackage{wasysym}
\usepackage{textcomp}
\usepackage{dsfont}
\usepackage{array}

\usepackage{algorithm, algorithmic}
\usepackage{threeparttable,booktabs,multirow}
\usepackage{xcolor}
\usepackage{graphicx}

\usepackage[colorlinks,citecolor=blue,urlcolor=blue]{hyperref}
\usepackage{setspace}
\usepackage[skip=0.5\baselineskip]{caption}
\usepackage{amsthm}

\newtheorem{theorem}{Theorem}
\newtheorem{proposition}{Proposition}
\newtheorem{lemma}{Lemma}
\newtheorem{remark}{Remark}

\newtheorem{condition}{Condition}

\newcounter{rmk}

\def\T{{ \mathrm{\scriptscriptstyle T} }}
\def\t{ \mathrm{\scriptscriptstyle T} }
\def\pr{P}
\def\Op{O_p}
\def\op{o_p}
\def\maha{\textnormal{Ma}}

\def\sumi{\sum_{i=1}^{n}}
\def\summ{\sum_{m=1}^{M}}
\def\sumim{\sum_{i \in [m]}}

\def\argmin{\mathop{\arg\min}}
\def\mean{\textnormal{mean}}
\def\cov{\textnormal{cov}}
\def\var{\textnormal{var}}
\newcommand\sign[1]{\textnormal{sign}{(#1)}}

\def\unadj{\textnormal{unadj}}
\def\ols{\textnormal{ols}}
\def\lasso{\textnormal{lasso}}
\def\proj{\textnormal{proj}}

\def\unadjM{{\textnormal{unadj}|\mathcal{M}_a}}

\newcommand\bb[1]{\boldsymbol #1}

\def\nt{n_1}
\def\nc{n_0}
\def\nz{n_z}
\def\nm{n_{[m]}}
\def\nmt{n_{[m]1}}
\def\nmc{n_{[m]0}}
\def\nmz{n_{[m]z}}
\newcommand\nma[1]{n_{[m]#1}}
\def\pim{\pi_{[m]}}

\def\emt{e_{[m]}}
\def\emz{e_{[m]z}}
\def\propc{c}

\def\ec{e_c}

\newcommand\wnz[1]{\omega_{i}(#1)}
\newcommand\wyz[1]{\omega_{i,Y}(#1)}
\newcommand\wxz[1]{\omega_{i,\xm}(#1)}

\def\H{H}
\def\Q{Q}

\def\xe{x} % element of X
\def\xr{\boldsymbol x} % row of X
\def\xc{\boldsymbol X} % column of X
\def\xm{\boldsymbol X} % X as a matrix

\def\wm{\boldsymbol X_{\mathcal{K}}} % W as a matrix

\def\z{\boldsymbol z}
\def\Z{\boldsymbol Z}
\def\e{\varepsilon}

\def\bx{\bar{\xr}}

\def\bxm{\bar{\xr}_{[m]}}
\def\bxmj{\bar{\xe}_{[m],j}}
\newcommand\bym[1]{\bar{Y}_{[m]}(#1)}

\newcommand\bxmz[1]{\bar{\xr}_{[m]#1}}
\def\bxmztj{\bar{\xe}_{[m]1,j}}
\newcommand\bymz[1]{\bar{Y}_{[m]#1}}

\def\Smx{S^2_{[m]\xm}}
\def\Smxs{S^2_{[m] \xm_{\mathcal{S}}}}
\def\Smw{S^2_{[m] \xm_{\mathcal{K}}}}

\newcommand\Smesz[1]{S^2_{[m]\e^*(#1)}}
\def\Smestmesc{S^2_{[m]\{\e^*(1)-\e^*(0)\}}}
\newcommand\Smyz[1]{S^2_{[m]Y(#1)}}
\def\Smytmyc{S^2_{[m]\{Y(1)-Y(0)\}}}

\newcommand\Smxyz[1]{S_{[m]\xm Y(#1)}}
\newcommand\Smxsyz[1]{S_{[m]\xm_{\mathcal{S}}Y(#1)}}

\newcommand\Smxsesz[1]{S_{[m]\xm_{\mathcal{S}}\e^*(#1)}}

\newcommand\Smwesz[1]{S_{[m] \wm \e^*(#1)}}

\def\smx{s^2_{[m] \xm }}

\newcommand\smxesz[1]{s_{[m] \xm \e^*(#1)}}

\def\smxz{s^2_{[m]\xm(z)}}
\newcommand\smxyz[1]{s_{[m] \xm Y(#1)}}
\def\tsmxz{\tilde{s}^2_{[m]\xm(z)}}
\newcommand\tsmxyz[1]{\tilde{s}_{[m] \xm Y(#1)}}

\def\Sx{\Sigma_{\xm \xm}}

\def\Ryw{R^2_{Y, \wm}}
\def\hRyw{\hat R^2_{Y, \wm}}

\def\Resw{R^2_{\e^*, \wm}}

\def\ti{\tau_i}
\def\tm{\tau_{[m]}}
\def\htu{\hat\tau_\unadj}
\def\htmu{\hat\tau_{[m]}}

\def\ttp{\check{\tau}_{\proj}}
\def\htlt{\hat\tau_\lasso}
\def\htp{\hat\tau_{\proj}}

\def\taux{\tau_{\xr}}
\def\htaux{\hat\tau_{\xr}}

\def\htauxs{\hat\tau_{\xr,{\mathcal{S}}}}

\def\htauw{\hat\tau_{\xr, \mathcal{K}}}

\def\su{\sigma_\unadj}
\def\hsu{\hat{\sigma}_\unadj}
\def\suM{\sigma_\unadjM}
\def\hsuM{\hat{\sigma}_\unadjM}

\def\slt{\sigma_\lasso}
\def\hslt{\hat{\sigma}_\lasso}

\def\bbeta{\boldsymbol\beta}
\newcommand\bz[1]{\bbeta(#1)}

\def\bgamma{\boldsymbol\gamma}
\def\g{\bgamma_\proj}
\def\gs{(\bgamma_\proj)_\mathcal{S}}
\def\gsc{(\bgamma_\proj)_{\mathcal{S}^c}}
\def\hg{\hat\bgamma_\lasso}
\newcommand\gz[1]{\bgamma(#1)}
\newcommand\gzs[1]{\{\bgamma(#1)\}_{\mathcal{S}}}
\newcommand\gzsc[1]{\{\bgamma(#1)\}_{\mathcal{S}^c}}
\newcommand\hgz[1]{\hat\bgamma_{\lasso,#1}}

\newcommand\hgzs[1]{\{\hat\bgamma_{\lasso,#1}\}_\mathcal{S}}
\newcommand\hgzsc[1]{\{\hat\bgamma_{\lasso,#1}\}_{\mathcal{S}^c}}

\def\hgoz{\hat\bgamma_{\ols,z}}

\def\hgozt{\hat\bgamma_{\ols,1}}
\def\hgozc{\hat\bgamma_{\ols,0}}
\def\hgo{\hat\bgamma_{\ols}}

\def\smean{\sigma_{\mean}}
\def\scov{\sigma_{\cov}}

\def\smcov{\sigma_{[m]\cov}}
\def\sa{\sigma_a^2}
\def\fma{\kappa_a^4}
\def\fmb{\kappa_b^4}
\def\subg{\kappa}
\def\h{\boldsymbol h}

\def\tY{R}

\def\setXwe{\mathcal{L}^\omega_{\xm\varepsilon}}
\def\setXwXw{\mathcal{L}^\omega_{\xm\xm}}

\def\ew{\varepsilon^{\omega}}

\newcommand\bewz[1]{\bar{\varepsilon}^{\omega}_{#1}}

\def\xwe{x^{\omega}}
\def\xwr{\boldsymbol x^{\omega}}
\newcommand\bxwz[1]{\bar{\boldsymbol x}^{\omega}_{#1}}

\def\yw{Y^{\omega}}
\newcommand\bywz[1]{\bar{Y}^{\omega}_{#1}}

\newcommand\Smxwewz[1]{S_{[m]\xm^{\omega}(#1) \ew(#1)}}
\def\Sxw{V_{\xm\xm}}
\def\hSxw{\hat{\Sigma}^\omega_{\xm\xm}}
\def\Sxewtw{\tilde \Sigma^\omega_{\xm\ew(1)}}
\def\hSxewtw{\hat{\Sigma}^\omega_{ \xm \ew(1)}}

\newcommand\zhu[1]{\textcolor{black}{#1}}

\newcommand\zk[1]{\textcolor{black}{#1}}

\begin{document}

\def\spacingset#1{\renewcommand{\baselinestretch}%
{#1}\small\normalsize} \spacingset{1}

%%%%%%%%%%%%%%%%%%%%%%%%%%%%%%%%%%%%%%%%%%%%%%%%%%%%%%%%%%%%%%%%%%%%%%%%%%%%%%

\if0\blind
{
\title{\bf Design-based theory for Lasso adjustment in randomized block experiments and rerandomized experiments}
\author{Ke Zhu, Hanzhong Liu\thanks{
Co-corresponding author: lhz2016@tsinghua.edu.cn.
}\hspace{.2cm}\\
Center for Statistical Science, Department of Industrial Engineering,\\ Tsinghua University, Beijing, China.\\
\hspace{.1cm}\\
Yuehan Yang\thanks{Corresponding author: yyh@cufe.edu.cn.
} \hspace{.2cm}\\
School of Statistics and Mathematics, \\Central University of Finance and Economics, Beijing, China.}
\date{}
\maketitle
} \fi

\if1\blind
{
\title{\bf Design-based theory for Lasso adjustment in randomized block experiments and rerandomized experiments}
\author{}
\date{}
\maketitle
} \fi

\bigskip
\begin{abstract}
Blocking, a special case of rerandomization, is routinely implemented in the design stage of randomized experiments to balance the baseline covariates. This study proposes a regression adjustment method based on the least absolute shrinkage and selection operator (Lasso) to efficiently estimate the average treatment effect in randomized block experiments with high-dimensional covariates. We derive the asymptotic properties of the proposed estimator and outline the conditions under which this estimator is more efficient than the unadjusted one. We provide a conservative variance estimator to facilitate valid inferences. Our framework allows one treated or control unit in some blocks and heterogeneous propensity scores across blocks, thus including paired experiments and finely stratified experiments as special cases. We further accommodate rerandomized experiments and a combination of blocking and rerandomization. Moreover, our analysis allows both the number of blocks and block sizes to tend to infinity, as well as heterogeneous treatment effects across blocks without assuming a true outcome data-generating model. Simulation studies and two real-data analyses demonstrate the advantages of the proposed method.
\end{abstract}

\noindent%

\vfill
\vfill
%\newpage

\spacingset{1.8} % DON'T change the spacing!

\section{Introduction}

% covariate imbalance, design & analysis
Randomized experiments are the basis for evaluating the effect of a treatment on an outcome and are widely used in the industry, social sciences, and biomedical sciences \citep{Fisher1935,box2005,Imbens2015,Rosenberger2015}. In randomized experiments, complete randomization of treatment assignments can balance the baseline covariates on average. However, covariate imbalances often occur in a particular treatment assignment \citep[see, for example,][]{Fisher1926, Morgan2012,Athey2017}. To increase the estimation efficiency of the treatment effect, some researchers have recommended balancing the key covariates in the design stage \citep{Fisher1926,Morgan2012,Krieger2019}, whereas others have emphasized the implementation of adjustments for covariate imbalances in the analysis stage \citep{Fisher1935,miratrix2013,lin2013}.

% blocking & rerandomization
\citet{Fisher1926} was the first to recommend the use of blocking or stratification in the design stage to balance a few discrete covariates that were the most predictive to the outcomes. Since then, blocking has been widely used in experimental designs \citep{Imai2008,Higgins2015,Schochet2016,Pashley2017,tabord2018stratification}. While blocking can effectively balance discrete covariates, balancing continuous covariates using this approach is less intuitive. Rerandomization is a more general approach for balancing continuous covariates \citep{student:1938,Morgan2012,Li2018,Li2020factorial}. Recently, scholars have recommended combining blocking and rerandomization techniques \citep{Schultzberg2019,Wang2020rerandomization}. Rubin summarized this design strategy as ``Block what you can and rerandomize what you cannot.''

% high-dimensional covariates
Blocking, rerandomization, or their combination can balance only a fixed number of covariates. However, in modern randomized experiments, a large number of baseline covariates are often collected and the number of covariates can be larger than the sample size. For example, in a randomized controlled trial, the researcher may record the demographic and genetic information of each participant. \citet{bloniarz2015lasso} highlighted that, in such high-dimensional settings, most of the covariates may not be related to the outcomes; thus, the important covariates must be selected to achieve efficient treatment effect estimation. In the design stage, when pre-experimental data (outcomes and covariates) were available, \citet{johansson2020} used Lasso \citep{Tibshirani1996} to select important covariates and rerandomization to balance the selected covariates. However, when pre-experimental outcome information is unavailable, it is difficult to perform covariate selection in the design stage. A more realistic approach is to use Lasso in the analysis stage to perform a regression adjustment. Regression adjustment has been widely used to analyze randomized experiments and increase associated efficiency \citep{Fisher1935,miratrix2013,lin2013,bloniarz2015lasso,Bugni2018,Bugni2019,Liu2019,lei2020,Wang2019model,Ma2020Regression,Liu2020,jiang2022regression,jiang2022noncomp}.

% blocking + regression, gap
Under a finite-population or design-based framework by conditioning on the potential outcomes and covariates, with treatment assignments being the only source of randomness, \citet{Liu2019} proposed a weighted regression adjustment method for randomized block experiments and demonstrated that this approach could increase the efficiency even when the number of blocks tends to infinity with the block sizes being fixed. However, this method has several limitations. First, it works only for homogeneous propensity scores (proportion of treated units in each block) across blocks and manages only low-dimensional covariates. 
\zk{However, propensity scores may be heterogeneous across blocks due to practical restrictions, as addressed by \citet{liu2021randomization} in the context of randomized block factorial experiments with low-dimensional covariates.}
Second, each block must have at least two treated and two control units, which may be unrealistic in many randomized experiments such as paired experiments and finely stratified experiments or observational studies \citep{fogarty2018finely,Pashley2017,bai2021inference}. For example, in observational studies, full matching is widely used to balance covariates and create \textit{fine blocks} with only one treated unit or only one control unit  \citep{rosenbaum1991characterization,hansen2004full}. After full matching, we can analyze the data as if they come from a finely stratified experiment \citep{bind2019}. Investigations of regression adjustment methods in finely stratified experiments are limited. Third, it does not consider rerandomization in the design stage. \citet{Li2020} established a unified theory for rerandomization followed by regression adjustment under complete randomization, but did not consider high-dimensional covariates or the combination of blocking and rerandomization in the design stage. To fill the gap, we develop general approaches and theoretical guarantees for the combination of blocking, rerandomization, and regression adjustment with high-dimensional covariates, homogeneous or heterogeneous block sizes, propensity scores, and treatment effects. Our methods do not require at least two treated and two control units in each block and thus are applicable to paired experiments, finely stratified experiments, and observational studies using full matching.

% contribution
Specifically, we propose a Lasso-adjusted average treatment effect (ATE) estimator for randomized block experiments from a projection perspective. We show that under mild conditions, the proposed estimator is consistent, asymptotically normal, and more efficient than, or at least as efficient as, the classic weighted difference-in-means estimator, even when the propensity scores differ across blocks or when only one treated or control unit exists in some blocks. We consider a general asymptotic regime in which both the total number of units and the number of covariates tend to infinity, allowing for a few large blocks, many small blocks, or a combination thereof. Our results extend the design-based theory of Lasso adjustment from completely randomized experiments \citep{bloniarz2015lasso} to randomized block experiments. As a by-product, we establish novel concentration inequalities for the weighted sample mean and covariance under stratified randomization with possibly heterogeneous block sizes and propensity scores. Moreover, we propose a Neyman-type conservative variance estimator to facilitate valid inferences. Finally, we investigate the asymptotic properties of the proposed Lasso-adjusted estimator under stratified rerandomization \citep{Wang2020rerandomization}, which combines stratification and rerandomization during the design stage. We outline the conditions under which it is more efficient than an unadjusted estimator. As another by-product, we extend the results of \citet{Li2020} to high-dimensional settings. Our asymptotic results were obtained under a finite population and randomization-based inference framework. Under this framework, the potential outcomes and covariates are fixed quantities, and the only source of randomness is the treatment assignment. Our theory allows the outcome data-generating model to be misspecified.

The remainder of this paper is organized as follows. Section \ref{sec:frame} introduces the framework and notation.
Section \ref{sec:b+r+l} describes a novel ``blocking + rerandomization + Lasso adjustment'' method and establishes its asymptotic theory.
The details of the simulation studies and two real data analyses are presented in Sections \ref{sec:sim} and \ref{sec:rd}, respectively. Section \ref{sec:dis} presents concluding remarks. All proofs are presented in the Supplementary Material.

\section{Framework and notation}
\label{sec:frame}

\zk{In Section \ref{sec:2.1}, we introduce randomized block experiments and review the existing inference results within the design-based causal inference framework. Section \ref{sec:2.2} presents stratified rerandomization, a more general experimental design for balancing covariates in the design stage, while Section \ref{sec:2.3} introduces the notations.}

\subsection{Randomized block experiments}
\label{sec:2.1}

% blocking
Blocking is a traditional approach to balancing discrete covariates in an experimental design. Experiments in which blocking is implemented are known as stratified randomized or randomized block experiments. Blocking can increase the estimation efficiency of the average treatment effect when blocking variables are predictive to the outcomes \citep{Fisher1926,imai2008variance,Imbens2015}.
Consider a randomized block experiment with $n$ units and a binary treatment. For each unit $i$, $i=1,\dots,n$, let $Z_i$ be an indicator of treatment assignment. Specifically, $Z_i = 1$ if unit $i$ is assigned to the treatment group, and $Z_i = 0$ otherwise.
Before the physical implementation of the experiment, we stratify the units into $M$ blocks. Let $B_i$ denote the block indicator of unit $i$. Let $\nm = \sumi I(B_i = m)$ denote the number of units in the block $m$ ($m=1,\dots,M$), where $I(\cdot)$ is an indicator function. Hereafter, subscript ``$[m]$'' indicates block-specific quantities. Let $\pim = \nm / n$ denote the proportion of block size for block $m$. Within block $m$, $\nmt$ units are randomly assigned to the treatment group, and the remaining $\nmc$ units are assigned to the control group. The total number of treated units is $\nt = \sum^M_{m = 1} \nmt$. We assume that $1 \leqslant \nmt \leqslant \nm - 1$, which includes both finely- and coarsely-stratified experiments (a \textit{fine block} has one treated or one control unit, whereas a \textit{coarse block} has at least two treated and two control units). In particular, pair experiments are special cases in our study. Let $\emt = \nmt / \nm$ denote the propensity score, which may differ across blocks for practical reasons, such as budget restrictions. The treatment assignments are independent across blocks, and thus the probability distribution of $ \Z = (Z_1,\dots,Z_n)^\T $ in randomized block experiments is $\pr ( \Z =  \boldsymbol z ) = \prod^M_{m=1} { (\nma{1}! \nma{0}! ) } / { \nm ! }, \ \sum_{i \in [m]}I(z_i = 1) = \nmt , \ z_i=0,1$, where $i \in [m]$ indexes unit $i$ in block $m$.

% estimand, etimator, inference
We define treatment effects using the Neyman--Rubin potential outcomes framework \citep{Neyman:1923,Rubin:1974}. For unit $i$, let $Y_i(1)$ and $Y_i(0)$ be the potential outcomes under treatment and control, respectively. We define the unit-level treatment effect as $\ti = Y_i(1) - Y_i(0)$. As each unit is assigned to either the treatment or control group, but not to both, we cannot simultaneously observe $Y_i(1)$ and $Y_i(0)$. Thus, $\ti$ cannot be identified without strong modeling assumptions pertaining to potential outcomes. Under the stable unit treatment value assumption (SUTVA) \citep{Rubin:1980}, the average treatment effect is identifiable.
%and is defined as $\tau = (1/n) \sumi \ti = (1/n) \sumi \big\{  Y_i(1) - Y_i(0) \big\}$.
Specifically, the block-specific average treatment effect is defined as
$\tm = ({1}/{\nm}) \sumim  Y_i(1)  -({1}/{\nm}) \sumim Y_i(0)  = \bym{1} - \bym{0}$
and the overall average treatment effect is defined as $$\tau = n^{-1} \sumi \big\{ Y_i(1) - Y_i(0) \big\} =  \summ \pim \tm.$$
Under stratified randomization, the observed outcome is $Y_i = Z_i Y_i(1) + ( 1 - Z_i ) Y_i(0)$.
An unbiased estimator of $\tm$ is the difference-in-means of the outcomes within block $m$, $\htmu = (1/\nmt) \sumim Z_i Y_i - (1/\nmc) \sumim (1 - Z_i) Y_i  = \bymz{1} - \bymz{0}$. Thus, the plug-in estimator of $\tau$ is the weighted difference-in-means, as follows: 
$$
\htu = \summ \pim \htmu = \summ \pim ( \bymz{1} - \bymz{0} ).
$$ 
Under mild conditions, $\htu$ is unbiased and $\sqrt{n}(\htu-\tau) \stackrel{d}\longrightarrow N(0, \su^2)$, where $\su^2$ is defined in Section A.3 in the Supplementary Material.

% covariate
However, the unadjusted estimator $\htu$ does not incorporate the covariate information beyond blocking and thus may lose efficiency.
For each unit $i$, consider a $p$-dimensional baseline/pre-treatment covariate vector $\xr_i = (\xe_{i1},\dots,\xe_{ip})^\T \in \mathbb{R}^{p}$, where $p$ is comparable to or even larger than $n$.
The covariates in $\xr_i$ can be continuous or categorical, but we assume that they cannot be represented by a linear combination of $I(B_i=m)$, $m=1,\ldots,M$. We denote $\xm=(\xr_{1},\dots,\xr_{n})^\T$.
In order to enhance both estimation and inference efficiency, we could incorporate these covariates during both the design and analysis stages.

\zk{According to prior experiments or domain knowledge, covariates highly predictive of the potential outcomes may be preferentially balanced in the design stage. Specifically, if we have pilot experiment data that includes outcomes and covariates, we can run a Lasso regression. The covariates selected by Lasso are considered predictive and should be balanced in the subsequent experimental design \citep{johansson2020}. If we do not have pilot experiment data, we need to rely on domain knowledge to choose predictive covariates. Even if the prior domain knowledge is incorrect, balancing these covariates won't harm the validity and efficiency of the statistical inference \citep{Wang2020rerandomization}.}
For instance, consider a study interested in evaluating the effect of academic achievement awards on the academic performance of college students. The outcome of interest is the grade point average (GPA) in the current semester.
Based on domain knowledge, the GPA from the previous year is highly predictive of the GPA in the current semester, thus it is better to be balanced in the design stage.
We denote these highly predictive covariates as $\xr_{i,\mathcal{K}}=(x_{ij}, j \in \mathcal{K})^\T \in \mathbb{R}^{k}$, where $\mathcal{K}$ is the index set and $k$ is the dimension of these covariates.
We denote $\xm_{\mathcal{K}}=(\xr_{1,\mathcal{K}},\dots,\xr_{n,\mathcal{K}})^\T$.
For the remaining $p-k$ covariates, when the designer does not have prior information regarding their importance, we can perform data-driven variable selection and regression adjustment during the analysis stage. Throughout this study, we assume that $k$ is fixed regardless of $n$ while $p$ diverges with $n$. For notation simplicity, we do not index $p$ with $n$. The objective is to make valid and efficient inferences on the average treatment effect $\tau$ using the observed data $\{Y_i, Z_i, \xr_i\}_{i=1}^{n}$.

\subsection{Stratified rerandomization (blocking plus rerandomization)}
\label{sec:2.2}

% blocking + rerandomization
In this section, we introduce stratified rerandomization, which could further balance low-dimensional covariates $\xm_{\mathcal{K}}$ beyond blocking in the design stage.
Although blocking is widely used in practice, it can only balance discrete covariates. Rerandomization is a more general approach for balancing both discrete and continuous covariates \citep{Morgan2012}. Rerandomization discards the treatment assignments that lead to covariate imbalances and accepts only those assignments that fulfill a pre-specified balance criterion.
Scholars have recommended combining blocking and rerandomization in the design stage \citep{Schultzberg2019,Wang2020rerandomization}. In particular, \citet{Wang2020rerandomization} proposed a stratified rerandomization strategy based on the Mahalanobis distance. Specifically, the weighted difference-in-means of covariates $\xm_{\mathcal{K}}$ is defined as $\htauw = \summ \pim \{ (\bar{\xr}_{[m]1})_\mathcal{K} - (\bar{\xr}_{[m]0})_\mathcal{K} \}$, where $(\bar{\xr}_{[m]z})_\mathcal{K}=\nmz^{-1} \sumim  I(Z_i = z) \xr_{\mathcal{K}}$, $z=0,1$. The Mahalanobis distance is defined as $\maha(\Z, \wm) = ( \htauw )^\T \{\cov( \htauw )\}^{-1} \htauw$. A treatment assignment is acceptable if and only if the corresponding Mahalanobis distance is less than or equal to a pre-specified threshold $a>0$; that is, $\maha(\Z, \wm) \leqslant a$. We denote $\mathcal{M}_a =\{\Z :  \maha(\Z, \wm ) \leqslant a \} $ the set of acceptable treatment assignments. \citet{Li2018} suggested choosing a suitable $a$ to ensure that the probability of a treatment assignment satisfying the balance criterion equals a certain value, for example, $p_a = \pr\{\maha( \Z, \wm) \leqslant a\}=0.001$. The procedure is presented in Algorithm \ref{alg::rr}.

\citet{Wang2020rerandomization} showed that the asymptotic distribution of $\htu$ under stratified rerandomization is a convolution of a normal distribution and a truncated normal distribution, and its asymptotic variance, denoted as $\sigma^2_{\unadjM}$, is less than or equal to that of $\htu$ in the case of stratified randomization. Moreover, the asymptotic variance can be estimated using a conservative estimator $\hat \sigma^2_{\unadjM}$, as indicated in the Supplementary Material. These conclusions hold for cases involving equal or unequal propensity scores. In this study, we propose a Lasso-based method for adjusting the high-dimensional covariates $\xm$ in the analysis stage and establish the asymptotic theory under stratified rerandomization.

\begin{algorithm}[t]
\caption{Stratified rerandomization using the Mahalanobis distance}
\label{alg::rr}
\begin{algorithmic}
    \STATE {1. Collect covariate data $\xm_{\mathcal{K}}$. }
    \STATE {2. (Re-)Randomize units into the treatment and control groups by stratified randomization and obtain the treatment assignment vector $\Z$. }
    \STATE {3. If $\maha(\Z, \wm) \leqslant a$, proceed to Step 4; otherwise, return to Step 2;}
    \STATE {4. Conduct the physical experiment using treatment assignment $\Z$.}
\end{algorithmic}
\end{algorithm}

\subsection{Notation}
\label{sec:2.3}

% notation
To facilitate the discussion, we use the following notation: For an $L$-dimensional column vector $\boldsymbol u = (u_1,\dots, u_L)^\T$, let $\| \boldsymbol u \|_0$, $\| \boldsymbol u \|_1$,  $\| \boldsymbol u \|_2$, and $\|  \boldsymbol u \|_\infty$ denote the $\ell_0$, $\ell_1$, $\ell_2$ and $\ell_\infty$ norms, respectively. For a subset $\mathcal{S} \subset \{1,\dots,L\}$, $\mathcal{S}^c$ is the complementary set of $\mathcal{S}$, and $\boldsymbol u_{\mathcal{S}} = (u_j, j \in \mathcal{S})^\T$ is the restriction of $\boldsymbol u$ on $\mathcal{S}$. Let $|\mathcal{S}|$ be the cardinality of $\mathcal{S}$. For matrix $A$, $\Lambda_{\max}(A)$ indicates the largest eigenvalue of $A$. Let $ \stackrel{d}\longrightarrow$ and $\stackrel{p}\longrightarrow$ denote the convergence in distribution and in probability, respectively. \zk{We use $c$ and $C$} to denote universal constants that do not change with $n$ but whose precise value may change from line to line.

For finite population quantities $\H = (\H_1,\dots,\H_n)^\t$ and $\Q = (\Q_1,\dots, \Q_n)^\t$, where $\H_i$ and $\Q_i$ could be potential outcomes (scalars), adjusted potential outcomes (scalars), or covariates (column vectors), the following notations are used. The block-specific finite population mean and sample mean are defined as $\bar{\H}_{[m]} = \nm^{-1} \sumim \H_i$ and $\bar{\H}_{[m]z} = \nmz^{-1} \sumim  I(Z_i = z) \H_i$, respectively.
The block-specific finite population covariance and sample covariance are defined as $S_{[m] \H \Q} = ( \nm - 1 )^{-1} \sumim ( \H_i - \bar{\H}_{[m]} ) ( \Q_i - \bar{\Q}_{[m]} )^\T$ and $s_{[m]\H \Q}=( \nmz - 1 )^{-1} \sumim I(Z_i = z) ( \H_i - \bar{\H}_{[m]z} ) ( \Q_i - \bar{\Q}_{[m]z} )^\T$. When $\H=\Q$, the subscripts are simplified to $S^2_{[m]\H} = S_{[m]\H \H}$ and $s^2_{[m]\H}=s_{[m]\H \H}$.
While these quantities do depend on $n$, they are not indexed with $n$ to maintain the simplicity of the notation.
For example, if $\H_i=Y_i(1)$ and $\Q_i=\xr_i$, we have $\bar{Y}_{[m]}(1) = \nm^{-1} \sumim Y_i(1)$, $\bar{Y}_{[m]1} = \nmt^{-1} \sumim  I(Z_i = 1) Y_i(1)$, $S_{[m] Y(1)\xm} = ( \nm - 1 )^{-1} \sumim \{ Y_i(1) - \bar{Y}_{[m]}(1) \} \{ \xr_i - \bar{\xr}_{[m]} \}^\T$, and $s_{[m] Y(1)\xm}=( \nmt - 1 )^{-1} \sumim I(Z_i = 1) \{  Y_i(1) - \bar{Y}_{[m]1} \} \{ \xr_i - \bar{\xr}_{[m]1} \}^\T$.
We summarize notations in two tables and include them in the Supplementary Material.

\section{Blocking, rerandomization, and Lasso adjustment}
\label{sec:b+r+l}

\zk{
In Section \ref{sec:3.1}, we propose a regression adjustment estimator with low-dimensional covariates ($p \ll n$) based on a projection perspective in randomized block experiments. In Section \ref{sec:Lasso-unequal}, we propose a Lasso-adjusted ATE estimator to handle high-dimensional covariates and accommodate stratified rerandomization. We also propose a variance estimator for the Lasso-adjusted estimator. Theoretical results, including asymptotic normality, validity of the variance estimator, and efficiency gain, are provided in Section \ref{sec:3.3}.
}

\subsection{Regression adjustment from a projection perspective}
\label{sec:3.1}

% projection
As the baseline covariates are not affected by the treatment, the average treatment effect of the covariates is $\taux = \summ \pim \{\bar{\xr}_{[m]} - \bar{\xr}_{[m]} \} = \bb{0}$.
Let $\htaux = \summ \pim ( \bxmz{1} - \bxmz{0} )$ be the weighted difference-in-means estimator for $\taux$.
To decrease the variance of $\htu$, we can project it onto $\htaux$. We define the projection coefficient vector $\g$ as follows:
$$
\begin{aligned}
\g =& \argmin_{\bgamma} E ( \htu - \tau -  \htaux^\T \bgamma )^2 = \cov( \htaux )^{-1} \cov( \htaux, \htu)  \\
= & \bigg\{ \summ \pim \frac{\Smx}{\emt ( 1 - \emt) } \bigg\}^{-1} \bigg\{ \summ \pim \frac{\Smxyz{1}}{\emt } + \summ \pim \frac{\Smxyz{0}}{ 1 - \emt } \bigg\}.
\end{aligned}
$$
The oracle projection estimator $\ttp = \htu - \htaux^\T \g $ is consistent, asymptotically normal, and has the smallest asymptotic variance among the estimators that are asymptotically equivalent to $\htu - \htaux^\T \bgamma$ for some adjusted vector $\bgamma\in\mathbb{R}^p$.

% decomposition
However, $\ttp$ is not feasible in practice, because it depends on the unknown vector $\g$.
To consistently estimate $\g$, we decompose it into two terms: Let $\emz=z \emt + (1 - z )( 1 - \emt)$, $z=0,1$. Then, $\g=\gz{0}+\gz{1}$ with
$$
\gz{z} = \bigg\{ \summ \pim \frac{\Smx}{\emt ( 1 - \emt) } \bigg\}^{-1} \bigg\{ \summ \pim \frac{\Smxyz{z}}{\emz }  \bigg\}, \quad z=0,1.
$$

% plug-in
Intuitively, when $\nmz \geq 2$, we can estimate $\gz{z}$ by replacing the block-specific covariances $\Smx$ and $\Smxyz{z}$ with the corresponding sample covariances $\smxz$ and $\smxyz{z}$ within block $m$ for treatment arm $z$. When $\nmz=1$ for some $m$ and $z$, both $\smxz$ and $\smxyz{z}$ are not well-defined. To address this issue, we can use the following estimators:
\begin{gather*}
\tsmxz=\frac{\nm}{\nmz(\nm-1)} \sumim I(Z_i=z) \left(\xr_{i}-\bxm\right)\left(\xr_{i}-\bxm\right)^\T,\\
\tsmxyz{z}=\frac{\nm}{\nmz(\nm-1)} \sumim I(Z_i=z) \left( \xr_{i}-\bxm\right) Y_{i},
\end{gather*}
that are unbiased and well-defined when $1\leqslant\nmz\leqslant\nm-1$.
Then, we can estimate $\gz{z}$ by the following plug-in estimator,
\begin{equation*}
\begin{aligned}
\hgoz &= \bigg\{ \summ \pim \frac{\tsmxz}{\emz ( 1 - \emz) } \bigg\}^{-1} \bigg\{ \summ \pim \frac{\tsmxyz{z}}{  \emz }  \bigg\}\\
&=\argmin_{\bgamma}
\summ \sum_{i \in [m], Z_i = z}  \frac{\nm^2}{\emz\nmz(\nm-1)}  \bigg \{ \sqrt{1 - \emz}\, Y_i - \frac{1}{\sqrt{1 - \emz}}\, ( \xr_{i} - \bar{\xr}_{[m]})^\T \bgamma \bigg \}^2.
\end{aligned}
\end{equation*}

% weight
To simplify the above expression, we introduce several weights. For $i\in [m]$, let $\wnz{z} = \nm^2/\{\emz\nmz(\nm-1)\}$, $\wyz{z} = 1 - \emz$, and $\wxz{z} = 1/(1 - \emz)$. The weighted potential outcomes and covariates are denoted by $\yw_i(z)=\sqrt{\wnz{z} \wyz{z}} \,Y_i(z)$ and $\xwr_{i}(z)=\sqrt{\wnz{z} \wxz{z}} \,( \xr_{i} - \bar{\xr}_{[m]} )$, respectively. Note that because the weights are different for $z=0,1$, we have two groups of weighted covariates $\xwr_{i}(z)$. Then, $\hgoz$ can be rewritten as
\begin{equation*}
\begin{aligned}
\hgoz =\argmin_{\bgamma}
\sum_{i:Z_i = z}  \Big \{  \yw_i - (\xwr_{i}- \bxwz{z})^\T \bgamma \Big \}^2,
\end{aligned}
\end{equation*}
where $\yw_i=Z_i\yw_i(1)+(1-Z_i)\yw_i(0)$, $\xwr_{i}=Z_i\xwr_i(1)+(1-Z_i)\xwr_i(0)$, and $\bxwz{z} =\nz^{-1}\sum_{i:Z_i= z} \xwr_{i}$.

\begin{remark}
The exact equivalent formula for $\hgoz$ is regressing $\yw_i$ on $\xwr_{i}$ for treatment arm $z$ without an intercept term, which has the same asymptotic distribution but worse finite sample performance. Therefore, we suggest using the with-intercept regression for estimating $\gz{z}$.
\end{remark}

Finally, we obtain a feasible estimator $\htp = \htu - \htaux^\T \hgo $, where $\hgo = \hgozt + \hgozc$. The feasible estimator $\htp$ achieves the same asymptotic efficiency as $\ttp$.
To the best of our knowledge, $\htp$ is the first regression-adjusted average treatment effect estimator that accommodates randomized block experiments involving both coarse and fine blocks. Meanwhile, $\htp$ is consistent with previous estimators in special cases. For instance, $\htp$ is exactly equivalent to the regression-adjusted estimator $\hat{\tau}_{R1}$ proposed by \citet{fogarty2018} for paired experiments. In the next section, we extend $\htp$ to a high-dimensional setting.

\subsection{Lasso-adjusted ATE estimator and variance estimator}
\label{sec:Lasso-unequal}

% add penalty
In this section, we consider the general framework that includes blocking, rerandomization, and high-dimensional covariates.
In a high-dimensional setting, if many covariates do not affect the potential outcomes, it is reasonable to assume that the projection coefficient $\g$ is sparse. More precisely, let $\mathcal{S} \in \{1,\dots,p\}$ be the set of relevant covariates that are predictive to the outcomes and let $s= |\mathcal{S}|$. We assume that $p \gg n$ but $s \ll n$. Both $\mathcal{S}$ and $s$ are unknown in practice.
Recall that in stratified rerandomization, we balance the low-dimensional covariates $\xm_{\mathcal{K}}$.
For simplicity, we assume that $\mathcal{K} \subset \mathcal{S}$. Otherwise, we consider $\mathcal{\tilde S} = \mathcal{S} \cup \mathcal{K}$ such that $\mathcal{K} \subset  \mathcal{\tilde S}$ and the following results holds with $\mathcal{S}$ replaced by $\mathcal{\tilde S}$. This also implies that we should include all covariates in $\mathcal{K}$ when conducting regression adjustment, as doing so is important to ensure the efficiency gain.

We still use $\g=\gz{0}+\gz{1}$ to denote the projection coefficient, where $\gzsc{z}= \boldsymbol 0$ and
$$
\gzs{z} = \bigg\{ \summ \pim \frac{\Smxs}{\emt ( 1 - \emt) } \bigg\}^{-1} \bigg\{ \summ \pim \frac{\Smxsyz{z}}{\emz }  \bigg\}, \quad z=0,1.
$$
We can estimate $\gz{z}$ using Lasso
$$
\begin{aligned}
\hgz{z} =  \argmin_{\bgamma}
\frac{1}{2 n_z}\sum_{i:Z_i = z}  \bigg \{  \yw_i - ( \xwr_{i} - \bxwz{z} )^\T \bgamma \bigg \}^2 + \lambda_z \| \bgamma \|_1,
\end{aligned}
$$
where $\lambda_z$ is the tuning parameter, $z=0,1$.
We replace $\g$ with its estimator $\hg=\hgz{1} + \hgz{0}$ and obtain the projection-originated Lasso-adjusted estimator, $\htlt = \htu -  \htaux^\T \hg$.
\zk{We summarize the implementation of $\htlt$ in Algorithm~\ref{alg:lasso}.}

\begin{algorithm}[t]
\caption{\zk{Lasso-adjusted ATE estimator}}
\label{alg:lasso}
\begin{algorithmic}
    \STATE \textbf{Input:} Outcome $Y_i$, treatment indicator $Z_i$, block indicator $B_i$, covariates $\xr_i$ (if rerandomization is performed for balancing $\xr_{i,\mathcal{K}}$, $\xr_{i,\mathcal{K}}$ should be included in $\xr_i$)
    \STATE \textbf{Output:} $\htlt$
    
    \STATE {\textbf{1. Transform outcome and covariates based on blocking scheme:}}
    \STATE  \hspace{\algorithmicindent}\textbf{for} $m=1,\ldots,M$ \textbf{do}
    \STATE \hspace{\algorithmicindent}\hspace{\algorithmicindent}$\nm = \sum_{i=1}^n I(B_i = m)$ \hfill // num. of units in block $m$
        \STATE  \hspace{\algorithmicindent}\hspace{\algorithmicindent}\textbf{for} $z=0,1$ \textbf{do}
        
        \STATE \hspace{\algorithmicindent}\hspace{\algorithmicindent}\hspace{\algorithmicindent} $\nmz = \sum_{i=1}^n I(B_i = m) I(Z_i = z)$ \hfill // num. of treated/control units in block $m$
        \STATE \hspace{\algorithmicindent}\hspace{\algorithmicindent}\hspace{\algorithmicindent} $\emz = \nmz / \nm$ \hfill // proportion of treated/control units in block $m$
\STATE  \hspace{\algorithmicindent}\hspace{\algorithmicindent}\textbf{end for}
\STATE  \hspace{\algorithmicindent}\textbf{end for}
    
\hspace{\algorithmicindent}\textbf{for} $i=1,\ldots,n$ \textbf{do}
    \STATE \hspace{\algorithmicindent}\hspace{\algorithmicindent} $\omega_i = \sum_{m=1}^M\sum_{z=0,1} [I(B_i = m) I(Z_i = z)\cdot \nm^2/\{\emz\nmz(\nm-1)\}]$
    \STATE \hspace{\algorithmicindent}\hspace{\algorithmicindent} $\omega_{i,Y} = 
    \sum_{m=1}^M\sum_{z=0,1} [I(B_i = m) I(Z_i = z)\cdot (1 - \emz)]$
    \STATE \hspace{\algorithmicindent}\hspace{\algorithmicindent} $\omega_{i,\xm} = 1 / \omega_{i,Y}$
    \STATE \hspace{\algorithmicindent}\hspace{\algorithmicindent} $\yw_i=\sqrt{\omega_i\, \omega_{i,Y}} \,Y_i$  \hfill // Transformed outcome
    \STATE \hspace{\algorithmicindent}\hspace{\algorithmicindent} $\xwr_{i}=\sqrt{\omega_i\, \omega_{i,\xm}} \,\big\{ \xr_{i} - \nm^{-1} \sum_{i=1}^n I(B_i = m)\cdot \xr_i \big\}$  \hfill // Transformed covariates
\STATE  \hspace{\algorithmicindent}\textbf{end for}

\STATE {\textbf{2. Fit Lasso in each treatment group:}}
  \STATE  \hspace{\algorithmicindent}\textbf{for} $z=0,1$ \textbf{do}

\STATE \hspace{\algorithmicindent}\hspace{\algorithmicindent} Fit Lasso with intercept based on data $\{\yw_i,\xwr_{i}\}_{i: Z_i=z}$
\STATE \hspace{\algorithmicindent}\hspace{\algorithmicindent} Determine tuning parameter $\lambda_z$ by cross-validation and obtain $\hgz{z}$

\STATE  \hspace{\algorithmicindent}\textbf{end for}
    
    \STATE {\textbf{3. Compute ATE estimator:}}
        \STATE  \hspace{\algorithmicindent} $\htu = \summ  (\nm/n)\big\{ \nmt^{-1} \sum_{i=1}^n I(B_i = m) Z_i Y_i - \nmc^{-1} \sum_{i=1}^n I(B_i = m)(1 - Z_i) Y_i \big\}$
        \STATE \hspace{\algorithmicindent} $\htaux=\summ  (\nm/n)\big\{ \nmt^{-1} \sum_{i=1}^n I(B_i = m) Z_i \xr_i - \nmc^{-1} \sum_{i=1}^n I(B_i = m)(1 - Z_i) \xr_i \big\}$
        \STATE \hspace{\algorithmicindent} $\htlt = \htu -  \htaux^\T (\hgz{0} + \hgz{1})$
\STATE {\textbf{Return}} $\htlt$
  
\end{algorithmic}
\end{algorithm}

% variance estimation
To facilitate valid inference, we propose a conservative variance estimator for $\htlt$.
To handle fine blocks with one treated or control unit, we extend the variance estimator proposed by \citet{Pashley2017}. Specifically, let $\mathcal{A}_c=\{1\leqslant m\leqslant M:\nmt > 1, \ \nmc > 1\}$ and $\mathcal{A}_f=\{1\leqslant m\leqslant M:\nmt= 1 \textnormal{ or } \nmc = 1\}$ denote sets of coarse and fine blocks, respectively. Let $n_f=\sum_{m\in\mathcal{A}_f}\nm$ be the total number of units in fine blocks. For $m\in\mathcal{A}_f$, we define a weight $\omega_{[m]}=\nm^2/(n_f-2\nm)$ and assume $n_f>2\nm$ throughout the paper. Define the adjusted outcomes $\tY_i=Y_i- (\xr_i-\bxm)^\T \hg$. Let $\hat{\tau}_{\tY,[m]}=\bar{\tY}_{[m]1}-\bar{\tY}_{[m]0}$ and $\hat{\tau}_{\tY,f}=\sum_{m\in \mathcal{A}_f} (\nm/n_f) \hat{\tau}_{\tY,[m]}$. Then, $\slt^2$ can be estimated by:
$$
\begin{aligned}
\hslt^2 = \frac{{n}}{{n} - \hat s} \Bigg[
&\sum_{m\in\mathcal{A}_c} \pim \Big\{ \frac{ s^2_{[m]\tY(1)}}{\emt}  + \frac{ s^2_{[m]\tY(0)}}{1 - \emt} \Big\}\\
&+\left(\frac{n_f}{n}\right)^2\frac{n}{n_f+\sum_{m\in\mathcal{A}_f}\omega_{[m]}}\sum_{m\in\mathcal{A}_f}\omega_{[m]}(\hat{\tau}_{\tY,[m]}-\hat{\tau}_{\tY,f})^2\Bigg],
\end{aligned}
$$
where $\hat s = ||\hg||_0$. The factor ${n}/\{ {n} - \hat s \}$ adjusts for the degrees of freedom of the residuals to achieve better finite sample performance. Define $\hsu^2$ similarly to $\hslt^2$ with $\tY_i=Y_i- (\xr_i-\bxm)^\T \hg$ replaced by $Y_i$.

\begin{remark}[\zk{Connection to \citet{Liu2019}}]
\label{rm:opt1}
The Lasso-adjusted estimator $\htlt$ is a generalization of the OLS-adjusted estimator proposed by \citet{Liu2019} to high-dimensional settings, but we use different weights to ensure an efficiency gain in cases of heterogeneous propensity scores across blocks or the presence of fine blocks, i.e., $\nmz = 1$ for some $m=1,\dots,M$ and $z=0,1$.
\end{remark}

\begin{remark}[\zk{Non-linear adjustments}]
\label{rm:nonlinear}
The proposed Lasso-adjusted estimator is motivated by a linear projection, which implicitly restricts the class of estimators to $\htu - \htaux^\T \bgamma$ where $ \bgamma \in \mathbb{R}^p$ and $\ \bgamma_{\mathcal{S}^c}= \boldsymbol 0$. It is possible to further improve efficiency by using non-linear adjustments. In practice, we can include non-linear transformations of the original covariates to $\xm$ before using the Lasso for performing non-linear adjustments. As a demonstration, we include quadratic terms of the continuous covariates and two-way interactions in the real data analysis; see Section \ref{sec:rd}.
\end{remark}

\subsection{Theoretical properties}
\label{sec:3.3}

% aysmptotic normality
To investigate the asymptotic properties of $\htlt$, we decompose the original potential outcomes and define the approximation error $\e^*_i(z)$ as follows:
\begin{equation*}
Y_i(z) =  \bar{Y}_{[m]}(z) +  (\xr_i - \bxm)^\T \g +  \e^*_i(z), \quad i \in [m], \quad z=0,1.
\end{equation*}
We require the following regularity conditions to guarantee the asymptotic normality of $\htlt$.

\begin{condition}\label{cond:propensity}
There exists a constant $\propc \in(0,0.5)$ independent of $n$ such that $\propc \leqslant \emt  \leqslant 1 - \propc$ for $m=1,\dots,M$.
\end{condition}

\begin{condition}\label{cond:moment-X}
There exists a constant $L < \infty$ independent of $n$ such that for $z=0,1$ and $j=1,\dots,p$,
$$
\frac{1}{n}\summ \sumim (\xe_{ij} - \bxmj )^4 \leqslant L,
\quad
\frac{1}{n}\summ \sumim \{  \e^*_i(z)  \}^4 \leqslant L,
$$
$$
\frac{1}{n}\summ \sumim [  \yw_i(z) - \bar{Y}^\omega(z)-\{\xwr_i(z)-\bar{\xr}^\omega(z)\}^\T \gz{z} ]^4 \leqslant L,
$$
where
$\bar{Y}^\omega(z)=n^{-1}\sum_{i=1}^n Y_i^\omega(z)$ and $\bar{\xr}^\omega(z)=n^{-1}\sum_{i=1}^n \xr_i^\omega(z)$.
\end{condition}

\begin{condition}\label{cond:limit-variance}
The weighted variances $\summ \pim \Smesz{1}/\emt$, $\summ \pim \Smesz{0}/(1 - \emt)$, and $\summ \pim \Smestmesc$ tend to finite limits, with positive values for the first two terms. The limit of $$\summ \pim \Smesz{1}/\emt + \summ \pim \Smesz{0}/(1 - \emt) - \summ \pim \Smestmesc $$ is strictly positive.
\end{condition}

\begin{condition}\label{cond:re}
There exist constants $C>0$ and $\xi > 1$ independent of $n$ such that $|| \h_\mathcal{S} ||_1 \leqslant C s || \Sxw \h ||_{\infty}$, $\forall \h \in \{\h:  || \h_{\mathcal{S}^c} ||_1 \leqslant \xi || \h_\mathcal{S} ||_1  \}$, where $\Sxw = \summ\pim\Smx/\{\emt(1-\emt)\}$.
\end{condition}

\begin{condition}\label{cond:tuning}
There exist constants $C>0$ and $0 < \eta < (\xi - 1) / (\xi + 1)$ such that the tuning parameters of Lasso satisfy
$$
s \sqrt{\log p} \lambda_z \rightarrow 0\quad \textnormal{and}\quad
\lambda_z\geqslant \frac{1}{\eta}\Big\{C\sqrt{\frac{\log p}{n}}+\delta_n\Big\}, \quad z=0,1,
$$
where $\delta_n = \max_{z=0,1} \| \summ (\pim - n^{-1}) \emz  S_{[m]\xm^{\omega}(z) \{ \yw(z) -\xm^{\omega}(z)  \gz{z} \} }  \|_{\infty}$.
\end{condition}

\begin{condition}\label{cond:rerandomization}
The weighted covariances $\summ \pim \Smw/\emt$, $\summ \pim \Smw/(1-\emt)$, $\summ \pim S_{[m]  \xm_{\mathcal{K}} \e^*(1)}/\emt$, and $\summ \pim S_{[m]  \xm_{\mathcal{K}} \e^*(0)}/(1-\emt)$ tend to finite limits, and the limit of $V_{\wm \wm}=\summ \pim \Smw / \{ \emt ( 1 - \emt) \}$ is strictly positive definite.
\end{condition}

\begin{remark}
Conditions~\ref{cond:propensity} and \ref{cond:limit-variance} are required for deriving the asymptotic normality of low-dimensional regression-adjusted treatment effect estimators \citep{Freedman2008,lin2013}. 
\zk{
Condition \ref{cond:moment-X} assumes the bounded fourth moment for both covariates and residuals, which is stronger than the maximum second moment condition in \citet{Liu2019}.
Conditions~\ref{cond:moment-X}, \ref{cond:re} and \ref{cond:tuning}} are similar to the typical conditions for deriving the $l_1$ convergence rate of Lasso under complete randomization, with treated units being sampled without replacement from the finite population \citep{bloniarz2015lasso}.
\end{remark}

\begin{remark}
When $s=o_p(\sqrt{n}/\log p)$, $\lambda_z \propto \sqrt{\log p/n}$, and $\delta_n=O_p(\sqrt{\log p/n})$ are satisfied, Condition~\ref{cond:tuning} holds. Sparsity scaling $s=o_p(\sqrt{n}/\log p)$ is also required for the de-biased Lasso in high-dimensional inference \citep{zhang2014confidence, van2014asymptotically, javanmard2014confidence}. The tuning parameter $\lambda_z$ is typically assumed to scale as $\sqrt{\log p/n}$ for the Lasso. A sufficient condition for $\delta_n=O_p(\sqrt{\log p/n})$ is the existence of $s$ transformed covariates that can perfectly predict the transformed potential outcomes. Specifically, perfect prediction means $
\yw_i(z) - \bar{Y}^\omega(z)-\{\xwr_i(z)-\bar{\xr}^\omega(z)\}^\T \gzs{z}=0$ for all $i$ and $z$, which implies $\delta_n=0$. In practice, perfect prediction is nearly impossible, so we allow $\delta_n$ to scale as $O_p(\sqrt{\log p/n})$.
\end{remark}

\begin{theorem}
\label{thm:srr2}
Suppose that Conditions~\ref{cond:propensity}--\ref{cond:rerandomization} hold. Under stratified rerandomization $\mathcal{M}_a$ with a fixed $a > 0$, $\sqrt{n} ( \htlt - \tau )$ is asymptotically normal, that is, $\{ \sqrt{n} ( \htlt - \tau ) \mid \mathcal{M}_a \} \stackrel{d}\longrightarrow N(0, \slt^2)$, where
$$
\slt^2 = \lim\limits_{n \rightarrow \infty} \summ \pim \Big\{ \frac{\Smesz{1}}{\emt} + \frac{\Smesz{0}}{1 - \emt} - \Smestmesc  \Big\}.
$$
Furthermore, $\htlt$ is asymptotically more efficient than $\htu$ under stratified rerandomization, as implied by $\slt^2 \leqslant \sigma^2_{\unadjM} \leqslant \sigma^2_{\unadj}$.

\end{theorem}

Theorem~\ref{thm:srr2} shows that the asymptotic distribution of $\htlt$ in the stratified rerandomization case is normal. Moreover, the asymptotic variance of $\htlt$ is no greater than that of $\htu$ in the stratified randomization and stratified rerandomization scenarios, even for cases involving unequal propensity scores or fine blocks. Thus, the efficiency achieved using $\htlt$ is never lower than that achieved using $\htu$.

% variance estimation

Subsequently, we outline the theoretical results regarding variance estimators within the general framework of stratified rerandomization.

\begin{condition}\label{cond:largest-RE}
There exists a constant  $C < \infty$ such that $\Lambda_{\max}( \Sxw ) \leqslant C$ and $n^{-1}\sumi\{ Y_i(z) -  (\xr_i - \bxm)^\T \g  \}^2\leqslant C$, $z=0,1$.
\end{condition}

\begin{theorem}\label{thm:srr2var}
Suppose that Conditions~\ref{cond:propensity}--\ref{cond:largest-RE} hold. Under stratified rerandomization $\mathcal{M}_a$ with a fixed $a>0$, $\hslt^2$ converges in probability to
$$
\slt^2 + \lim\limits_{n \rightarrow \infty}  \left\{\sum_{m\in\mathcal{A}_c} \pim\Smestmesc +
\left(\frac{n_f}{n}\right)^2\frac{n}{n_f+\sum_{m\in\mathcal{A}_f}\omega_{[m]}}\sum_{m\in\mathcal{A}_f}\omega_{[m]}({\tau}_{[m]}-{\tau}_{f})^2
\right\},
$$
which is no less than $\slt^2$.
Moreover, $\hslt^2 \leqslant \hat \sigma^2_{\unadjM} \leqslant \hat \sigma^2_{\unadj}$ holds in probability.
\end{theorem}

If the treatment effects within each coarse block are constant; that is, $\ti = \tm$ for all $i \in [m]$, then we have $\Smestmesc = 0$. If further the average treatment effects across fine blocks are constants, i.e., ${\tau}_{[m]} ={\tau}_{f}$ for all $m\in\mathcal{A}_f$, $\hslt^2$ is a consistent estimator of $\slt^2$. In general, $\hslt^2$ is a conservative variance estimator. Given $0<\alpha <1$, let $q_{\alpha/2}$ denote the upper $\alpha/2$ quantile of the standard normal distribution.
Based on Theorems~\ref{thm:srr2} and \ref{thm:srr2var}, we can construct an asymptotically valid $1 - \alpha$ confidence interval for $\tau$: $\big[\htlt - q_{\alpha/2} {\hslt} / \sqrt{n}, \htlt + q_{\alpha/2} {\hslt} / \sqrt{n} \big]$, whose asymptotic coverage rate is greater than or equal to $1 - \alpha$.
The length of this confidence interval is less than or equal to that based on the estimated asymptotic distributions of $\htu$ in the stratified randomization and stratified rerandomization cases. Therefore, $\htlt$ is the most efficient estimator for all the considered scenarios.
\zk{We provide remarks for discussions and connections to the existing literature.}

\begin{remark}[\zk{Optimality}]
\label{rm:opt2}
Because $\slt^2  = \lim\limits_{n \rightarrow \infty}  E \{ \htu - \tau -  \htaux^\T \bgamma_{\proj} \}^2$, Theorem~\ref{thm:srr2} indicates that $\htlt$ is not only feasible but also has the same asymptotic distribution as $\ttp$ even for unequal propensity scores.
In other words, $\htlt$ has the smallest asymptotic variance among the estimators that have the same asymptotic distribution as $\htu - \htaux^\T \bgamma$ for $\bgamma \in \mathbb{R}^p$ with $\bgamma_{\mathcal{S}^c}= \boldsymbol 0$.
\end{remark}

\begin{remark}[\zk{Distributions of $\htu$ and $\htlt$ under rerandomization}]
The asymptotic distribution of the unadjusted estimator under rerandomization, $\{ \sqrt{n} ( \htu - \tau ) \mid \mathcal{M}_a \}$, is typically a convolution of a normal distribution and a truncated normal distribution \citep{Li2018}. In contrast, the asymptotic distribution of the OLS-adjusted estimator under rerandomization is normal if the covariates used in the design stage are included in the regression adjustment \citep{Li2020}. Because the Lasso-adjusted estimator uses all covariates in $\mathcal{K}$, the set of covariates used in the rerandomization, the asymptotic distribution of $\htlt$ under stratified rerandomization is normal.
\end{remark}

\begin{remark}[\zk{Special cases}]
When $M=1$, stratified rerandomization becomes rerandomization.
Therefore, Theorem~\ref{thm:srr2} extends the results from \citet{Li2020} to high-dimensional settings.
When $a=\infty$, stratified rerandomization becomes stratified randomization.
Our method and theory also apply to this important special case without the need for Condition \ref{cond:rerandomization}. The corresponding theory is provided in the Supplemental Materials.
\end{remark}

\begin{remark}[\zk{Finite sample gain from rerandomization}]
In comparison to stratified randomization, the asymptotic efficiency in the stratified rerandomization scenario does not increase when $\htlt$ is used in the analysis stage. Similar conclusions were derived by \citet{Li2020}, who examined the combination of rerandomization and OLS adjustment.
However, our simulation results indicate that stratified rerandomization can decrease the mean squared error (MSE) of $\htlt$ in finite samples and is thus recommended.
The distinction between the asymptotic theory and finite sample performance may be due to the following two reasons: (1) the asymptotic results for a fixed threshold $a$ do not accurately represent the finite sample behavior when $p_a$ is close to zero \citep{wang2022rerandomization}, and (2) from an asymptotic point of view, Lasso is expected to have adjusted all covariates such that using rerandomization to balance any covariates in advance makes no difference, regardless of how much balance is achieved or how small the fixed threshold $a$ is. However, in finite samples, Lasso may not always be consistent for variable selection. As a result, some relevant covariates, especially those in $\mathcal{K}$, may not be selected and adjusted by the Lasso (resulting in corresponding Lasso-adjusted coefficients being zero due to regularization). In contrast, rerandomization always ensures the balancing of those covariates, leading to better finite sample performance. 
\zk{We can also force Lasso to adjust covariates in $\mathcal{K}$ by setting their corresponding $l_1$ penalties as 0. 
The simulation in the Supplementary Materials shows that the Lasso with forced adjustment could reduce the MSE, but it is not as efficient as the Lasso with rerandomization in the design stage.}
\end{remark}

\begin{remark}[\zk{Power issue}]
\zk{When considering testing $H_0:\tau=0$ versus $H_1:\tau\neq0$. 
We reject $H_0$ if and only if $\tau_0$ does not belong to $\big[\htlt - q_{\alpha/2} {\hslt} / \sqrt{n}, \htlt + q_{\alpha/2} {\hslt} / \sqrt{n} \big]$.
By Theorems~\ref{thm:srr2} and \ref{thm:srr2var}, this test is asymptotically valid (conservative) under $H_0$.
Due to conservative inference, the power of the proposed test may be less than the power of the test based on $(\htu,\hat\sigma^2_{\unadj})$ when the true $\tau$ is small.
A similar conclusion regarding the conservative inference for rerandomization is also obtained in \citet{branson2024power}.
The simulation results in the Supplementary Materials show that the proposed test could increase power compared to the test based on $(\htu,\sigma^2_{\unadj})$ provided that the true $\tau$ is not very small.}
\end{remark}

\begin{remark}[\zk{Proof techniques}]
The proof of Theorem~\ref{thm:srr2} relies on novel concentration inequalities for the weighted sample mean and covariance under stratified randomization. These inequalities are crucial for deriving the $l_1$ convergence rate of the Lasso estimator in a finite population and randomization-based inference framework. We obtain these inequalities in general asymptotic regimes, including the cases of (1) $M$ tending to infinity with fixed $\nm$ and (2) $\nm$ tending to infinity with fixed $M$. These inequalities are of independent interest in other fields where stratified sampling without replacement is performed. This aspect is discussed extensively in the Supplementary Material.
\end{remark}

\section{Simulation}
\label{sec:sim}

% setup
This section describes simulation studies performed to examine the finite-sample performance of the proposed methods. We set the sample size as $n=300$ and $600$. We consider four types of blocks: two large coarse blocks with $n_{[m]}=150$ or 300 and $M=2$, many small coarse blocks with $n_{[m]}=10$ and $M=30$ or 60, hybrid coarse blocks with $n_{[m]}^S=10$, $M^S=10$ or 20, $n_{[m]}^L=100$ or 200, and $M^L=2$, where the superscripts ``S'' and ``L'' denote small and large blocks, respectively, and many triplet fine blocks with $n_{[m]}=3$ and $M=100$ or 200. For the first three types of blocks, we consider equal propensity scores with $\emt$ equal to $0.5$ and unequal propensity scores with $\emt$ evenly spaced in values between 0.3 and 0.7. For the last type of block, we set $\emt$ to be $2/3$ or evenly spaced in values between 0.3 and 0.7. The number of treated units in each block is equal to $\mathrm{round}(\emt n_{[m]})$. We also consider a scenario without blocking and set $n_1 = \sum_{i=1}^{n} Z_i = 0.5 n$. The potential outcomes are generated as follows: $Y_i(z)=(B_i/M)^{2z+1}+\xr_i^\T \bz{z}-2\xr_{bc,i}^\T \bz{z} + \e_i(z)$, $i=1,...,n$, $z = 0, 1$, where $\xr_i$ is generated from a $p$-dimensional multivariate normal distribution $N(0,\Sigma)$ with $\Sigma_{ij}=\rho^{|i-j|}$, $\xr_{bc,i}$ is generated by centering $\xr_i$ in each block, the first $s$ elements of $\bz{z}$ are generated from the uniform distribution on $[0,1]$, the remaining elements are zero, and $\e_i(z)$ is generated from a normal distribution with a mean of zero and variance of $\sigma^2$ such that the signal-to-noise ratio is equal to 5. We set $p=400$, $s=10$, and $\rho=0.6$. The potential outcomes and covariates are generated once and then kept fixed.

% method
For each scenario, we consider two designs (with/without rerandomization) and two estimators (unadjusted and Lasso-adjusted estimators). We set $\xr_{i,\mathcal{K}}$ as the first $k=5$ dimensions of $\xr_i$ and $p_a=0.001$ for rerandomization. We use the R package ``glmnet'' to fit the solution path of Lasso. We choose the tuning parameter in Lasso via 10-fold cross validation with the default option ``lambda.1se.'' We replicate the randomization/rerandomization 1000 times to evaluate the repeated sampling properties.
We employ a bootstrap approach to evaluate the variability across simulations. Specifically, subsequent to obtaining 1000 sets of inferential results (point estimates and confidence intervals), we resample with replacement from them to construct new bootstrap samples of inferential results. These bootstrap samples are utilized to compute bootstrap versions of summary statistics, such as the coverage probability. Through the iterative resampling process conducted 500 times, we obtain 500 bootstrap versions of summary statistics. Finally, we utilize the standard error associated with these 500 bootstrap versions to provide an approximation of the Monte Carlo standard error corresponding to summary statistics.

\begin{figure}[t]
\centering
\includegraphics[scale = 0.6]{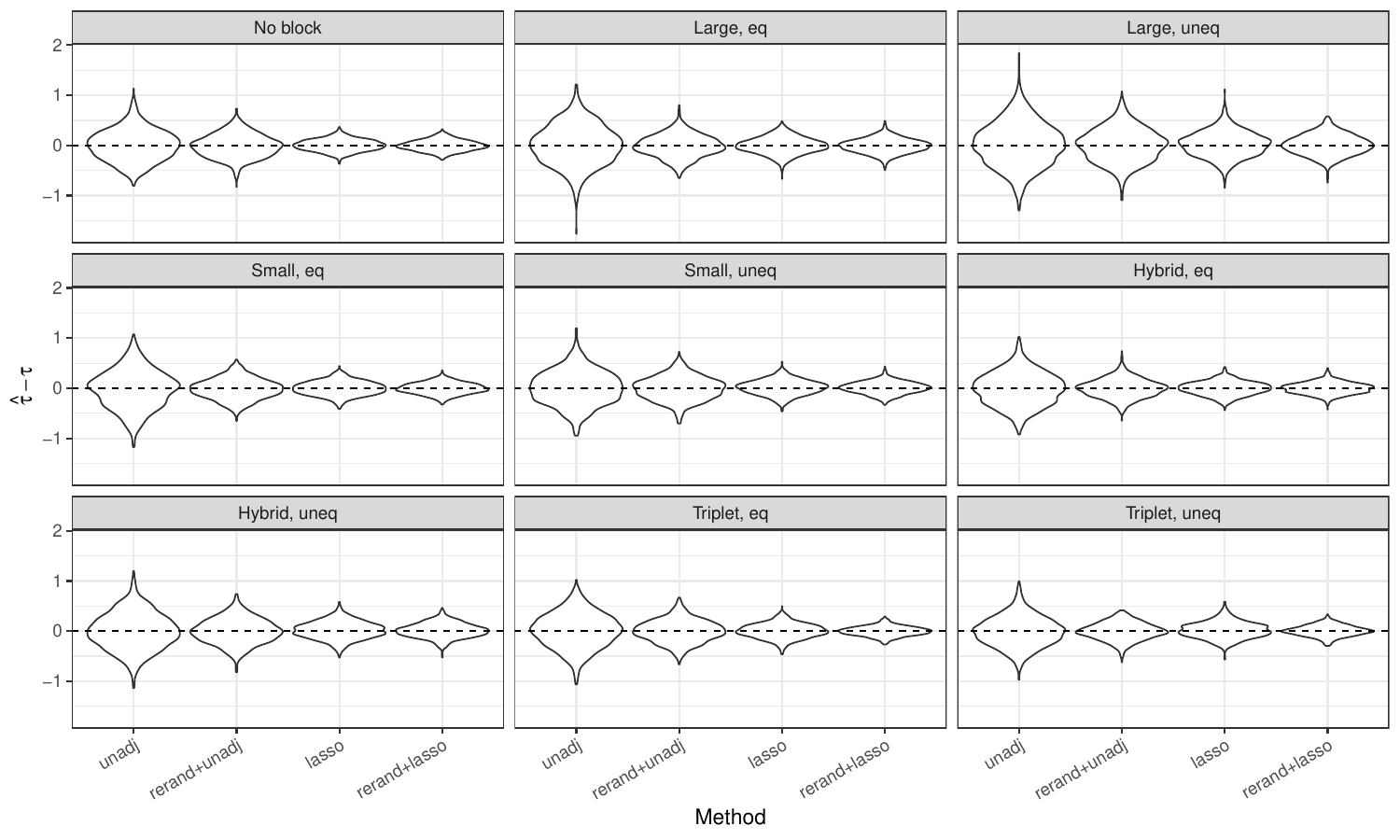}
\caption{Distributions of the average treatment effect estimators minus the true value of the average treatment effect for different scenarios when $n=300$.}
\label{fig:sim_dist300}
\end{figure}

% result
Figure \ref{fig:sim_dist300} shows the distributions (violin plots) of different estimators for $n=300$. All distributions are symmetric around the true value of the average treatment effect. The distributions of the Lasso-adjusted estimator are more concentrated than those of the unadjusted estimator under both randomization and rerandomization. Table \ref{tab:sim_300} presents several summary statistics for different estimators when $n=300$. The simulation results for $n=600$ are similar and relegated to the Supplementary Material.
First, for all designs, the absolute value of the bias of each estimator is considerably smaller than the standard deviation (SD). Second, compared with the unadjusted estimator without rerandomization, the Lasso-adjusted estimator with rerandomization reduces the standard deviation and root mean squared error (RMSE) by 54\%--77\%. Third, the empirical coverage probabilities (CP) of all estimators reach the nominal level $95\%$ (in a few cases, the coverage probabilities are less than $95\%$ but very close to $95\%$). Fourth, compared to the unadjusted estimator without rerandomization, the Lasso-adjusted estimator with rerandomization decreases the mean confidence interval length (Length) by 37\%--58\%. Finally,
given using the Lasso-adjusted estimator in the analysis stage, compared to stratified randomization, stratified rerandomization can further decrease the root mean squared error.
Thus, based on the simulation results, our final recommendation is to implement stratified rerandomization in the design stage and to use the Lasso-adjusted estimator in the analysis stage.

\begin{table}[p]
\caption{\label{tab:sim_300}Simulation results for different scenarios when $n=300$.}
\centering
\begin{threeparttable}
\begin{tabular}[t]{>{\centering\arraybackslash}p{1.4cm}ccrrrrr}
\toprule
\multicolumn{1}{c}{Scenario} & \multicolumn{1}{c}{Rerand.} & \multicolumn{1}{c}{Est.} & \multicolumn{1}{c}{Bias} & \multicolumn{1}{c}{SD} & \multicolumn{1}{c}{RMSE} & \multicolumn{1}{c}{CP} & \multicolumn{1}{c}{Length}\\
\midrule
& no & unadj & 0.6 (1.0) & 30.7 (0.7) & 30.6 (0.7) & 96.6 (0.6) & 130.5 (0.1)\\

& yes & unadj & -0.8 (0.7) & 22.7 (0.5) & 22.7 (0.5) & 97.1 (0.5) & 98.2 (0.1)\\

& no & lasso & 0.4 (0.4) & 11.5 (0.3) & 11.5 (0.3) & 99.6 (0.2) & 66.9 (0.1)\\

\multirow{-4}{1.4cm}{\centering\arraybackslash No block} & yes & lasso & -0.2 (0.3) & 10.6 (0.2) & 10.6 (0.2) & 100.0 (0.0) & 67.0 (0.1)\\
\midrule

& no & unadj & 1.2 (1.2) & 39.4 (0.9) & 39.4 (0.9) & 96.0 (0.6) & 161.2 (0.1)\\

& yes & unadj & -1.1 (0.7) & 22.0 (0.5) & 22.0 (0.5) & 97.5 (0.5) & 99.0 (0.1)\\

& no & lasso & -0.3 (0.5) & 16.4 (0.4) & 16.4 (0.4) & 98.7 (0.3) & 81.6 (0.1)\\

\multirow{-4}{1.4cm}{\centering\arraybackslash Large, equal} & yes & lasso & -0.7 (0.4) & 14.2 (0.3) & 14.2 (0.3) & 99.2 (0.3) & 81.8 (0.1)\\
\midrule

& no & unadj & 2.8 (1.3) & 43.5 (1.0) & 43.5 (1.0) & 96.4 (0.6) & 182.8 (0.3)\\

& yes & unadj & 0.7 (1.1) & 32.9 (0.7) & 32.9 (0.7) & 96.5 (0.6) & 139.9 (0.2)\\

& no & lasso & 2.6 (0.8) & 24.8 (0.6) & 25.0 (0.6) & 98.0 (0.4) & 115.1 (0.3)\\

\multirow{-4}{1.4cm}{\centering\arraybackslash Large, unequal} & yes & lasso & 1.3 (0.6) & 20.0 (0.4) & 20.0 (0.4) & 99.6 (0.2) & 115.2 (0.2)\\
\midrule

& no & unadj & -1.8 (1.2) & 37.9 (0.8) & 37.9 (0.8) & 95.7 (0.6) & 154.9 (0.1)\\

& yes & unadj & -1.3 (0.7) & 20.4 (0.4) & 20.4 (0.4) & 97.1 (0.5) & 90.0 (0.1)\\

& no & lasso & -0.0 (0.4) & 13.9 (0.3) & 13.8 (0.3) & 99.4 (0.2) & 72.0 (0.1)\\

\multirow{-4}{1.4cm}{\centering\arraybackslash Small, equal} & yes & lasso & -0.4 (0.4) & 11.1 (0.2) & 11.1 (0.2) & 100.0 (0.0) & 72.2 (0.1)\\
\midrule

& no & unadj & 1.0 (1.1) & 33.8 (0.7) & 33.8 (0.7) & 95.3 (0.6) & 137.0 (0.2)\\

& yes & unadj & -0.5 (0.7) & 23.2 (0.5) & 23.2 (0.5) & 94.7 (0.7) & 93.1 (0.2)\\

& no & lasso & 1.7 (0.5) & 14.8 (0.3) & 14.9 (0.3) & 97.4 (0.5) & 65.7 (0.1)\\

\multirow{-4}{1.4cm}{\centering\arraybackslash Small, unequal} & yes & lasso & 0.9 (0.4) & 11.8 (0.3) & 11.8 (0.3) & 99.3 (0.3) & 65.6 (0.1)\\
\midrule

& no & unadj & -1.1 (1.1) & 32.6 (0.7) & 32.6 (0.7) & 97.2 (0.5) & 144.5 (0.1)\\

& yes & unadj & -0.1 (0.6) & 18.6 (0.5) & 18.6 (0.5) & 98.7 (0.3) & 94.3 (0.1)\\

& no & lasso & -0.4 (0.4) & 13.4 (0.3) & 13.4 (0.3) & 99.8 (0.2) & 82.6 (0.1)\\

\multirow{-4}{1.4cm}{\centering\arraybackslash Hybrid, equal} & yes & lasso & 0.0 (0.4) & 11.5 (0.3) & 11.5 (0.3) & 99.9 (0.1) & 82.8 (0.1)\\
\midrule

& no & unadj & -0.2 (1.2) & 36.2 (0.8) & 36.2 (0.8) & 96.3 (0.6) & 149.3 (0.1)\\

& yes & unadj & -0.2 (0.7) & 24.8 (0.6) & 24.8 (0.6) & 96.6 (0.6) & 106.0 (0.2)\\

& no & lasso & 1.8 (0.5) & 17.3 (0.4) & 17.3 (0.4) & 98.0 (0.4) & 80.2 (0.2)\\

\multirow{-4}{1.4cm}{\centering\arraybackslash Hybrid, unequal} & yes & lasso & 1.3 (0.5) & 14.3 (0.3) & 14.3 (0.3) & 99.2 (0.3) & 80.3 (0.2)\\
\midrule

& no & unadj & 1.1 (1.0) & 34.3 (0.7) & 34.3 (0.7) & 95.1 (0.7) & 135.7 (0.2)\\

& yes & unadj & -0.1 (0.7) & 22.0 (0.5) & 22.0 (0.5) & 93.8 (0.8) & 86.4 (0.4)\\

& no & lasso & 0.2 (0.5) & 14.5 (0.3) & 14.5 (0.3) & 96.2 (0.6) & 60.0 (0.2)\\

\multirow{-4}{1.4cm}{\centering\arraybackslash Triplet, equal} & yes & lasso & -0.2 (0.3) & 9.4 (0.2) & 9.4 (0.2) & 100.0 (0.0) & 59.5 (0.2)\\
\midrule

& no & unadj & -0.9 (1.0) & 29.7 (0.7) & 29.7 (0.7) & 95.1 (0.7) & 118.9 (0.2)\\

& yes & unadj & -0.7 (0.6) & 17.1 (0.4) & 17.1 (0.4) & 93.9 (0.8) & 67.1 (0.3)\\

& no & lasso & 1.7 (0.5) & 16.4 (0.4) & 16.4 (0.4) & 95.6 (0.7) & 66.5 (0.2)\\

\multirow{-4}{1.4cm}{\centering\arraybackslash Triplet, unequal} & yes & lasso & 0.5 (0.3) & 10.8 (0.3) & 10.8 (0.3) & 99.9 (0.1) & 67.0 (0.2)\\
\bottomrule
\end{tabular}
\begin{tablenotes}
\item Note: The numbers in brackets are the corresponding standard errors estimated using the bootstrap with 500 replications. Bias, SD, RMSE, CP, Length, and their standard errors are multiplied by 100.
\end{tablenotes}
\end{threeparttable}
\end{table}

\section{Real data illustration}
\label{sec:rd}

\zk{In this section, we consider a randomized block experiment with heterogeneous propensity scores across blocks and a matched observational study with fine blocks.}

\subsection{``Opportunity Knocks'' experiment}
\label{sec:rd_ok}

% data
In this part, we use experimental data to illustrate the merits of the combination of stratified rerandomization and the Lasso adjustment. The ``Opportunity Knocks'' (OK) randomized experiment aimed at evaluating the effect of academic achievement awards on the academic performance of college students \citep{angrist2014opportunity}. Based on sex and discretized high school grades, $n=506$ second-year college students were stratified into $M=8$ blocks, with sizes ranging from 42 to 90. In each block, only approximately 25 students were assigned to the treatment group (receiving incentives); thus, the propensity scores were significantly different across blocks.

% analysis
We consider the grade point average (GPA) at the end of the fall semester as the outcome. There were 23 baseline covariates, such as demographic variables, GPA in the previous year, and whether the students correctly answered tests about the scholarship formula. We adjust for the main effect, quadratic terms of the continuous covariates, and two-way interactions. The design matrix $\xm$ contains $p=253$ columns (covariates) and $n=506$ rows (observations). Based on the unadjusted estimator, the average treatment effect estimate is 0.032 and the 95\% confidence interval is $[-0.099, 0.163]$. Based on the Lasso-adjusted estimator, which selects two covariates (``gpapreviousyear'' and ``gpapreviousyear:test1correct'') into the model, the average treatment effect estimate is 0.038 and the 95\% confidence interval is $[-0.069, 0.146]$. Both confidence intervals contain zero, which means that there is insufficient evidence to support the effectiveness of the scholarship program. The Lasso-adjusted estimator appears to be more efficient than the unadjusted estimator because it shortens the interval length by 18\%.

\subsection{An observational study of fish consumption}
\label{sec:rd_fish}

% data
The second dataset was obtained from the US National Health and Nutrition Examination Survey (NHANES) 2013–2014. \citet{zhao2018cross} investigated the effect of high-level fish consumption on biomarkers based on this dataset. We use this dataset to illustrate the application of the proposed Lasso-adjusted estimator in the matched observational study. There were 88 people with high fish consumption and 663 with low fish consumption. We use the blood cadmium level on the log scale as the outcome. Many covariates were available, such as demographic variables, disability, and history of drugs, alcohol, and smoking. As in \citet{Pashley2017}, we match for age, sex, race, income, education, and smoking. We fit the propensity score model using the R package ``brglm'' and perform full matching using the R package ``optmatch,'' which produces 88 fine blocks, each consisting of one treated unit and a range of one to eight control units.

% analysis
After full matching, we could analyze the data as a finely stratified experiment \citep{bind2019,Pashley2017}. In contrast to \citet{Pashley2017}, where inference is based on the weighted difference-in-means estimator, we perform covariate adjustment to improve efficiency. Thirty baseline covariates were included in the initial dataset. We perform feature engineering and include the main effect, quadratic terms of the continuous covariates, and two-way interactions, which generate a design matrix $\xm$ with $p=390$ columns (covariates) and $n=751$ rows (observations). Based on the unadjusted estimator, the average treatment effect estimate is 0.139 and the 95\% confidence interval is $[-0.052, 0.330]$. Based on the Lasso-adjusted estimator, which selects 16 covariates into the model, the average treatment effect estimate is 0.129 and the 95\% confidence interval is $[-0.046, 0.305]$. Both methods indicate that there is insufficient evidence to support the hypothesis that high fish consumption affects blood cadmium level. Notably, the Lasso-adjusted estimator shortens the confidence interval length by 8\% and is thus more efficient than the unadjusted estimator.

\section{Discussion}
\label{sec:dis}

This study aimed to enhance the estimation and inference efficiencies of the average treatment effect in randomized experiments when many baseline covariates are available. We propose a Lasso-adjusted average treatment effect estimator in randomized block experiments based on a projection perspective. Under mild conditions, we obtain the asymptotic distribution of the proposed estimator when blocking, rerandomization, or both are implemented in the design stage. We demonstrate that the proposed estimator enhances, or at least does not deteriorate, the precision compared with that associated with the unadjusted estimator.
Our results are design-based and robust to model misspecification as long as the relationship between the potential outcomes and covariates can be approximated well by a sparse linear projection.
Our results are also robust to heterogeneous block sizes, propensity scores, and treatment effects. In addition, we propose a conservative variance estimator to construct asymptotically conservative confidence intervals or tests for the average treatment effect.
Based on the simulation results, we recommend using blocking and rerandomization in the design stage to balance a subset of covariates that are most predictive to the potential outcomes and then implementing regression adjustment using the Lasso in the analysis stage to adjust for the remaining covariate imbalances.
Similar to the findings reported by \citet{Li2020}, when rerandomization or the combination of blocking and rerandomization is used in the design stage, the Lasso adjustment should consider all of the covariates used in the rerandomization to ensure efficiency gains.

To render the theory and methods more intuitive, we focus on inferring the average treatment effect for a binary treatment. Our analysis can be generalized to multiple value treatments, including factorial experiments \citep{Fisher1935,Li2020factorial,liu2021randomization}. Moreover, it may be interesting to extend our results to other complicated settings, such as binary outcomes based on penalized logistic regression \citep{freedman2008randomization, zhang2008binary} and the use of other machine learning methods such as random forest \citep{wager2016}, $L_2$-boosting \citep{kueck2022estimation}, and neural networks \citep{farrell2021deep}.

In practice, some experimental units may not comply with their treatment assignment; that is, the actual treatments received by the experimental units may be different from the treatments assigned \citep{imbens1994identification,angrist1996identification}. When there is noncompliance, investigators often use the two-stage least squares to estimate the complier average treatment effect \citep{angrist2008mostly}. It would be interesting to extend our methods to noncompliance settings.

\section*{Acknowledgement}
Dr. Liu is supported by the National Natural Science Foundation of China (Grant No. 12071242). Dr. Yang is supported by the National Natural Science Foundation of China (Grant No. 12001557, 12371281), the Emerging Interdisciplinary Project, Program for Innovation Research, and the Disciplinary Funding in Central University of Finance and Economics. The authors thank the editor, associate editor and reviewers for their valuable suggestions.

\section*{Supplementary Material}

The Supplementary Material provides concentration inequalities under stratified randomization, proofs and additional simulation results.

\bibliographystyle{agsm}
\bibliography{causal}

%%%%%%%%%%%%%%%%%%%%%%%%%%%%%%%%%%%%%%%%%%%%%%%%%%%%%%%%%%%%%%%%%%%%%%%%%%%%%%

\newpage

\if0\blind
{
  \title{\bf Supplementary Material for ``Design-based theory for Lasso adjustment in randomized block experiments and rerandomized experiments''}
  \author{}
  \date{}
  \maketitle
} \fi

\if1\blind
{
  \title{\bf Supplementary Material for ``Design-based theory for Lasso adjustment in randomized block experiments and rerandomized experiments''}
  \author{}
  \date{}
  \maketitle
} \fi

\spacingset{1.8} % DON'T change the spacing!

\appendix

\setcounter{equation}{0}
\renewcommand{\theequation}{S\arabic{equation}}
\setcounter{table}{0}
\renewcommand{\thetable}{S\arabic{table}}
\setcounter{figure}{0}
\renewcommand{\thefigure}{S\arabic{figure}}
\setcounter{theorem}{0}
\renewcommand{\thetheorem}{S\arabic{theorem}}
\setcounter{proposition}{0}
\renewcommand{\theproposition}{S\arabic{proposition}}
\setcounter{lemma}{0}
\renewcommand{\thelemma}{S\arabic{lemma}}
\setcounter{condition}{0}
\renewcommand{\thecondition}{S\arabic{condition}}
\setcounter{remark}{0}
\renewcommand{\theremark}{S\arabic{remark}}

Section \ref{sec:pre} provides bounds for sampling without replacement, concentration inequalities under stratified randomization, and asymptotic results of $\htu$ under stratified randomization and stratified rerandomization.
Section \ref{sec:proof-main} provides proofs of Propositions~\ref{thm:sr2} and \ref{thm:sr2var} and Theorems~\ref{thm:srr2} and \ref{thm:srr2var}.
Section \ref{sec:addsim} provides additional simulation results.
\zhu{We collect notations in the main text into Tables \ref{tab:notations} and \ref{tab:notations2}: one for notations related to blocking and rerandomization, another for notations related to covariate adjustment.}

%\tableofcontents

\begin{table}[H]
\centering
\renewcommand{\arraystretch}{1.5}
\caption{Summary of notations for blocking and rerandomization}
\label{tab:notations}
\begin{threeparttable}
\begin{tabular}{cp{0.8\textwidth}}
\toprule
Notation & Definition \\
\midrule

$M$ & Number of blocks \\

$\nm$  & Number of units in block $m$ \\

$\pim$ &  Proportion of block size for block $m$: $\pim = \nm / n$ \\

$\nmt$ ($\nmc$) & treatment (control) group size in block $m$ \\

$\emt$ & Propensity score in block $m$: $\emt = \nmt / \nm$ \\

$\bar{\H}_{[m]}$ & Block-specific finite population mean: $\bar{\H}_{[m]}= \nm^{-1} \sumim \H_i$ \\

$\bar{\H}_{[m]z}$ & Block-specific sample mean: $\bar{\H}_{[m]z} = \nmz^{-1} \sumim  I(Z_i = z) \H_i$ \\

% $\bar{R}(z)$ & Overall finite population mean: $\bar{R}(z) =\summ \pim \bar{R}_{[m]}(z)$ \\

% $\bar{R}_{z}$ & Weighted-sample mean: $\bar{R}_{z} = \summ \pim  \bar{R}_{[m]z}$ \\

$S_{[m] \H \Q}$ & Block-specific finite population covariance: $S_{[m] \H \Q} = ( \nm - 1 )^{-1} \sumim ( \H_i - \bar{\H}_{[m]} ) ( \Q_i - \bar{\Q}_{[m]} )^\T$ \\

$S^2_{[m]\H}$ & $S^2_{[m]\H} = S_{[m]\H \H}$ \\

$s_{[m]\H \Q}$ & Block-specific sample covariance: $s_{[m]\H \Q}= ( \nmz - 1 )^{-1} \sumim I(Z_i=z) ( \H_i - \bar{\H}_{[m]z} ) ( \Q_i - \bar{\Q}_{[m]z} )^\T$ \\

$s^2_{[m]\H}$ & $s^2_{[m]\H}=s_{[m]\H \H}$ \\

% $\Sigma_{\H \Q}$ & Overall finite population covariance: $\Sigma_{\H \Q}= \summ\pim S_{[m]\H \Q}$ \\

% $\hat{\Sigma}_{\H \Q}$ & Weighted-sample covariance: $\hat{\Sigma}_{\H \Q} = \summ\pim s_{[m]\H \Q}$ \\

$\tm$ & Block-specific ATE: $\tm = \bym{1} - \bym{0}$ \\

$\htmu$ & Block-specific difference-in-means: $\htmu = \bymz{1} - \bymz{0}$ \\

$\tau$ & Overall ATE: $\tau = n^{-1} \sumi \big\{ Y_i(1) - Y_i(0) \big\} =  \summ \pim \tm$ \\

$\htu$ & Weighted difference-in-means estimator: $\htu = \summ \pim \htmu$ \\

$\hsu^2$ & Variance estimator for $\htu$\\

$\xm_{\mathcal{K}}$ & Covariates that are balanced in the design stage \\

$\htauw$ & Weighted difference-in-means of $\xm_{\mathcal{K}}$: $\htauw = \summ \pim \{ (\bar{\xr}_{[m]1})_\mathcal{K} - (\bar{\xr}_{[m]0})_\mathcal{K} \}$ \\

$\maha(\Z, \wm)$ & Mahalanobis distance \\

$a$ & Rerandomization threshold \\

\bottomrule
\end{tabular}

\begin{tablenotes}
\item Note: $\H$ and $\Q$ could be potential outcomes, adjusted potential outcomes, or covariates.
\end{tablenotes}
\end{threeparttable}
\end{table}

\begin{table}[H]
\centering
\renewcommand{\arraystretch}{1.5}
\caption{Summary of notations for covariate adjustment}
\label{tab:notations2}
\begin{tabular}{cp{0.8\textwidth}}
\toprule
Notation & Definition \\
\midrule

$\taux$ & Average treatment effect of $\xm$: $\taux = \summ \pim \{\bar{\xr}_{[m]} - \bar{\xr}_{[m]} \} = \bb{0}$ \\

$\htaux$ & Weighted difference-in-means of $\xm$: $\htaux = \summ \pim ( \bxmz{1} - \bxmz{0} )$ \\

$\g$ & Projection coefficient vector: $\g = \argmin_{\bgamma} E ( \htu - \tau -  \htaux^\T \bgamma )^2$ \\

$\ttp$ & Oracle projection estimator: $\ttp = \htu - \htaux^\T \g $ \\

% $\emz$ & Treatment (control) proportion in block $m$: $\emz=z \emt + (1 - z )( 1 - \emt)$ \\

$\gz{z}$ & Projection coefficient vector: $\g=\gz{0}+\gz{1}$ \\

$\tsmxz$ & Unbiased estimators for $\Smx$ \\

$\tsmxyz{z}$ & Unbiased estimators for $\Smxyz{z}$ \\

$\hgoz$ & Plug-in estimator of $\gz{z}$; $\hgoz$ could also be obtained through a weighted regression \\

% $\wnz{z}$ & Weight for the weighted regression: $\wnz{z} = \nm^2/\{\emz\nmz(\nm-1)\}$ \\

% $\wyz{z}$ & Weight for the weighted regression: $\wyz{z} = 1 - \emz$ \\

% $\wxz{z}$ & Weight for the weighted regression: $\wxz{z} = 1/(1 - \emz)$ \\

$\yw_i(z)$ & Weighted potential outcomes: $\yw_i(z)=\sqrt{\wnz{z} \wyz{z}} \,Y_i(z)$ \\

$\xwr_{i}(z)$ & Weighted covariates: $\xwr_{i}(z)=\sqrt{\wnz{z} \wxz{z}} \,( \xr_{i} - \bar{\xr}_{[m]} )$ \\

$\yw_i$ & Observed weighted outcome:  $\yw_i=Z_i\yw_i(1)+(1-Z_i)\yw_i(0)$ \\

$\xwr_{i}$ & Observed weighted covariates: $\xwr_{i}=Z_i\xwr_i(1)+(1-Z_i)\xwr_i(0)$ \\

$\bxwz{z}$ & Sample mean of $\xwr_{i}$: $\bxwz{z} =\nz^{-1}\sum_{i:Z_i= z} \xwr_{i}$ \\

%$\hgoz$ & $\gz{z}$'s estimator obtained through a weighted regression with an intercept \\

$\hgo$ & Estimator of $\g$: $\hgo = \hgozt + \hgozc$ \\

$\htp$ & Regression-adjusted ATE estimator: $\htp = \htu - \htaux^\T \hgo $\\

$\mathcal{S}$ & Set of relevant covariates that are predictive to the outcomes \\

$\hgz{z}$ & Lasso estimator of $\gz{z}$ \\

$\hg$ & Lasso estimator of $\g$: $\hg=\hgz{1} + \hgz{0}$ \\

$\htlt$ & Lasso-adjusted ATE estimator: $\htlt = \htu -  \htaux^\T \hg$ \\

$\mathcal{A}_c$ & sets of coarse blocks: $\mathcal{A}_c=\{1\leqslant m\leqslant M:\nmt > 1, \ \nmc > 1\}$ \\

$\mathcal{A}_f$ & sets of fine blocks: $\mathcal{A}_f=\{1\leqslant m\leqslant M:\nmt= 1 \textnormal{ or } \nmc = 1\}$ \\

$n_f$ & Total number of units in fine blocks: $n_f=\sum_{m\in\mathcal{A}_f}\nm$ \\

%$\omega_{[m]}$ & Weight for variance estimator: $\omega_{[m]}=\nm^2/(n_f-2\nm)$ \\

$\tY_i$ & Adjusted outcome: $\tY_i=Y_i- (\xr_i-\bxm)^\T \hg$ \\

$\hat{\tau}_{\tY,[m]}$ & Block-specific difference-in-means of $\tY_i$: $\hat{\tau}_{\tY,[m]}=\bar{\tY}_{[m]1}-\bar{\tY}_{[m]0}$ \\

$\hat{\tau}_{\tY,f}$ & Weighted difference-in-means of $\tY_i$ in fine blocks: $\hat{\tau}_{\tY,f}=\sum_{m\in \mathcal{A}_f} (\nm/n_f) \hat{\tau}_{\tY,[m]}$ \\

$\hslt^2$ & Variance estimator for Lasso-adjusted ATE estimator $\htlt$ \\

$\hat s$ & Sparsity of $\hg$: $\hat s = ||\hg||_0$ \\

\bottomrule
\end{tabular}
\end{table}

\section{Some preliminary results}
\label{sec:pre}

\subsection{Bounds for sampling without replacement}

The connection between randomized experiments and survey sampling has been discussed in depth by many scholars \citep{lin2013,fcltxlpd2016,mukerjee2018,lei2020}. Both of them are based on a probability model of sampling without replacement from a finite population. We start by introducing Bobkov's inequality, a powerful tool to prove concentration inequalities for sampling without replacement. In this section, we consider completely randomized experiments; that is,
$$
\pr ( \Z =  \z ) = \frac{ \nt ! \nc ! }{ n ! }, \quad \sumi I(z_i = 1) = \nt, \quad z_i=0,1.
$$
We denote the propensity score by $\ec=\nt/n$. The value space of $\Z$ is defined as
$$
\mathcal{G}=\left\{\z=(z_1,\dots,z_n)\in\{0,1\}^n: z_{1}+\cdots+z_{n}=n_1\right\}.
$$
For every $\z\in\mathcal{G}$, we pick a pair of units $(i,j)$ such that $z_i=1$ and $z_j=0$, and switch the value of $z_i$ and $z_j$ to obtain a ``neighbour'' of $\z$, denoted by $\z^{i,j}$. Clearly, for different $(i,j)$, $\z$ has totally $n_1(n-n_1)$ neighbours. For every real-valued function $f$ on $\mathcal{G}$, we define the discrete gradient as follows: $\nabla f(\z)= (f(\z)-f\left(\z^{i,j}\right))_{i,j}$, which is an $n_1(n-n_1)$ dimensional vector. We define the $\ell_2$ norm of $\nabla f(\z)$ as
$$
\|\nabla f(\z)\|^{2}_2=\sum_{i:z_i=1} \sum_{j:z_j=0}\left|f(\z)-f\left(\z^{i,j}\right)\right|^{2}.
$$

\begin{lemma}[\citet{bobkov2004}]\label{lem:bobkov}
For every real-valued function $f$ on $\mathcal{G}$, if $\|\nabla f(\z)\|_{2}\leqslant \sigma$ for all $z\in\mathcal{G}$, then
$$
E \exp \left[t\left\{f(\Z)-Ef(\Z)\right\}\right] \leqslant \exp \left\{\sigma^{2} t^{2} /(n+2)\right\}, \quad t\in\mathbb{R}.
$$
\end{lemma}

Consider two sequences of real numbers $(a_1,\dots,a_n)$ and $(b_1,\dots,b_n)$, we denote
$$
\bar a=\frac1n\sum_{i=1}^n a_i,\quad \bar a_1=\bar a_1(\Z)=\frac{1}{\nt}\sum_{i=1}^n Z_ia_i,
$$
$$
S_{a b} = \frac{ 1 }{ n - 1 } \sum_{i=1}^n ( a_i - \bar a ) ( b_i - \bar b ) ,\quad
s_{a b} =s_{a b}(\Z) = \frac{ 1 }{ n_1 - 1 } \sum_{i=1}^n Z_i ( a_i - \bar a_1 ) ( b_i - \bar b_1 ).
$$
The following result from \citet{zhang2012some} is useful to bound $\|\nabla f(\z) \|^{2}_2$.
\begin{lemma}[\citet{zhang2012some}]\label{lem:zhang}
We have
$$
\sum_{i=1}^n ( a_i - \bar a ) ( b_i - \bar b ) =
\frac1{n}\sum_{1\leqslant i<j\leqslant n} \left(a_i-a_j\right)\left(b_i-b_j\right). $$
\end{lemma}

Next, we apply Bobkov's inequality to derive the bounds for the sample mean and sample covariance, respectively.
\begin{lemma}\label{lem:expmean}
For $t\in\mathbb{R}$, we have
\begin{equation}
 E \exp \left\{t\left(\bar a_1-\bar a\right)\right\} \leqslant \exp \left\{\smean^{2} t^{2} /(n+2)\right\},\nonumber
\end{equation}
where $\smean^2=\ec^{-2}n^{-1}\sum_{i=1}^n(a_i-\bar a)^2$.
\end{lemma}

\begin{proof}
By definition and simple calculation, we have
\begin{equation}\label{eq:da1}
\|\nabla \bar a_1(\z)\|^{2}_2=\sum_{i:z_i=1} \sum_{j:z_j=0}\left|\bar a_1(\z)-\bar a_1(\z^{i,j})\right|^{2}
=\frac1{n_1^2}\sum_{i:z_i=1} \sum_{j:z_j=0}\left|a_i-a_j\right|^{2}.
\end{equation}
Then, by Lemma \ref{lem:zhang}, we have
\begin{equation}\label{eq:da2}
\frac1{n_1^2}\sum_{i:z_i=1} \sum_{j:z_j=0}\left|a_i-a_j\right|^{2}\leqslant \frac1{n_1^2}\sum_{1\leqslant i<j\leqslant n}\left|a_i-a_j\right|^{2}=\frac{n}{n_1^2}\sum_{i=1}^n(a_i-\bar a)^2=:\smean^2.
\end{equation}
Combining (\ref{eq:da1}) and (\ref{eq:da2}), Lemma~\ref{lem:expmean} follows from Lemma \ref{lem:bobkov}.
\end{proof}

\begin{lemma}\label{lem:expcov}
If $\nmt\geqslant2$ for all $m$, then, for $t\in\mathbb{R}$,
\begin{equation}
 E \exp \left\{t\left(s_{ab}-S_{ab}\right)\right\} \leqslant \exp \left\{\scov^{2} t^{2} /(n+2)\right\},\nonumber
\end{equation}
where
$$
\scov^2=\Bigg\{\sqrt{ \frac{1}{\ec^2n}\sum_{i=1}^n(a_i-\bar a)^2(b_i-\bar b)^2
} + \sqrt{
\frac{8}{\ec^3n^2}\sum_{i=1}^n(a_i-\bar a)^2\sum_{i=1}^n(b_i-\bar b)^2
}\Bigg\}^2.
$$
\end{lemma}

\begin{proof}

We start by examining $s_{ab}(\z)-s_{ab}(\z^{i,j})$. By Lemma \ref{lem:zhang} and some simple calculation, we have
$$
\begin{aligned}
&s_{ab}(\z)-s_{ab}(\z^{i,j}) \\
=&\frac1{n_1(n_1-1)}\sum_{1\leqslant i'<j'\leqslant n}\{z_{i'}z_{j'} (a_{i'}-a_{j'})(b_{i'}-b_{j'}) -z^{i,j}_{i'}z^{i,j}_{j'} (a_{i'}-a_{j'})(b_{i'}-b_{j'})\}\\
=&\frac1{n_1(n_1-1)}\sum_{l\neq i}z_l\{(a_{l}-a_{i})(b_{l}-b_{i}) - (a_{l}-a_{j})(b_{l}-b_{j})\}\\
=&\frac1{n_1(n_1-1)}\sum_{l\neq i}z_l\{a_ib_i-a_jb_j+(a_{j}-a_{i})b_{l} + a_{l}(b_{j}-b_{i}) \}\\
=&\frac1{n_1(n_1-1)}\sum_{l\neq i}z_l\{(a_ib_i-a_i\bar b-\bar ab_i+\bar a\bar b)-(a_jb_j-a_j\bar b-\bar ab_j+\bar a\bar b)  \\
&+(a_{j}-a_{i})(b_{l}-\bar b) + (a_{l}-\bar a)(b_{j}-b_{i})\}\\
=&\frac1{n_1}(U_{ij}+V_{ij}),
\end{aligned}
$$
where
$$
U_{ij}:=(a_i-\bar a)(b_i-\bar b)-(a_j-\bar a)(b_j-\bar b),
$$
$$
V_{ij}:=\frac{(a_{j}-a_{i})}{n_1-1}\sum_{l\neq i}z_l(b_{l}-\bar b) + \frac{(b_{j}-b_{i})}{n_1-1}\sum_{l\neq i}z_l(a_{l}-\bar a).
$$
By Cauchy--Schwarz inequality, we have
\begin{align}
    V_{ij}\leqslant&
    |a_{j}-a_{i}|\sqrt{\frac{1}{n_1-1}\sum_{l=1}^n(b_l-\bar b)^2}+
    |b_{j}-b_{i}|\sqrt{\frac{1}{n_1-1}\sum_{l=1}^n(a_l-\bar a)^2}\nonumber\\
    \leqslant& |a_{j}-a_{i}|\sqrt{\frac{2}{\ec n}\sum_{l=1}^n(b_l-\bar b)^2}+
    |b_{j}-b_{i}|\sqrt{\frac{2}{\ec n}\sum_{l=1}^n(a_l-\bar a)^2}.\label{eq:Vij}
\end{align}
Then, we can bound $\|\nabla s_{ab}(\z)\|^{2}_2$. By definition and Minkowski's inequality, we have
\begin{align}
\|\nabla s_{ab}(\z)\|^{2}_2
=&\sum_{i:z_i=1} \sum_{j:z_j=0}\left|s_{ab}(\z)-s_{ab}(\z^{i,j})\right|^{2}\nonumber\\
=&\sum_{i:z_i=1} \sum_{j:z_j=0}\left|U_{ij}/n_1+V_{ij}/n_1\right|^{2}\nonumber\\
\leqslant&\sum_{1\leqslant i<j\leqslant n}\left|U_{ij}/n_1+V_{ij}/n_1\right|^{2}\nonumber\\
\leqslant&(\sqrt{U}+\sqrt{V})^2,\label{eq:dsab}
\end{align}
where
$$
U:=\frac1{n_1^2}\sum_{1\leqslant i<j\leqslant n}U_{ij}^2,\quad V:=\frac1{n_1^2}\sum_{1\leqslant i<j\leqslant n}V_{ij}^2.
$$
We bound $U$ and $V$ separately. By Lemma \ref{lem:zhang}, we have
\begin{align}
U=&\frac1{n_1^2}\sum_{1\leqslant i<j\leqslant n}\{(a_i-\bar a)(b_i-\bar b)-(a_j-\bar a)(b_j-\bar b)\}^2 \nonumber\\
=&\frac{n}{n_1^2}\sum_{i=1}^n\Big\{(a_i-\bar a)(b_i-\bar b)-\frac1n\sum_{j=1}^n(a_j-\bar a)(b_j-\bar b)\Big\}^2 \nonumber\\
\leqslant&\frac{n}{n_1^2}\sum_{i=1}^n\{(a_i-\bar a)(b_i-\bar b)\}^2 \nonumber\\
=&\frac{1}{\ec^2n}\sum_{i=1}^n(a_i-\bar a)^2(b_i-\bar b)^2. \label{eq:Ubound}
\end{align}
By (\ref{eq:Vij}), Minkowski's inequality, and Lemma \ref{lem:zhang}, we have
\begin{align}
V\leqslant&\frac{1}{n_1^2}\sum_{1\leqslant i<j\leqslant n}\Big(|a_{j}-a_{i}|\sqrt{\frac{2}{\ec n}\sum_{i=1}^n(b_i-\bar b)^2}+|b_{j}-b_{i}|\sqrt{\frac{2}{\ec n}\sum_{i=1}^n(a_i-\bar a)^2}\Big)^2\nonumber\\
\leqslant&\frac{1}{n_1^2} \Bigg(\sqrt{\frac{2}{\ec n}\sum_{i=1}^n(b_i-\bar b)^2\sum_{1\leqslant i<j\leqslant n}(a_i-a_j)^2}
+\sqrt{\frac{2}{\ec n}\sum_{i=1}^n(a_i-\bar a)^2\sum_{1\leqslant i<j\leqslant n}(b_i-b_j)^2}\Bigg)^2
\nonumber\\
=& \frac{8}{\ec^3n^2}\sum_{i=1}^n(a_i-\bar a)^2\sum_{i=1}^n(b_i-\bar b)^2.\label{eq:Vbound}
\end{align}
Combining (\ref{eq:dsab}), (\ref{eq:Ubound}), and (\ref{eq:Vbound}), we have
$$
\|\nabla s_{ab}(\z)\|^{2}_2 \leqslant
\Bigg\{\sqrt{ \frac{1}{\ec^2n}\sum_{i=1}^n(a_i-\bar a)^2(b_i-\bar b)^2
} + \sqrt{
\frac{8}{\ec^3n^2}\sum_{i=1}^n(a_i-\bar a)^2\sum_{i=1}^n(b_i-\bar b)^2
}\Bigg\}^2
=:\scov^2.
$$
Then, the conclusion follows from Lemma \ref{lem:bobkov}.
\end{proof}

% \begin{remark}
% \citet{massart1986} first established Lemma \ref{lem:expmean} with a better constant. His proof assumed that $n/n_1$ is an integer. \citet{bloniarz2015lasso} generalized Massart's result, allowing $n/n_1$ to be a non-integer. \citet{tolstikhin2017} proved Lemma~\ref{lem:expmean} based on Bobkov's approach. The bound for the sample covariance in Lemma~\ref{lem:expcov} is novel, to the best of our knowledge.
% \end{remark}

\subsection{Concentration inequalities for stratified randomization}

\citet{massart1986,bloniarz2015lasso}, and \citet{tolstikhin2017} established concentration inequalities for the sample mean under simple random sampling without replacement. We apply Lemmas~\ref{lem:expmean} and \ref{lem:expcov} in each block to obtain concentration inequalities for \textit{the weighted sample mean and sample covariance} under \textit{stratified random sampling without replacement}. These novel inequalities hold for a wide range of number of blocks, block sizes, and propensity scores.

\begin{theorem}\label{prop:conmean}
Consider a sequence of real numbers $ \{a_1,\dots,a_n\}$. For any $t>0$,
$$
\pr\Big( \summ \pim ( \bar{a}_{[m]1} - \bar{a}_{[m]} ) \geqslant t \Big)
\leqslant
\exp\Big\{-\frac{nt^2}{4\sa}\Big\},
$$
where $ \sa = (1/n) \summ \sumim ( a_i - \bar{a}_{[m]} )^2 /  \emt^2 .$
\end{theorem}

\begin{proof}
For any $\lambda>0$ and $t >0$, by Markov's inequality, we have
\begin{align*}
\pr\Big( \summ \pim (\bar{a}_{[m]1} - \bar{a}_{[m]}) \geqslant t \Big)
& \leqslant \exp\{-\lambda t\} \cdot E \exp\Big\{\lambda \summ \pim (\bar{a}_{[m]1} - \bar{a}_{[m]})\Big\}\\
& = \exp\{-\lambda t\} \cdot \prod^M_{m=1}  E \exp\Big\{\lambda \pim (\bar{a}_{[m]1} - \bar{a}_{[m]})\Big\}.
\end{align*}
By Lemma~\ref{lem:expmean}, we have
\begin{align*}
\prod^M_{m=1}  E \exp\Big\{\lambda \pim (\bar{a}_{[m]1} - \bar{a}_{[m]})\Big\}
& \leqslant  \prod^M_{m=1} \exp\Big\{\frac{\lambda^2 \pim^2}{\emt^2\nm^2} \sumim ( a_i - \bar{a}_{[m]} )^2\Big\}\\
&=\exp\Big\{\frac{\lambda^2}{n^2}\sum^M_{m = 1} \sumim ( a_i - \bar{a}_{[m]} )^2/\emt^2 \Big\}\\
&=\exp\Big\{\frac{\lambda^2}{n}\sa\Big\}.
\end{align*}
Thus,
\begin{align*}
\pr\Big( \summ \pim (\bar{a}_{[m]1} - \bar{a}_{[m]}) \geqslant t \Big)
\leqslant \exp\Big\{-\lambda t+\frac{\lambda^2}{n}\sa\Big\}.
\end{align*}
The conclusion follows by taking $\lambda = nt/(2\sa)$.
\end{proof}

\begin{theorem}\label{prop:concov}
Consider two sequences of real numbers $ \{a_1,\dots,a_n\}$ and $\{b_1,\dots,b_n\}$. If $\nmt\geqslant2$ for all $m$, then, for any $t>0$,
$$
\pr\Big( \summ \pim ( s_{[m]ab} - S_{[m]ab} ) \geqslant t \Big)
\leqslant
\exp\Big\{-\frac{nt^2}{60(\fma\fmb)^{1/2}}\Big\},
$$
where $\fma = (1/n)\summ\sumim(a_i-\bar{a}_{[m]})^4/\emt^3$ and $\fmb = (1/n)\summ\sumim(b_i-\bar{b}_{[m]})^4/\emt^3$.
\end{theorem}

\begin{proof}
We denote
\begin{align*}
\smcov^2=&\Bigg\{\sqrt{\frac{1}{\emt^{2}\nm}  \sumim(a_i-\bar{a}_{[m]})^2(b_i-\bar{b}_{[m]})^2 } \\
&+\sqrt{\frac{8}{\emt^3\nm^2}\sumim(a_i- \bar{a}_{[m]})^2\sumim(b_i- \bar{b}_{[m]})^2 }\Bigg\}^2.
\end{align*}
For any $\lambda>0$ and $t >0$, by Markov's inequality and Lemma~\ref{lem:expcov}, we have
\begin{align}
\pr\Big( \summ \pim ( s_{[m]ab} - S_{[m]ab} ) \geqslant t \Big)
& \leqslant \exp\{-\lambda t\} \cdot E \exp\Big\{\lambda \summ \pim (s_{[m]ab} - S_{[m]ab})\Big\}\nonumber\\
& = \exp\{-\lambda t\} \cdot \prod^M_{m=1}  E \exp\Big\{\lambda \pim (s_{[m]ab} - S_{[m]ab})\Big\}\nonumber\\
& \leqslant \exp\{-\lambda t\} \cdot \prod^M_{m=1} \exp\Big\{\frac{\lambda^2 \pim^2}{\nm} \smcov^2\Big\}\nonumber\\
&=\exp\Big\{-\lambda t+\frac{\lambda^2}{n}\summ  \pim\smcov^2 \Big\}.\label{eq:blocksab}
\end{align}
By Minkowski's inequality, we have
\begin{align}
\summ\pim\smcov^2
=&\summ\Bigg\{\sqrt{\frac{\pim}{\emt^{2}\nm}  \sumim(a_i-\bar{a}_{[m]})^2(b_i-\bar{b}_{[m]})^2 } \nonumber\\
&+\sqrt{\frac{8\pim}{\emt^3\nm^2}\sumim(a_i-\bar{a}_{[m]})^2\sumim(b_i-\bar{b}_{[m]})^2 }\Bigg\}^2 \nonumber\\
\leqslant& \Bigg\{\sqrt{\frac{1}{n}\summ\sumim(a_i-\bar{a}_{[m]})^2(b_i-\bar{b}_{[m]})^2/\emt^{2} } \nonumber\\
&+ \sqrt{\frac{8}{n}\summ\sumim(a_i-\bar{a}_{[m]})^2\sumim(b_i-\bar{b}_{[m]})^2 /(\emt^3\nm) }\Bigg\}^2.\label{eq:scov0}
\end{align}
Then, we deal with the two terms in (\ref{eq:scov0}) separately. By Cauchy--Schwarz inequality,
\begin{align}
&\frac{1}{n}\summ  \sumim(a_i-\bar{a}_{[m]})^2(b_i-\bar{b}_{[m]})^2/\emt^{2}\nonumber\\
\leqslant&\Big\{\frac1n\summ\sumim(a_i-\bar{a}_{[m]})^4/\emt^2\Big\}^{1/2}\Big\{\frac1n\summ\sumim(b_i-\bar{b}_{[m]})^4/\emt^2\Big\}^{1/2}\nonumber\\
\leqslant&\Big\{\frac1n\summ\sumim(a_i-\bar{a}_{[m]})^4/\emt^3\Big\}^{1/2}\Big\{\frac1n\summ\sumim(b_i-\bar{b}_{[m]})^4/\emt^3\Big\}^{1/2}.\label{eq:scov1}
\end{align}
Applying Cauchy--Schwarz inequality twice, we have
\begin{align}
&\frac{8}{n}\summ\Big[\sumim(a_i-\bar{a}_{[m]})^2\sumim(b_i-\bar{b}_{[m]})^2 /(\emt^3\nm)\Big]\nonumber\\
\leqslant& \frac{8}{n} \Big\{
\summ\Big[\sumim(a_i-\bar{a}_{[m]})^2\Big]^2/(\emt^3\nm)
\Big\}^{1/2}
\Big\{
\summ\Big[\sumim(b_i-\bar{b}_{[m]})^2\Big]^2/(\emt^3\nm)
\Big\}^{1/2}\nonumber\\
\leqslant& 8 \Big\{
\frac{1}{n}
\summ\sumim(a_i-\bar{a}_{[m]})^4/\emt^3
\Big\}^{1/2}
\Big\{
\frac{1}{n}
\summ\sumim(b_i-\bar{b}_{[m]})^4/\emt^3
\Big\}^{1/2}.\label{eq:scov2}
\end{align}
Recall that
$$
\fma = \frac{1}{n}\summ\sumim(a_i-\bar{a}_{[m]})^4/\emt^3,\quad \fmb = \frac{1}{n}\summ\sumim(b_i-\bar{b}_{[m]})^4/\emt^3.
$$
Combining (\ref{eq:scov0}), (\ref{eq:scov1}), and (\ref{eq:scov2}), we have
\begin{equation}\label{eq:scov3}
    \summ\pim\smcov^2\leqslant(1+2\sqrt{2})^2(\fma\fmb)^{1/2}\leqslant 15(\fma\fmb)^{1/2}.
\end{equation}
Combining (\ref{eq:blocksab}) and (\ref{eq:scov3}), we have
$$
\pr\Big( \summ \pim ( s_{[m]ab} - S_{[m]ab} ) \geqslant t \Big)
\leqslant\exp\Big\{-\lambda t+\frac{15\lambda^2}{n}(\fma\fmb)^{1/2}\Big\}.
$$
The conclusion follows by taking $\lambda=nt/\{30(\fma\fmb)^{1/2}\}$.
\end{proof}

\begin{theorem}\label{th:Vm}
Consider two sequences of real numbers $ \{a_1,\dots,a_n\}$ and $\{b_1,\dots,b_n\}$. We denote $U_m=\left(\bar{a}_{[m] 1}-\bar{a}_{[m]}\right)\left(\bar{b}_{[m] 1}-\bar{b}_{[m]}\right)$. For any $t>0$,
$$
P\left(\sum_{m=1}^{M} \pi_{[m]}\left(U_m-E U_m\right) \geqslant  t\right) \leqslant
\exp \left\{-nt^2/(4\sigma^2_U)\right\},
$$
where $\sigma^2_U=\sum_{m=1}^M\pi_{[m]}(4/e^3_{[m]})(S^2_{[m]a}+S^2_{[m]b})^2$.
\end{theorem}

\begin{proof}
Let $A_m=A_m(z)=\bar{a}_{[m] 1}-\bar{a}_{[m]}$. By Cauchy–Schwarz inequality, we have
$$
\begin{aligned}
|A_m| \leq \frac{1}{n_{[m] 1}} \sum_{i \in[m]} Z_{i}\left|a_{i}-\bar{a}_{[m]}\right|\leqslant\sqrt{\frac{S^2_{[m]a}}{e_{[m]}}}.
\end{aligned}
$$
Let $A_m'=A_m(z^{i,j})=\bar{a}_{[m] 1}'-\bar{a}_{[m]}$. Similarly, we define $B_m$ and $B_m'$. We have
$$
\begin{aligned}
|\nabla U_m|^{2}
=&
\sum_{i:z_i=1} \sum_{j:z_j=0}\left\{A_mB_m-A_m'B_m'\right\}^2 \\
=&
\sum_{i:z_i=1} \sum_{j:z_j=0}\left\{A_m(B_m-B_m')- (A_m'-A_m)B_m'\right\}^2 \\
\leqslant&
\frac{S^2_{[m]a}+S^2_{[m]b}}{e_{[m]}} \sum_{i:z_i=1} \sum_{j:z_j=0}\left\{|B_m-B_m'|+ |A_m'-A_m|\right\}^2\\
=&
\frac{S^2_{[m]a}+S^2_{[m]b}}{e_{[m]}n_{[m]1}^2} \sum_{i:z_i=1} \sum_{j:z_j=0}\left(|a_i-a_j|+ |b_i-b_j|\right)^2 \\
\leqslant&
\frac{S^2_{[m]a}+S^2_{[m]b}}{e_{[m]}n_{[m]1}^2} \sum_{1\leqslant i<j\leqslant n_{[m]}}\left(|a_i-a_j|+ |b_i-b_j|\right)^2 \\
\leqslant&
\frac{S^2_{[m]a}+S^2_{[m]b}}{e_{[m]}n_{[m]1}^2} \left(\sqrt{\sum_{1\leqslant i<j\leqslant n_{[m]}}(a_i-a_j)^2}+ \sqrt{\sum_{1\leqslant i<j\leqslant n_{[m]}}(b_i-b_j)^2}\right)^2 \\
=&
\frac{S^2_{[m]a}+S^2_{[m]b}}{e_{[m]}n_{[m]1}^2} \left(\sqrt{n_{[m]}(n_{[m]}-1)S^2_{[m]a}}+ \sqrt{n_{[m]}(n_{[m]}-1)S^2_{[m]b}}\right)^2 \\
\leqslant& \frac{4}{e^3_{[m]}}(S^2_{[m]a}+S^2_{[m]b})^2=:\sigma^2_{[m]U}.
\end{aligned}
$$

For any $\lambda>0$ and $t>0$, by Markov's inequality and Bobkov's inequality, we have
$$
\begin{aligned}
P\left(\sum_{m=1}^{M} \pi_{[m]}\left(U_m-E U_m\right) \geqslant  t\right)
\leqslant&
e^{-\lambda t} \cdot \prod_{m=1}^{M} E \exp \left\{\lambda \pi_{[m]}\left(U_m-E U_m\right)\right\} \\
\leqslant&
e^{-\lambda t} \cdot \prod_{m=1}^{M} \exp \left\{ \lambda^2 \pi_{[m]}^{2} \frac{\sigma^2_{[m]U}}{n_{[m]}}\right\} \\
=&
\exp \left\{-\lambda t+\sigma^2_U\lambda^2 /n\right\},
\end{aligned}
$$
where the last equality is due to $\sigma^2_U=\sum_{m=1}^M\pi_{[m]}\sigma^2_{[m]U}$. Taking $\lambda=nt/(2\sigma^2_U)$, we have
$$
P\left(\sum_{m=1}^{M} \pi_{[m]}\left(U_m-E U_m\right) \geqslant  t\right) \leqslant \exp \left\{-nt^2/(4\sigma^2_U)\right\}.
$$
\end{proof}

\subsection{Asymptotic theory of $\htu$ under stratified randomization and stratified rerandomization}

In this section, we review some useful results on the asymptotic distributions of $\htu$ under stratified randomization \citep{Liu2019} and stratified rerandomization \citep{Wang2020rerandomization}, respectively. The maximum second moment condition (Conditions~\ref{cond:1} and \ref{cond:3}) used in this section is weaker than the bounded fourth moment condition (Condition~\ref{cond:moment-X}) used in the main text.

\begin{condition}\label{cond:1}
The maximum block-specific squared distance of the potential outcomes satisfies $n^{-1} \max_{m=1,\dots,M} \max_{i \in [m]}\left\{Y_i(z) - \bym{z}\right\}^2 \rightarrow 0,$ for $z = 0,1$.
\end{condition}

\begin{condition}\label{cond:2}
The weighted variances
$\sum^M_{m = 1} \pi_{[m]} \Smyz{1}/\emt$, $\sum^M_{m = 1} \pi_{[m]} \Smyz{0}/(1 - \emt)$, and $\sum^M_{m = 1} \pi_{[m]} \Smytmyc$ tend to finite limits, positive for the first two, and the limit of $\sum^M_{m = 1} \pi_{[m]} \big[ \Smyz{1}/\emt + \Smyz{0}/(1 - \emt) - \Smytmyc \big]$ is strictly positive.
\end{condition}

\begin{proposition}[\citet{Liu2019}]\label{prop:sr}
If Conditions~\ref{cond:propensity}, \ref{cond:1}, and \ref{cond:2} hold, then $\sqrt{n}(\htu - \tau) \stackrel{d}\longrightarrow N(0,\sigma^2_{\unadj})$. Moreover, if Conditions~\ref{cond:propensity} and \ref{cond:1} hold and $\nmz \geqslant2$ for $m=1,\dots,M$ and $z=0,1$, then
$$
\summ \pim \Big\{ \frac{s^2_{[m]Y(z)}}{\emt} \Big\} - \summ \pim \Big\{ \frac{\Smyz{z}}{\emt} \Big\} \stackrel{p}\longrightarrow  0, \quad z = 0, 1.
$$
\end{proposition}

By \citet[][Proposition 2]{Wang2020rerandomization}, the covariance of $\sqrt{n} ( \htu, \htauw ^\T )^\T$ under stratified randomization is
\begin{eqnarray}
V\{ Y, \wm  \} &=& \left( \begin{array}{cc}
    V_{Y Y} & V_{Y \wm}  \\
    V_{\wm Y}  & V_{\wm \wm}
\end{array} \right), \nonumber
\end{eqnarray}
where
$$
V_{Y Y}=\su^2 = \summ \pim \Big\{ \frac{\Smyz{1}}{\emt} + \frac{\Smyz{0}}{1 - \emt} - \Smytmyc \Big\},
$$
$$
V_{\wm \wm}=\summ \pim \Big\{\frac{\Smw }{\emt ( 1 - \emt)}\Big\},
$$
$$
V_{\wm Y}=V_{Y \wm}^\T=\summ \pim \Big\{\frac{ S_{[m]  \wm {Y}(1) }}{\emt} + \frac{ S_{[m]  \wm {Y}(0) }}{1 - \emt}\Big\}.
$$
The asymptotic distribution of $\htu$ under stratified rerandomization depends on the squared multiple correlation between $\htu$ and $\htauw$,
$$
\Ryw =\lim\limits_{n \rightarrow \infty} ( V_{Y \wm}  V_{\wm \wm}^{-1} V_{\wm Y} ) / V_{Y Y}.
$$
Since $V_{YY} = \sigma^2_{\unadj}$, we can conservatively estimate $V_{YY}$ by $\hat \sigma^2_{\unadj}$. In addition, we can consistently estimate $V_{ \wm Y}$ by $\hat V_{ \wm Y } = \summ \pim \left\{  s_{[m] \wm Y(1)} /\emt +  s_{[m] \wm Y(0)} /(1 - \emt) \right\}$ and directly calculate $V_{\wm \wm}$ based on $\xr_i$. Then, we can estimate $\Ryw$ by
$$
\hRyw = (\hat V_{ \wm Y }^\T V_{\wm \wm}^{-1} \hat V_{\wm Y} ) / \hat \sigma^2_{\unadj}.
$$

\begin{remark}
When $\nmz = 1$, we define
$$
s_{[m] \wm Y(z)} = \frac{\nm}{ \nmz ( \nm - 1 )} \sumim I(Z_i = z) ( \xr_{i, \mathcal{K}} - \bar{\xr}_{[m], \mathcal{K}} ) Y_i.
$$
\end{remark}

Recall that $\varepsilon_0$, $D_1,\dots,D_{k}$ are independent standard normal random variables and  $L_{k,a} \sim D_1 \mid \sum_{i=1}^{k} D_i^2 \leqslant a $. Let $v_{k,a} = \pr(\chi^2_{k+2} \leqslant a)/\pr( \chi^2_{k} \leqslant a) \in (0,1)$.
We can conservatively estimate the variance of $\htu$ under stratified rerandomization by
$$
\hat \sigma^2_{\unadjM} = \hat \sigma^2_{\unadj} \big\{ 1 - (1 - v_{k,a} )  \hRyw \big\}.
$$

\begin{condition}\label{cond:3}
The maximum block-specific squared distance of the covariates $\xr_i$ satisfies $n^{-1} \max_{m=1,\dots,M} \max_{i \in [m]} || \xr_i - \bar{\xr}_{[m]} ||_{\infty}^2 \rightarrow 0$.
\end{condition}

\begin{condition}\label{cond:4}
The weighted covariances $\summ \pim \Smw / \emt$, $\summ \pim \Smw /( 1 - \emt )$, $\summ \pim  S_{[m] \wm Y(1) } / \emt $, and $\summ \pim  S_{[m] \wm Y(0) } /( 1 - \emt )$ tend to finite limits, and the limit of $V_{ \wm \wm }$ is strictly positive definite.
\end{condition}

\begin{condition}\label{cond:5}
There exists a constant $C$ such that $n^{-1}\sumi Y^2_i(z)\leqslant C$, $z=0,1$.
\end{condition}

\begin{proposition}[\cite{Wang2020rerandomization}]
\label{prop:srr}
If Conditions~\ref{cond:propensity}, \ref{cond:1}--\ref{cond:4} hold, then, for  fixed $a>0$, $\pr(\mathcal{M}_a) \rightarrow \pr(\chi^2_{k} \leqslant a)$,
$$
\{ \sqrt{n} (\htu - \tau) \mid \mathcal{M}_a \} \stackrel{d}\longrightarrow  \sigma_{\unadj} \Big\{ \sqrt{1 - R_{Y, \wm}^2} \varepsilon_0 + \sqrt{R_{Y, \wm}^2} L_{k,a}  \Big\},
$$
and the asymptotic variance of $\sqrt{n} (\htu - \tau) \mid \mathcal{M}_a$  is
$$
\sigma^2_{\unadjM} = \sigma^2_{\unadj} \big\{ 1 - (1 - v_{k,a} )  \Ryw \big\} \leqslant \sigma^2_{\unadj}.
$$
Furthermore, if Condition~\ref{cond:5} holds, then
$$
\begin{aligned}
\hat \sigma^2_{\unadjM} \stackrel{p}\longrightarrow \sigma^2_{\unadjM} +
\lim\limits_{n \rightarrow \infty}  \Bigg\{&\sum_{m\in\mathcal{A}_c} \pim S^2_{[m]\{ Y(1) - Y(0) \}}+
\\&\left(\frac{n_f}{n}\right)^2\frac{n}{n_f+\sum_{m\in\mathcal{A}_f}\omega_{[m]}}\sum_{m\in\mathcal{A}_f}\omega_{[m]}({\tau}_{[m]}-{\tau}_{f})^2
\Bigg\}.
\end{aligned}
$$
\end{proposition}

Proposition~\ref{prop:srr} shows that the asymptotic distribution of $\htu$ under stratified rerandomization is a convolution of a normal distribution and a truncated normal distribution, and its asymptotic variance is less than or equal to that of $\htu$ under stratified randomization. Moreover, we can conservatively estimate the asymptotic variance.

\section{Proofs of Theorems}
\label{sec:proof-main}

\zhu{We first prove theoretical results under the special case of stratified randomization (Propositions~\ref{thm:sr2} and \ref{thm:sr2var}). Then we extend the proofs to the general scenario of stratified rerandomization (Theorems~\ref{thm:srr2} and \ref{thm:srr2var}).}

\begin{proposition}
\label{thm:sr2}
Suppose that Conditions~\ref{cond:propensity}--\ref{cond:tuning} hold. Under stratified randomization, we have $\sqrt{n} ( \htlt - \tau ) \stackrel{d}\longrightarrow N(0, \slt^2)$.
Furthermore, $\htlt$ is asymptotically more efficient than $\htu$, that is,
$$
\slt^2 - \su^2  = - \lim\limits_{n\rightarrow \infty} \g^\T \bigg\{ \summ \pim \frac{\Smx}{\emt ( 1 - \emt ) }  \bigg\} \g \leqslant 0.
$$
\end{proposition}

\begin{proposition}
\label{thm:sr2var}
Suppose that Conditions~\ref{cond:propensity}--\ref{cond:tuning} and \ref{cond:largest-RE} hold. Under stratified randomization, $\hslt^2$ converges in probability to
$$
\slt^2 + \lim\limits_{n \rightarrow \infty}  \left\{\sum_{m\in\mathcal{A}_c} \pim\Smestmesc +
\left(\frac{n_f}{n}\right)^2\frac{n}{n_f+\sum_{m\in\mathcal{A}_f}\omega_{[m]}}\sum_{m\in\mathcal{A}_f}\omega_{[m]}({\tau}_{[m]}-{\tau}_{f})^2
\right\},
$$
which is no less than $\slt^2$ and no greater than the probability limit of $\hat \sigma^2_{\unadj}$.
\end{proposition}

\subsection{Proof of Proposition~\ref{thm:sr2}}

\begin{proof}

We first prove the asymptotic normality of $\htlt$.
By definition, we have
\begin{align*}
\htlt - \tau
=&\htu - \tau - \htaux^\t \hg\\
=& \summ \pim \bigg[\big\{ \bymz{1} - \bym{1} -\big( \bxmz{1} - \bxm \big)^\t \hg \big\} \\
& \quad \quad - \big\{ \bymz{0} - \bym{0} -\big( \bxmz{0} - \bxm \big)^\t \hg  \big\} \bigg].
\end{align*}
Recall the decomposition of the potential outcomes,
\begin{equation}\label{eq thm 4 *A}
Y_i(z) = \bar Y_{[m]}(z) + \big( \bx_i - \bxm \big)^\t
\g + \e^*_i(z), \quad i \in [m],\quad z=0,1,
\end{equation}
we have
\begin{align*}
\htlt - \tau = & \summ \pim \big( \bar \e^*_{[m]1} - \bar \e^*_{[m]0} \big) + \summ \pim \big( \bxmz{1} - \bxm \big)^\t \big(\g - \hg \big)\\
& - \summ \pim \big( \bxmz{0} - \bxm \big)^\t \big(\g - \hg \big).
\end{align*}
Applying Proposition~\ref{prop:sr} to $\e^*_i(z)$, we have
\begin{equation}\label{eq thm 4 **1}
\sqrt{n} \summ \pim \big( \bar \e^*_{[m]1} - \bar \e^*_{[m]0}\big) \stackrel{d}\longrightarrow N(0,\slt^2).
\end{equation}
It suffices for the asymptotic normality of $\htlt$ to show that
\begin{equation}\label{eq thm 4 **2}
\sqrt{n}\summ \pim \big(\bxmz{1} - \bxm\big)^\t \big( \g - \hg \big) \stackrel{p}\longrightarrow 0,
\end{equation}
\begin{equation}\label{eq thm 4 **3}
\sqrt{n} \summ \pim \big(\bxmz{0} - \bxm\big)^\t \big( \g - \hg \big) \stackrel{p}\longrightarrow 0.
\end{equation}
By H$\ddot{o}$lder inequality,
\begin{align*}
\Big|\summ \pim \big(\bxmz{1} - \bxm\big)^\t \big( \g - \hg \big) \Big| & \leqslant \Big\|\summ \pim \big(\bxmz{1} - \bxm\big)\Big\|_\infty \cdot \Big\|\g - \hg\Big\|_1.
\end{align*}

%%%%%%%%%%%%%%

To bound $\|\summ \pim \big(\bxmz{1} - \bxm\big)\|_\infty$ and $\|\g - \hg \|_1 $, we have the following two lemmas with their proofs given later.

\begin{lemma} \label{lem:concen-meanx}
If Conditions \ref{cond:propensity} and \ref{cond:moment-X} hold, then
$$
\Big\|\summ \pim \big(\bxmz{1} - \bxm\big)\Big\|_\infty = \Op\bigg(\sqrt{\dfrac{\log p}{n}}\bigg).
$$
\end{lemma}

\begin{lemma} \label{lem:ghat}
If Conditions \ref{cond:propensity}, \ref{cond:moment-X}, \ref{cond:re}, and \ref{cond:tuning} hold, then $\big\|\gz{z} - \hgz{z}\big\|_1 = \Op(s \lambda_z)$, $z=0,1$.
\end{lemma}

By Lemma~\ref{lem:ghat}, $\big\|\g - \hg \big\|_1 = \Op(s \lambda_1 + s \lambda_0)$. Therefore,
\begin{align*}
& \sqrt{n} \Big|\summ \pim \big(\bxmz{1} - \bxm\big)^\t \big( \g - \hg \big) \Big| \\
& \leqslant  \sqrt{n} \Big\|\summ \pim \big(\bxmz{1} - \bxm\big)\Big\|_\infty \cdot \Big\|\g - \hg\Big\|_1\\
& = \Op\Big( \sqrt{n} \cdot \sqrt{ \dfrac{\log p}{n} } \Big) \cdot \Op(s  \lambda_1 + s \lambda_0 ) \\
& = \op(1),
\end{align*}
where the last equality is because of Condition~\ref{cond:tuning}. Thus, Statement~\eqref{eq thm 4 **2} holds. Similarly, Statement~\eqref{eq thm 4 **3} holds. Combining \eqref{eq thm 4 **1}--\eqref{eq thm 4 **3}, we obtain the asymptotic normality of $\htlt$.

Next, we compare the asymptotic variances $\su^2$ and $\slt^2$. By the decomposition in \eqref{eq thm 4 *A}, we have
\begin{equation}\label{eq thm 4 5.1}
\Smesz{1} = \Smyz{1} + \g^\t \Smx \g - 2 \Smxyz{1}^\t\g.
\end{equation}
Similarly, we have
\begin{equation}\label{eq thm 4 5.2}
\Smesz{0} = \Smyz{0} + \g^\t \Smx \g - 2 \Smxyz{0}^\t\g,
\end{equation}
and
$$
\begin{aligned}
\Smestmesc =& \Smytmyc+ (\g-\g)^\t \Smx (\g-\g) \\
& - 2 \{\Smxyz{1}-\Smxyz{0}\}^\t(\g-\g)\\
=& \Smytmyc.
\end{aligned}
$$
Thus,
\begin{align*}
\su^2  = & \lim_{n \rightarrow \infty} \summ \pim \Big[ \dfrac{\Smyz{1}}{\emt} + \dfrac{\Smyz{0}}{1 - \emt }  - \Smytmyc\Big]\\
 =&  \slt^2 - \lim_{n \rightarrow \infty} \Bigg[ \g^\t \Big\{ \summ \pim  \dfrac{\Smx}{ \emt (1 -\emt)} \Big\} \g\\
& + 2 \summ \pim \Big\{ \dfrac{\Smxyz{1}}{\emt} + \dfrac{\Smxyz{0}}{1 - \emt} \Big\}^\t \g \Bigg]\\
 = & \slt^2 + \lim_{n \rightarrow \infty} \g^\t \Big\{ \summ \pim \dfrac{\Smx}{ \emt ( 1 - \emt) } \Big\} \g,
\end{align*}
where the last equality is because of the definition of $\g$:
$$
\gs =  \bigg\{ \summ \pim \frac{\Smxs}{\emt ( 1 - \emt) } \bigg\}^{-1} \bigg\{ \summ \pim \frac{\Smxsyz{1}}{\emt } + \summ \pim \frac{\Smxsyz{0}}{ 1 - \emt } \bigg\},
$$
and $\gsc = \bb{0}$.
\end{proof}

\subsection{Proof of Proposition~\ref{thm:sr2var}}

\begin{proof}

We first introduce a lemma that bounds the number of the covariates selected by Lasso and will be proved in Section \ref{sec:proof-shat-gamma}.

\begin{lemma}\label{lem:shat-gamma}
If Conditions~\ref{cond:propensity}, \ref{cond:moment-X}, and \ref{cond:re}--\ref{cond:largest-RE}
hold, then there exists a constant $C$ independent of $n$, such that the following holds with probability tending to one,
\[\left\|\hgz{1}\right\|_0 \leqslant C  s,~~  \left\|\hgz{0}\right\|_0  \leqslant C   s.\]
\end{lemma}
% Since Conditions~\ref{cond:propensity}--\ref{cond:moment-X}, \ref{cond:re}--\ref{cond:tuning}, and \ref{cond:largest-RE}
% hold for $\e_i^\omega(z)$ and the transformed covariates $\xr_i^\omega$, we apply Lemma \ref{lem:shat} to $\hgz{z}$ and obtain
% $$
% \left\|\hgz{z}\right\|_0\leqslant C s, \quad z=0,1,
% $$
% which leads to
By Lemma \ref{lem:shat-gamma}, we have
$$
\hat s  = \left\|\hgz{1}+\hgz{0}\right\|_0 = \Op(s).
$$
Then, $n/(n-\hat{s}) \stackrel{p}\longrightarrow 1 $. Thus, we only need to derive the probability limit of
\begin{equation}\label{eq:var_fine}
\begin{aligned}
\sum_{m\in\mathcal{A}_c} \pim \Big\{ \frac{ s^2_{[m]\tY(1)}}{\emt}  + \frac{ s^2_{[m]\tY(0)}}{1 - \emt} \Big\}+
\left(\frac{n_f}{n}\right)^2\frac{n}{n_f+\sum_{m\in\mathcal{A}_f}\omega_{[m]}}\sum_{m\in\mathcal{A}_f}\omega_{[m]}(\hat{\tau}_{\tY,[m]}-\hat{\tau}_{\tY,f})^2.
\end{aligned}
\end{equation}

% \[  \summ \pim \dfrac{\hat s^2_{[m]\e^*(1)}}{\emt} - \summ \pim \dfrac{\Smesz{1}}{\emt}  \stackrel{p}\longrightarrow 0, ~~ \summ \pim \dfrac{\hat s^2_{[m]\e^*(0)}}{1 - \emt} - \summ \pim \dfrac{\Smesz{0}}{1 - \emt}  \stackrel{p}\longrightarrow 0.\]
% Again, we will only prove the first statement.

\textbf{Step 1 (Coarse blocks)}: We derive the probability limit of the first term in \eqref{eq:var_fine}. By the definition of $\tY_i=Y_i- (\xr_i-\bxm)^\T \hg$ and the decomposition \eqref{eq thm 4 *A}, we have
\begin{align*}
s^2_{[m]\tY(1)} & = \dfrac{1}{\nmt -1 } \sumim Z_i \big\{ Y_i(1) - \bymz{1} - (\xr_i - \bx_{[m]1})^\t \hg \big\}^2\\
& = s^2_{[m]\e^*(1)} + \big( \g - \hg \big)^\t \smx \big( \g - \hg \big) + 2 \smxesz{1}^\t \big( \g - \hg \big),
\end{align*}
where $\smx=s^2_{[m]\xm(1)}$ stands for the sample covariance of $\xm$ under treatment. In the following, we will use this simplified notation if there is no exceptional clarity.
Therefore,
\begin{align}\label{eqn::sme1}
\sum_{m\in\mathcal{A}_c} \pim  \dfrac{s^2_{[m]\tY(1)}}{\emt} = \sum_{m\in\mathcal{A}_c} \pim \dfrac{s^2_{[m]\e^*(1)}}{\emt} & + \big( \g - \hg \big)^\t \bigg( \sum_{m\in\mathcal{A}_c} \pim  \dfrac{\smx}{\emt} \bigg) \big( \g - \hg \big) \nonumber \\
& + 2 \sum_{m\in\mathcal{A}_c}\pim  \dfrac{\smxesz{1}^\t}{\emt} \big( \g - \hg \big).
\end{align}

For the first term on the right hand of \eqref{eqn::sme1}, applying Proposition~\ref{prop:sr} to $\e^*_i(1)$, we have
\[ \sum_{m\in\mathcal{A}_c} \pim \dfrac{s^2_{[m]\e^*(1)}}{\emt} - \sum_{m\in\mathcal{A}_c} \pim \dfrac{S^2_{[m]\e^*(1)}}{\emt} \stackrel{p}\longrightarrow 0. \]

For the second term on the right-hand side of \eqref{eqn::sme1}, by Condition~\ref{cond:propensity}, we have
\begin{align*}
& \big( \g - \hg \big)^\t \bigg( \sum_{m\in\mathcal{A}_c} \pim  \dfrac{\smx}{\emt} \bigg) \big( \g - \hg \big) \\
\leqslant &\dfrac{1}{\propc} \big( \g - \hg \big)^\t \bigg( \sum_{m\in\mathcal{A}_c} \pim  \smx \bigg) \big( \g - \hg \big)\\
\leqslant &\dfrac{1}{\propc} \| \g - \hg \|^2_1 \cdot \big \| \sum_{m\in\mathcal{A}_c} \pim  \smx \big \|_\infty,
\end{align*}
where $\| H\|_\infty=\max _{i j}\left|h_{i j}\right|$ denotes the max norm of matrix $H$.
By the fourth moment condition of the covariates (see Condition~\ref{cond:moment-X}), we have $\| \sum_{m\in\mathcal{A}_c} \pim  \Smx \|_\infty \leqslant C$ for some constant $C$. Similar to the proof of Lemma~\ref{lem:concen-covx-gamma}, we can show that
$$
\big \| \sum_{m\in\mathcal{A}_c} \pim  \smx  - \sum_{m\in\mathcal{A}_c} \pim  \Smx \big\|_\infty = \Op ( \sqrt{(\log p)/n} ).
$$
Thus,
$
\big \| \sum_{m\in\mathcal{A}_c} \pim  \smx \big \|_\infty =  \Op(1).
$
By Lemma~\ref{lem:ghat} and Condition~\ref{cond:tuning}, we have
$ \| \g - \hg \|_1  = \op(1).$
Therefore,
\begin{align*}
& \big( \g - \hg \big)^\t \bigg( \sum_{m\in\mathcal{A}_c} \pim  \dfrac{\smx}{\emt} \bigg) \big( \g - \hg \big) \stackrel{p}\longrightarrow 0.
\end{align*}

%For the second term on the right-hand side of \eqref{eqn::sme1}, applying Proposition~\ref{prop:sr} to $\e_i(1)$, we have
%\[ \summ \pim \dfrac{ \smez{1}}{\emt} - \summ \pim %\dfrac{\Smez{1}}{\emt}  \stackrel{p}\longrightarrow 0. \]

The third term on the right-hand side of \eqref{eqn::sme1} tends to zero in probability by Cauchy-Schwarz inequality. Therefore,
$$
\sum_{m\in\mathcal{A}_c} \pim  \dfrac{s^2_{[m]\tY(1)}}{\emt} - \sum_{m\in\mathcal{A}_c} \pim \dfrac{S^2_{[m]\e^*(1)}}{\emt} \stackrel{p}\longrightarrow 0.
$$
Similarly,
$$
\sum_{m\in\mathcal{A}_c} \pim  \dfrac{s^2_{[m]\tY(0)}}{\emt} - \sum_{m\in\mathcal{A}_c} \pim \dfrac{S^2_{[m]\e^*(0)}}{\emt} \stackrel{p}\longrightarrow 0.
$$

\textbf{Step 2 (Fine blocks)}: We derive the probability limit of the second term in \eqref{eq:var_fine}. By the definition of $\tY_i=Y_i- (\xr_i-\bxm)^\T \hg$ and the decomposition \eqref{eq thm 4 *A}, we have
$$
\begin{aligned}
\frac{1}{\nmz}\sum_{i\in[m], Z_i=z} R_i=&\frac{1}{\nmz}\sum_{i\in[m], Z_i=z} \{Y_i- (\xr_i-\bxm)^\T \hg\}\\
=&\frac{1}{\nmz}\sum_{i\in[m], Z_i=z} \{\bym{z} + (\xr_i-\bxm)^\T (\g-\hg)+\e^*_i(z)\}\\
=&\bym{z} +(\bxmz{z}-\bxm)^\T(\g-\hg)+\bar{\e}^*_{[m]z}.
\end{aligned}
$$
Then, we obtain a key decomposition:
$$
\begin{aligned}
\hat{\tau}_{\tY,[m]}=&\left[\left\{\bar{\e}^*_{[m]1}+\bym{1}\right\} -\left\{\bar{\e}^*_{[m]0}+\bym{0}\right\}\right] +(\bxmz{1}-\bxmz{0})^\T(\g-\hg)\\
=&(\bar{\tY}^*_{[m]1} -\bar{\tY}^*_{[m]0}) +(\bxmz{1}-\bxmz{0})^\T(\g-\hg),
\end{aligned}
$$
where $\tY_i^*(z)=\e^*_i(z)+\bym{z}$.
We denote
$$
\begin{aligned}
\phi_{[m]}=&\frac{1}{\pim} \left(\frac{n_f}{n}\right)^2\frac{n}{n_f+\sum_{m\in\mathcal{A}_f}\omega_{[m]}}\omega_{[m]}\\
=&\frac{1}{\pim}\left(\frac{n_f}{n}\right)^2\frac{n}{n_f+\sum_{m\in\mathcal{A}_f}\{\nm^2/(n_f-2\nm)\}}\cdot\frac{\nm^2}{(n_f-2\nm)}.
\end{aligned}
$$
By the above decomposition and definition, we have
$$
\begin{aligned}
&\left(\frac{n_f}{n}\right)^2\frac{n}{n_f+\sum_{m\in\mathcal{A}_f}\omega_{[m]}}\sum_{m\in\mathcal{A}_f}\omega_{[m]}(\hat{\tau}_{\tY,[m]}-\hat{\tau}_{\tY,f})^2\\
=&\sum_{m\in\mathcal{A}_f}\phi_{[m]}\pim(\hat{\tau}_{\tY,[m]}-\hat{\tau}_{\tY,f})^2\\
=&\sum_{m\in\mathcal{A}_f}\phi_{[m]}\pim(\hat{\tau}_{\tY^*,[m]}-\hat{\tau}_{\tY^*,f})^2+\sum_{m\in\mathcal{A}_f}\phi_{[m]}\pim\Big\{(\bxmz{1}-\bxmz{0})^\T(\g-\hg)\\
&\quad\quad\quad\quad\quad\quad\quad\quad\quad\quad\quad\quad\quad\quad\quad\quad\quad\quad-\sum_{m\in\mathcal{A}_f}\frac{\nm}{n_f}(\bxmz{1}-\bxmz{0})^\T(\g-\hg)\Big\}^2\\
&+2\sum_{m\in\mathcal{A}_f}\phi_{[m]}\pim(\hat{\tau}_{\tY^*,[m]}-\hat{\tau}_{\tY^*,f})\Big\{(\bxmz{1}-\bxmz{0})^\T(\g-\hg)\\
&\quad\quad\quad\quad\quad\quad\quad\quad\quad\quad\quad\quad\quad\quad\quad\quad\quad\quad-\sum_{m\in\mathcal{A}_f}\frac{\nm}{n_f}(\bxmz{1}-\bxmz{0})^\T(\g-\hg)\Big\}.
\end{aligned}
$$
Since $\bar{\e}^*_{[m]}(z)=0$, we have ${\tau}_{\tY^*,[m]}={\tau}_{[m]}$ and $\tau_{\tY^*,f}={\tau}_{f}$. By replacing $Y_i(z)$ with $\tY_i^*(z)$, we apply Proposition \ref{prop:srr} with $a=\infty$ and obtain
$$
\begin{aligned}
\sum_{m\in\mathcal{A}_f}\phi_{[m]}\pim(\hat{\tau}_{\tY^*,[m]}-\hat{\tau}_{\tY^*,f})^2\stackrel{p}\longrightarrow
\lim\limits_{n \rightarrow \infty}\Bigg[
&\sum_{m\in\mathcal{A}_f} \pim \Big\{ \frac{\Smesz{1}}{\emt} + \frac{\Smesz{0}}{1 - \emt} - \Smestmesc  \Big\}\\
&+\left(\frac{n_f}{n}\right)^2\frac{n}{n_f+\sum_{m\in\mathcal{A}_f}\omega_{[m]}}\sum_{m\in\mathcal{A}_f}\omega_{[m]}({\tau}_{[m]}-{\tau}_{f})^2
\Bigg].
\end{aligned}
$$
Next, we show that the following term converges to zero in probability:
$$
\begin{aligned}
&\sum_{m\in\mathcal{A}_f}\phi_{[m]}\pim\Big\{(\bxmz{1}-\bxmz{0})^\T(\g-\hg)-\sum_{m\in\mathcal{A}_f}\frac{\nm}{n_f}(\bxmz{1}-\bxmz{0})^\T(\g-\hg)\Big\}^2\\
=&\sum_{m\in\mathcal{A}_f}\phi_{[m]}\pim\Big\{(\bxmz{1}-\bxmz{0})^\T(\g-\hg)\Big\}^2\\
&+\sum_{m\in\mathcal{A}_f}\phi_{[m]}\pim\Big\{\sum_{m\in\mathcal{A}_f}\frac{\nm}{n_f}(\bxmz{1}-\bxmz{0})^\T(\g-\hg)\Big\}^2\\
&-2 \sum_{m\in\mathcal{A}_f}\phi_{[m]}\pim\Big\{(\bxmz{1}-\bxmz{0})^\T(\g-\hg)\Big\} \Big\{\sum_{m\in\mathcal{A}_f}\frac{\nm}{n_f}(\bxmz{1}-\bxmz{0})^\T(\g-\hg)\Big\}\\
&\equiv A_1+A_2+A_{12}.
\end{aligned}
$$
By Cauchy--Schwarz inequality, we only need to show that $A_1$ and $A_2$ converge to zero in probability.
Note that $(\bxmz{1}-\bxmz{0})=(\bxmz{1}-\bxm)/(1-\emt)$. For $m\in\mathcal{A}_f$, by Condition \ref{cond:propensity},  $\nm$ is bounded, which leads to that $\phi_{[m]}$ has the same order as a constant.
By Lemma~\ref{lem:ghat}, Conditions \ref{cond:propensity}, \ref{cond:moment-X}, and \ref{cond:tuning}, we have
$$
\begin{aligned}
A_1=&(\g-\hg)^\T\left\{
\sum_{m\in\mathcal{A}_f}\phi_{[m]}\pim(\bxmz{1}-\bxmz{0})(\bxmz{1}-\bxmz{0})^\T
\right\}(\g-\hg)\\
=&(\g-\hg)^\T\left\{
\sum_{m\in\mathcal{A}_f}\frac{\phi_{[m]}}{(1-\emt)^2}\pim(\bxmz{1}-\bxm)(\bxmz{1}-\bxm)^\T
\right\}(\g-\hg)\\
=&o_p(1).
\end{aligned}
$$
By Lemmas~\ref{lem:concen-meanx} and \ref{lem:ghat}, Conditions \ref{cond:propensity}, \ref{cond:moment-X}, and \ref{cond:tuning}, we have
$$
\begin{aligned}
A_2=&\Big(\sum_{m\in\mathcal{A}_f}\phi_{[m]}\pim\Big)\cdot\Big\{\sum_{m\in\mathcal{A}_f}\frac{\nm}{n_f(1-\emt)}(\bxmz{1}-\bxm)^\T(\g-\hg)\Big\}^2
=&o_p(1).
\end{aligned}
$$

Combining Step 1 and Step 2, \eqref{eq:var_fine} converges in probability to
$$
\slt^2 + \lim\limits_{n \rightarrow \infty}  \left\{\sum_{m\in\mathcal{A}_c} \pim\Smestmesc +
\left(\frac{n_f}{n}\right)^2\frac{n}{n_f+\sum_{m\in\mathcal{A}_f}\omega_{[m]}}\sum_{m\in\mathcal{A}_f}\omega_{[m]}({\tau}_{[m]}-{\tau}_{f})^2
\right\}.
$$

\textbf{Step 3 ($\hslt^2$ is conservative)}: We compare the limits of $\hslt^2$ and $ \slt^2$. By definition and the above proof, it is easy to see that
$$
\begin{aligned}
&\lim_{n \rightarrow \infty} ( \hslt^2 - \slt^2 )\\ = &\lim\limits_{n \rightarrow \infty}  \left\{\sum_{m\in\mathcal{A}_c} \pim\Smestmesc +
\left(\frac{n_f}{n}\right)^2\frac{n}{n_f+\sum_{m\in\mathcal{A}_f}\omega_{[m]}}\sum_{m\in\mathcal{A}_f}\omega_{[m]}({\tau}_{[m]}-{\tau}_{f})^2
\right\}\\
\geqslant &0.
\end{aligned}
$$

\textbf{Step 4 (Improved efficiency)}: We compare the limits of $\hslt^2$ and $\hsu^2$. By Proposition~\ref{prop:srr} with $a = \infty$,
$$
\begin{aligned}
\hat \sigma^2_{\unadj} \stackrel{p}\longrightarrow \sigma^2_{\unadj} +
\lim\limits_{n \rightarrow \infty}  \Bigg\{&\sum_{m\in\mathcal{A}_c} \pim S^2_{[m]\{ Y(1) - Y(0) \}}+
\\&\left(\frac{n_f}{n}\right)^2\frac{n}{n_f+\sum_{m\in\mathcal{A}_f}\omega_{[m]}}\sum_{m\in\mathcal{A}_f}\omega_{[m]}({\tau}_{[m]}-{\tau}_{f})^2
\Bigg\}.
\end{aligned}
$$
Sine $ S^2_{[m]\{ Y(1) - Y(0) \}} = \Smestmesc$, then
$$
\lim_{n \rightarrow \infty} (\hslt^2 - \hsu^2) = \lim_{n \rightarrow \infty} \Big\{ \summ \pim \dfrac{\Smesz{1}-\Smyz{1}}{\emt} +\summ \pim \dfrac{\Smesz{0}-\Smyz{0}}{1 - \emt} \Big\}.
$$
We have shown in \eqref{eq thm 4 5.1} and \eqref{eq thm 4 5.2} that
\[ \Smesz{z} = \Smyz{z} + \g^\t \Smx \g - 2 \Smxyz{z}^\t \g, ~~ z= 0,1. \]
Therefore,
\begin{align*}
& \lim_{n \rightarrow \infty} (\hslt^2 - \hsu^2) \\
= &\lim_{n \rightarrow \infty} \Big[  \g^\t \summ \pim \cdot \dfrac{\Smx}{\emt(1-\emt)} \g - 2 \summ \pim \Big\{ \dfrac{\Smxyz{1}}{\emt} + \dfrac{\Smxyz{0}}{1-\emt} \Big\}^\t \g \Big]\\
=& - \lim_{n \rightarrow \infty} \g^\t \Big\{ \summ \pim \dfrac{\Smx}{\emt(1-\emt)} \Big\}\g \leqslant 0,
\end{align*}
where the second equality is because of the definition of $\g$.

\end{proof}

\subsection{Proof of Theorem~\ref{thm:srr2}}

\begin{proof}

First, we prove the result on the asymptotic distribution of $\htlt$ under stratified rerandomization.
In the proof of Proposition~\ref{thm:sr2}, we have shown that
\begin{align*}
\htlt - \tau = & \summ \pim \big( \bar \e^*_{[m]1} - \bar \e^*_{[m]0} \big) + \summ \pim \big( \bxmz{1} - \bxm \big)^\t \big(\g - \hg \big)\\
& - \summ \pim \big( \bxmz{0} - \bxm \big)^\t \big(\g - \hg \big).
\end{align*}
Applying Proposition~\ref{prop:srr} to $\e^*_i(z)$, we have
\[ \Big\{ \sqrt{n} \summ \pim  \big( \e^*_{[m]1} - \e^*_{[m]0} \big) \mid \mathcal{M}_a \Big\} \stackrel{d}\longrightarrow  \slt \Big( \sqrt{1 - \Resw} \varepsilon_0 + \sqrt{\Resw} L_{k,a}  \Big), \]
where
$$ \Resw = \lim_{n \rightarrow \infty} \Big( V_{\e^* \wm}  V_{\wm \wm} ^{-1} V_{ \wm \e^*} \Big)/V_{\e^* \e^*},
$$
$$
V_{ \wm \e^*} = \summ \pim  \Big( \dfrac{\Smwesz{1}}{\emt} + \dfrac{\Smwesz{0}}{1 - \emt} \Big).
$$
Recall the definition of $\g$:
$$
\gs =\bigg\{ \summ \pim \frac{\Smxs}{\emt ( 1 - \emt) } \bigg\}^{-1} \bigg\{ \summ \pim \frac{\Smxsyz{1}}{\emt } + \summ \pim \frac{\Smxsyz{0}}{ 1 - \emt } \bigg\},
$$
$\gsc = \bb{0}$, and the decomposition $Y_i(z) = \bym{z} + (\xr_i - \bxm)^\t \g + \e^*_i(z)$, we have
$$
\summ \pim \Big\{ \dfrac{\Smxsesz{1}}{\emt} + \dfrac{\Smxsesz{0}}{1- \emt}\Big\} = \bb{0}.
$$
Since $\mathcal{K} \subset \mathcal{S}$, we have
\[ \summ \pim \Big\{ \dfrac{\Smwesz{1}}{\emt} + \dfrac{\Smwesz{0}}{1- \emt}\Big\} = \bb{0}. \]
Therefore, $V_{\wm \e^*} = 0$. Then, $\Resw = 0$ and
\[ \sqrt{n} \summ \pim \big( \bar{\e}^*_{[m]1} - \bar{\e}^*_{[m]0}\big) \mid \mathcal{M}_a \stackrel{d}\longrightarrow N(0, \slt^2). \]
It suffices for the asymptotic normality of $\htlt$ to show that,
\begin{equation}\label{eq thm 8 *}
\sqrt{n} \summ \pim \big( \bxmz{z} - \bxm \big)^\t \big(\g-\hg \big) \mid \mathcal{M}_a  \stackrel{p}\longrightarrow 0, ~~ z = 0,1.
\end{equation}
By Proposition~\ref{prop:srr}, $\pr(\mathcal{M}_a) \longrightarrow \pr(\chi^2_{k} < a) >0$.
Thus, it suffices for \eqref{eq thm 8 *} to show that, under stratified randomization,
\[ \summ \pim \big( \bxmz{z} - \bxm \big)^\t \big(\g - \hg \big)  \stackrel{p}\longrightarrow 0, \quad z = 0, 1,\]
which hold as shown in the proof of Proposition~\ref{thm:sr2}.

Next, we compare the asymptotic variances. By Proposition~\ref{prop:srr}, $\suM^2 \leqslant \su^2$. Thus, it suffices to show $\slt^2 \leqslant \suM^2$. By Proposition~\ref{prop:srr}, we have
\begin{align*}
\suM^2 & = \su^2 \Big[ 1 - \big\{ 1 - v_{k,a} \big\} \Ryw  \Big]\\
& \geqslant \su^2 (1 - \Ryw)\\
& = \lim_{n \rightarrow \infty} \Big[ \var\big( \htu \big) -  \var\big\{ \proj \big(\htu \mid \htauw \big) \big\} \Big] \\
& = \lim_{n \rightarrow \infty} \var\big\{ \htu - \proj \big(\htu\mid \htauw \big) \big\},
\end{align*}
where $\proj \big(\htu \mid \htauw \big)$ denote the projection (minimizing the variance) of $\htu$ onto $ \htauw$,  and the last but one equality holds because (see \citet{Li2020}):
\[ \Ryw  = \lim_{n \rightarrow \infty} \dfrac{\var\big\{ \proj \big( \htu \mid \htauw \big) \big\}}{\var\big( \htu \big)} = \lim_{n \rightarrow \infty} \dfrac{\var\big\{ \proj \big( \htu \mid \htauw \big) \big\}}{ \su^2 }.\]
%By definition and simple algebra, we have
%\[ \slt^2  =  \lim_{n \rightarrow \infty} \var \Big\{ \htu - \proj \big( \htu \mid \htauxs \big) \Big\}.\]
%Since $\mathcal{K} \subset \mathcal{S}$, we have $ \slt^2 \leqslant \suM^2$.

By definition and the above proof for the asymptotic normality, we have
\begin{align*}
\slt^2 & = \lim_{n \rightarrow \infty} \summ \pim \Big\{ \dfrac{\Smesz{1}}{\emt} + \dfrac{\Smesz{0}}{1- \emt}  - \Smestmesc\Big\}\\
& = \lim_{n \rightarrow \infty} E \Big\{ \htu - \tau - \htaux^\t \g \Big\}^2 \\
& =  \lim_{n \rightarrow \infty} \var \Big\{ \htu -  \proj \big(\htu\mid \htauxs \big)  \Big\}\\
& \leqslant  \lim_{n \rightarrow \infty} \var \big\{ \htu - \proj \big(\htu\mid \htauw \big) \big\} \\
& \leqslant \suM^2,
\end{align*}
where the first inequality is due to $\mathcal{K} \subset \mathcal{S}$.

\end{proof}

\subsection{Proof of Theorem~\ref{thm:srr2var}}

\begin{proof}

By Proposition~\ref{thm:sr2var}, under stratified randomization,
\begin{align*}
  \hslt^2  \stackrel{p}\longrightarrow \slt^2 +  \lim\limits_{n \rightarrow \infty}  \bigg\{ & \sum_{m\in\mathcal{A}_c} \pim\Smestmesc + \\
 &
\left(\frac{n_f}{n}\right)^2\frac{n}{n_f+\sum_{m\in\mathcal{A}_f}\omega_{[m]}}\sum_{m\in\mathcal{A}_f}\omega_{[m]}({\tau}_{[m]}-{\tau}_{f})^2
\bigg\}.
\end{align*}
Since $\pr(\mathcal{M}_a) \rightarrow \pr(\chi^2_{k} \leqslant a) > 0$, then the above statement also holds under stratified rerandomization, i.e., conditional on $\mathcal{M}_a$.

Next, we show that $\hsuM^2 \leqslant \hsu^2$ holds in probability under stratified rerandomization. By Proposition~\ref{prop:srr},
\begin{align*}
    & \hsu^2 \\
    \stackrel{p}\longrightarrow & \sigma^2 _\unadj +  \lim\limits_{n \rightarrow \infty}  \Bigg\{\sum_{m\in\mathcal{A}_c} \pim S^2_{[m]\{ Y(1) - Y(0) \}}+\left(\frac{n_f}{n}\right)^2\frac{n}{n_f+\sum_{m\in\mathcal{A}_f}\omega_{[m]}}\sum_{m\in\mathcal{A}_f}\omega_{[m]}({\tau}_{[m]}-{\tau}_{f})^2
\Bigg\},
\end{align*}
\begin{eqnarray}
&& \hsuM^2 \nonumber \\
& \stackrel{p}\longrightarrow & \suM^2 +  \nonumber \\
& & \lim\limits_{n \rightarrow \infty}  \Bigg\{\sum_{m\in\mathcal{A}_c} \pim S^2_{[m]\{ Y(1) - Y(0) \}}+\left(\frac{n_f}{n}\right)^2\frac{n}{n_f+\sum_{m\in\mathcal{A}_f}\omega_{[m]}}\sum_{m\in\mathcal{A}_f}\omega_{[m]}({\tau}_{[m]}-{\tau}_{f})^2
\Bigg\} \nonumber \\
& = & \sigma^2 _\unadj \big[ 1 - \{ 1 - v_{k,a} \} \Ryw  \big] + \nonumber \\
&&   \lim\limits_{n \rightarrow \infty}  \Bigg\{\sum_{m\in\mathcal{A}_c} \pim S^2_{[m]\{ Y(1) - Y(0) \}}+\left(\frac{n_f}{n}\right)^2\frac{n}{n_f+\sum_{m\in\mathcal{A}_f}\omega_{[m]}}\sum_{m\in\mathcal{A}_f}\omega_{[m]}({\tau}_{[m]}-{\tau}_{f})^2
\Bigg\}  \nonumber \\
& = & \sigma^2 _\unadj  - \big\{ 1- v_{k,a} \big\} \sigma^2_{\unadj} \Ryw +  \nonumber \\
&&  \lim\limits_{n \rightarrow \infty}  \Bigg\{\sum_{m\in\mathcal{A}_c} \pim S^2_{[m]\{ Y(1) - Y(0) \}}+\left(\frac{n_f}{n}\right)^2\frac{n}{n_f+\sum_{m\in\mathcal{A}_f}\omega_{[m]}}\sum_{m\in\mathcal{A}_f}\omega_{[m]}({\tau}_{[m]}-{\tau}_{f})^2
\Bigg\}. \nonumber
\end{eqnarray}
Therefore, with probability tending to one, $\hsuM^2 \leqslant \hsu^2$.

%We have shown in Theorem~\ref{thm:srr1var} that $\hsuM^2 \leqslant \hsu^2$ holds in probability under stratified rerandomization.

Finally, we show that $\hslt^2 \leqslant \hsuM^2$ holds in probability under stratified rerandomization. Under stratified rerandomization, the probability limit of $\hsuM^2$ satisfies:
\begin{align*}
&\lim_{n \rightarrow \infty} \hsuM^2 \\
= &\sigma^2_{\unadjM} +
\lim\limits_{n \rightarrow \infty}  \Bigg\{\sum_{m\in\mathcal{A}_c} \pim S^2_{[m]\{ Y(1) - Y(0) \}}+\left(\frac{n_f}{n}\right)^2\frac{n}{n_f+\sum_{m\in\mathcal{A}_f}\omega_{[m]}}\sum_{m\in\mathcal{A}_f}\omega_{[m]}({\tau}_{[m]}-{\tau}_{f})^2
\Bigg\}. \\
\geqslant &\slt^2 +  \lim\limits_{n \rightarrow \infty}  \left\{\sum_{m\in\mathcal{A}_c} \pim\Smytmyc +
\left(\frac{n_f}{n}\right)^2\frac{n}{n_f+\sum_{m\in\mathcal{A}_f}\omega_{[m]}}\sum_{m\in\mathcal{A}_f}\omega_{[m]}({\tau}_{[m]}-{\tau}_{f})^2
\right\}.
\end{align*}
In the proof of Proposition~\ref{thm:sr2}, we have shown that $\Smestmesc = \Smytmyc$.
Therefore, $\hslt^2 \leqslant \hsuM^2$ holds in probability under stratified rerandomization.

\end{proof}

\subsection{Proof of Lemma~\ref{lem:concen-meanx}}
\label{sec:proof-concen-meanx}

\begin{proof}
For any $t > 0$, we have
\begin{align*}
& \pr\Big(\Big\|\summ \pim \big(\bxmz{1} - \bxm\big)\Big\|_\infty \geqslant t \Big) \leqslant p \cdot \max_{1\leqslant j \leqslant p} \pr\Big( \Big| \summ \pim \big(\bxmztj - \bxmj \big) \Big| \geqslant t \Big).
\end{align*}
Applying Theorem~\ref{prop:conmean} to the $j$th covariate $\xc_j$ (and $-\xc_j$), we have
\begin{align*}
& \pr\Big(\Big\|\summ \pim \big(\bxmz{1} - \bxm\big)\Big\|_\infty \geqslant t \Big)  \leqslant 2 p \cdot \max_{1 \leqslant j \leqslant p} \exp\Big\{ -\dfrac{nt^2}{4 \sigma^2_{\xe,j}}\Big\},
\end{align*}
where $\sigma^2_{\xe,j} = (1/n) \summ \sumim ( \xe_{ij} - \bxmj)^2/ \emt^2 $. By Conditions~\ref{cond:propensity}--\ref{cond:moment-X} and Cauchy-Schwarz inequality, we have,
$$
\sigma^2_{\xe,j}
\leqslant \dfrac{1}{\propc^2} \cdot \dfrac{1}{n} \sum^M_{m = 1} \sum_{i \in [m]} \Big(\xe_{ij} - \bxmj \Big)^2
\leqslant \dfrac{1}{\propc^2} \cdot \bigg\{ \dfrac{1}{n} \sum^M_{m = 1} \sum_{i \in [m]} \Big(\xe_{ij} - \bxmj\Big)^4 \bigg\}^{1/2}
\leqslant \dfrac{ L^{1/2} }{\propc^2}.
$$
Therefore,
\[ \pr\Big(\Big\|\summ \pim \big(\bxmz{1} - \bxm\big)\Big\|_\infty \geqslant t \Big) \leqslant 2 \exp\Big(\log p-\dfrac{\propc^2 n t^2}{4  L^{1/2}}\Big). \]
Taking $t =\sqrt{8 L^{1/2} / \propc^2}\cdot \sqrt{ (\log p) /n}$ gives the result.
\end{proof}

\subsection{Proof of Lemma~\ref{lem:ghat}}
\label{sec:proof-ghat}

Before proving Lemma~\ref{lem:ghat}, we first prove two useful lemmas. Define
$$\ew_i(z) =\yw_i(z) - \bar{Y}^\omega(z)-\{\xwr_i(z)-\bar{\xr}^\omega(z)\}^\T \gz{z}, \quad z=0,1,$$
\begin{align*}
& \Sxw = \summ\pim\frac{\Smx}{\emt(1-\emt)}, \quad \Sxewtw = \summ (\pim - n^{-1}) \emt \Smxwewz{1},
\end{align*}
$$
\hSxw=\frac{1}{n}\summ\sum_{i\in[m]}Z_i(\xwr_i-\bxwz{1})(\xwr_i-\bxwz{1})^\T, \quad
\hSxewtw=\frac{1}{n}\summ\sum_{i\in[m]}Z_i(\xwr_i-\bxwz{1})\{\ew_i(1)-\bewz{1}\}.
$$

\begin{lemma}\label{lem:concen-covx-gamma}
Suppose that Conditions \ref{cond:propensity} and \ref{cond:moment-X} hold, then there exists a constant $C$, such that for the event
$$
\setXwXw = \left\{ \left\|\hSxw - \Sxw\right\|_\infty \leqslant C\sqrt{(\log p)/n} \right\},
$$
we have $\pr(\setXwXw)\rightarrow 1$.
\end{lemma}

\begin{proof}
Since
$$
\hSxw=\frac{1}{n}\summ\sum_{i\in[m]}Z_i(\xwr_i-\bxwz{1})(\xwr_i-\bxwz{1})^\T=\frac{1}{n}\summ\sum_{i\in[m]}Z_i(\xwr_i)(\xwr_i)^\T-\frac{\nt}{n}(\bxwz{1})(\bxwz{1})^\T,
$$
we have
\begin{equation}\label{eq:xxw}
\begin{aligned}
\|\hSxw-\Sxw\|_\infty\leqslant\|\frac{1}{n}\summ\sum_{i\in[m]}Z_i(\xwr_i)(\xwr_i)^\T-\Sxw\|_\infty+
\|\frac{\nt}{n}(\bxwz{1})(\bxwz{1})^\T\|_\infty.
\end{aligned}
\end{equation}
Thus, we only need to bound the above two terms.

To bound the first term in \eqref{eq:xxw}, we have
$$
\begin{aligned}
&E\left\{\frac{1}{n}\summ\sum_{i\in[m]}Z_i(\xwr_i)(\xwr_i)^\T\right\}\\
=&E \left\{\frac{1}{n}\summ\sum_{i\in[m]}Z_i \wnz{1} \wxz{1}  \big\{ \xr_{i} - \bar{\xr}_{[m]} \big\} \big\{ \xr_{i} - \bar{\xr}_{[m]} \big\}^\T\right\}\\
=& \frac{1}{n}\summ\sum_{i\in[m]} \frac{\nmt}{\nm} \frac{\nm^2}{ \emt \nmt (\nm-1) }  \frac{1}{1 - \emt}  \big\{ \xr_{i} - \bar{\xr}_{[m]} \big\} \big\{ \xr_{i} - \bar{\xr}_{[m]} \big\}^\T\\
%=&E\left[\summ\pim\frac{1}{\nmt}\sum_{i\in[m]}Z_i\left\{\frac{\nmt}{\nm}(\xwr_i)(\xwr_i)^\T\right\}\right]\\
=&\summ\pim\frac{\Smx}{\emt(1-\emt)}\\
=&\Sxw .
\end{aligned}
$$
For any $t>0$, applying Theorem~\ref{prop:conmean} to $a_i=(\nmt/\nm)\xwe_{ij}\xwe_{il}$, we have
$$
\begin{aligned}
&\pr\bigg( \big\| \frac{1}{n}\summ\sum_{i\in[m]}Z_i(\xwr_i)(\xwr_i)^\T - \Sxw \big\|_\infty \geqslant t \bigg) \\
=&\pr\bigg( \big\|  \summ\pim\frac{1}{\nmt}\sum_{i\in[m]}Z_i\left\{\frac{\nmt}{\nm}(\xwr_i)(\xwr_i)^\T\right\} - \Sxw \big\|_\infty \geqslant t \bigg) \\
\leqslant& p^2 \max_{1 \leqslant j,l \leqslant p} \pr\bigg( \big| \summ\pim\frac{1}{\nmt}\sum_{i\in[m]}Z_i\left(\frac{\nmt}{\nm}\xwe_{ij}\xwe_{il}\right) - \summ\pim\frac{S_{[m]X_jX_l}}{\emt(1-\emt)} \big|\geqslant t \bigg)\\
\leqslant& 2p^2 \exp\Big\{-\frac{c^6 nt^2}{16L}\Big\},
\end{aligned}
$$
where the last inequality is due to Condition~\ref{cond:moment-X}.

To bound the second term in \eqref{eq:xxw}, we have
$$
\begin{aligned}
E\left\{\bxwz{1}\right\}
=E\left\{\summ\pim\frac{1}{\nmt}\sum_{i\in[m]}Z_i\left(\frac{n}{n_1}\frac{\nmt}{\nm}\xwr_i\right)\right\}=\summ\pim\frac{1}{\nmt}\sum_{i\in[m]}\left(\frac{n}{n_1}\frac{\nmt^2}{\nm^2}\xwr_i\right)=\bb{0},
\end{aligned}
$$
where the last equality is due to $\sum_{i\in[m]}\xwr_i=\bb{0}$. Applying Theorem~\ref{prop:conmean} to $a_i=(n/n_1)(\nmt/\nm)\xwe_{ij}$, we have
$\|\bxwz{1}\|_\infty=O_p(\sqrt{ ( \log{p} ) /n})$. Thus,  $\|(\nt/n)(\bxwz{1})(\bxwz{1})^\T\|_\infty=o_p(\sqrt{ (\log p) /n})$.

The conclusion follows from \eqref{eq:xxw} and the above bounds for the two terms in \eqref{eq:xxw}.
\end{proof}

\begin{lemma}\label{lem:concen-covxe-gamma}
Suppose that Conditions \ref{cond:propensity}, \ref{cond:moment-X}, \ref{cond:re}, and \ref{cond:tuning} hold, then for the event
$$
\setXwe = \left\{ \|\hSxewtw\|_\infty  \leqslant \eta \lambda_1 \right\},
$$
we have $\pr(\setXwe)\rightarrow 1$.
\end{lemma}

\begin{proof}
Since
$$
\begin{aligned}
\hSxewtw=&\frac{1}{n}\summ\sum_{i\in[m]}Z_i(\xwr_i-\bxwz{1})\{\ew_i(1)-\bewz{1}\} \\
=&\frac{1}{n}\summ\sum_{i\in[m]}Z_i\xwr_i\ew_i(1)-\frac{n_1}{n}\bxwz{1}\bewz{1},
\end{aligned}
$$
we have
\begin{equation}\label{eq:xew}
\|\hSxewtw\|_\infty\leqslant
\|\frac{1}{n}\summ\sum_{i\in[m]}Z_i\xwr_i\ew_i(1)\|_\infty +\|\frac{n_1}{n}\bxwz{1}\bewz{1}\|_\infty.
\end{equation}
We bound the above two terms separately.

To bound the first term in \eqref{eq:xew}, we have
\begin{eqnarray*}
E\left\{\frac{1}{n}\summ\sum_{i\in[m]}Z_i\xwr_i\ew_i(1)\right\} &=& \frac{1}{n} \summ \sum_{i\in[m]} \emt \xwr_i(1)\ew_i(1)
= \summ (\pim - n^{-1}) \emt \Smxwewz{1} \\
&=& \Sxewtw .
\end{eqnarray*}
%$$
%E\left\{\frac{1}{n}\summ\sum_{i\in[m]}Z_i\xwr_i\ew_i(1)\right\} = \frac{1}{n} \summ \sum_{i\in[m]} \emt \xwr_i\ew_i(1)
%= \summ (\pim - n^{-1}) \emt \Smxwewz{1} = \Sxewtw .
%$$
%Denote
%$$
%\Sxewtw = \summ (\pim - n^{-1}) \emt \Smxwewz{1}.
%$$
%$$
%\begin{aligned}
%&E\left\{\frac{1}{n}\summ\sum_{i\in[m]}Z_i\xwr_i\ew_i(1)\right\}\\
%=&\frac{1}{n} \summ \sum_{i\in[m]} \emt \xwr_i\ew_i(1) \\
%=&E\left[\summ\pim\frac{1}{\nmt}\sum_{i\in[m]}Z_i\left\{\frac{\nmt}{\nm}\xwr_i\ew_i(1)\right\}\right]\\
%=& \frac{1}{n}\summ (\nm - 1) \emt \Smxwewz{1} .
%=&\summ\pim\frac{\Smxwewz{1}}{\emt}\\
%=&\Sxewtw .
%\end{aligned}
%$$
For any $t >0$, by triangle inequality, we have
\begin{align*}
&\pr\big(\|\frac{1}{n}\summ\sum_{i\in[m]}Z_i\xwr_i\ew_i(1)\|_\infty \geqslant t \big) \\
=& \pr\big( \| \frac{1}{n}\summ\sum_{i\in[m]}Z_i \xwr_i\ew_i(1) - \Sxewtw + \Sxewtw \|_\infty \geqslant t \big)\\
\leqslant& \pr\big( \big\| \frac{1}{n}\summ\sum_{i\in[m]}Z_i\xwr_i\ew_i(1) - \Sxewtw \big\|_\infty \geqslant t - \big\| \Sxewtw \big\|_\infty\big)\\
\leqslant &\pr\big( \big\| \frac{1}{n}\summ\sum_{i\in[m]}Z_i\xwr_i\ew_i(1) - \Sxewtw \big\|_\infty \geqslant t - \delta
_n \big)\\
\leqslant & p \cdot \max_{1 \leqslant j \leqslant p} \pr\big( \big |\frac{1}{n}\summ\sum_{i\in[m]}Z_i\xwe_{ij}\ew_i(1) - \Sxewtw \big | \geqslant t - \delta_n \big) .
\end{align*}
%\begin{align*}
%&\pr\big(\|\frac{1}{n}\summ\sum_{i\in[m]}Z_i\sqrt{\wn\wy}\,\xwr_i\ew_i(1)\|_\infty \geqslant t \big) \\
%=& \pr\big( \| \frac{1}{n}\summ\sum_{i\in[m]}Z_i\sqrt{\wn\wy}\,\xwr_i\ew_i(1) - \Sxewtw + \Sxewtw \|_\infty \geqslant t \big)\\
%\leqslant& \pr\big( \big\| \frac{1}{n}\summ\sum_{i\in[m]}Z_i\sqrt{\wn\wy}\,\xwr_i\ew_i(1) - \Sxewtw \big\|_\infty \geqslant t - \big\| \Sxewtw \big\|_\infty\big)\\
%\leqslant &\pr\big( \big\| \frac{1}{n}\summ\sum_{i\in[m]}Z_i\sqrt{\wn\wy}\,\xwr_i\ew_i(1) - \Sxewtw \big\|_\infty \geqslant t - \delta
%_n \big)\\
%\leqslant & p \cdot \max_{1 \leqslant j \leqslant p} \pr\big( \big |\frac{1}{n}\summ\sum_{i\in[m]}Z_i\sqrt{\wn\wy}\,\xwe_{ij}\ew_i(1) - \Sigma_{X_j\ew(1)}^\omega \big | \geqslant t - \delta_n \big) .
%\end{align*}
Applying Theorem~\ref{prop:conmean} to $a_i = (\nmt/\nm)\xwe_{ij}\ew_i(1)$, we have
\begin{equation*}
\begin{aligned}
\pr\left( \big |\frac{1}{n}\summ\sum_{i\in[m]}Z_i\xwe_{ij}\ew_i(1) - \Sigma_{X_j\ew(1)}^\omega \big | \geqslant t - \delta_n \right)
\leqslant
2\exp\left\{-\frac{c^3 n(t-\delta_n)^2}{8L}\right\},
\end{aligned}
\end{equation*}
where the inequality is due to Condition \ref{cond:moment-X}. Let $t= 4 \sqrt{L (\log{p}) / (c^3n)}+\delta_n$, we have
\begin{equation}\label{eq:Sxew1}
\begin{aligned}
\pr\big(\|\frac{1}{n}\summ\sum_{i\in[m]}Z_i\xwr_i\ew_i(1)\|_\infty \geqslant 4 \sqrt{L (\log{p}) / (c^3n)}+\delta_n \big)\leqslant\frac{2}{p}\rightarrow0.
\end{aligned}
\end{equation}

To bound the second term in \eqref{eq:xew}, we have
$$
\begin{aligned}
E\left\{\bxwz{1}\right\}
=E\left\{\summ\pim\frac{1}{\nmt}\sum_{i\in[m]}Z_i\left(\frac{n}{n_1}\frac{\nmt}{\nm}\xwr_i\right)\right\}=\summ\pim\frac{1}{\nmt}\sum_{i\in[m]}\left(\frac{n}{n_1}\frac{\nmt^2}{\nm^2}\xwr_i\right)=\bb{0},
\end{aligned}
$$
where the last equality is due to $\sum_{i\in[m]}\xwr_i=\bb{0}$. Applying Theorem~\ref{prop:conmean} to $a_i=(n/n_1)(\nmt/\nm)\xwe_{ij}$, we have
$$
\pr ( \|\bxwz{1}\|_\infty \leq C_1 \sqrt{ ( \log{p} ) /n} ) \rightarrow 1
$$
for some constant $C_1 >0$. Moreover, by Condition~\ref{cond:moment-X},
\begin{eqnarray*}
| \bewz{1} | \leq \frac{1}{\nt} \sumi | \ew_i(1) | \leq  \frac{n}{\nt} L^{1/4}.
\end{eqnarray*}
Therefore,
\begin{equation}\label{eq:Sxew2}
 \pr \Big( \|\frac{n_1}{n}\bxwz{1}\bewz{1}\|_\infty \leq  L^{1/4} C_1 \sqrt{ ( \log{p} ) /n} \Big) \rightarrow 1.
\end{equation}
%$$
%\begin{aligned}
%E\left\{\bewz{1}\right\}
%&=E\left\{\summ\pim\frac{1}{\nmt}\sum_{i\in[m]}Z_i\left(\frac{n}{n_1}\frac{\nmt}{\nm}\sqrt{\wn\wy}\,\ew_i\right)\right\}\\
%&=\summ\pim\frac{1}{\nmt}\sum_{i\in[m]}\left(\frac{n}{n_1}\frac{\nmt^2}{\nm^2}\sqrt{\wn\wy}\,\ew_i\right)=0,
%\end{aligned}
%$$
%where the last equality is due to $\sum_{i\in[m]}\sqrt{\wn\wy}\,\ew_i=0$. Applying Theorem~\ref{prop:conmean} to $a_i=(n/n_1)(\nmt/\nm)\sqrt{\wn\wy}\,\ew_{i}$, we have $\|\bewz{1}\|_\infty=O_p(\sqrt{\log{p}/n})$. Thus, we have
%\begin{equation}\label{eq:Sxew2}
%\pr\left(\|\frac{\nt}{n}(\bxwz{1}\bewz{1})\|_\infty\leqslant C\sqrt{\frac{\log{p}}{n}}\right)\rightarrow 1.
%\end{equation}

Combining \eqref{eq:xew}, \eqref{eq:Sxew1}, and \eqref{eq:Sxew2}, there exists a constant $C>0$, such that for $\eta \lambda_1\geqslant C\sqrt{(\log p)/n}+\delta_n$, we have $\pr\Big( \|\hSxewtw\|_\infty \leqslant \eta \lambda_1   \Big) \rightarrow 1$.
\end{proof}

Now, we can prove Lemma~\ref{lem:ghat}.

\begin{proof}[Proof of Lemma~\ref{lem:ghat}]

%The proof is similar to the proof of Lemma S3 in \citet{bloniarz2015lasso}. The difference is that (1) we use the weighted $l_2$ loss function for the Lasso, and (2) we use stratified randomization, which needs new concentration inequalities to handle.
%Recall that,
%\begin{align*}
%& \Sxw = \summ\pim\frac{\Smx}{\emt(1-\emt)},~~
%\hSxw=\frac{1}{n}\summ\sum_{i\in[m]}Z_i(\xwr_i-\bxwz{1})(\xwr_i-\bxwz{1})^\T,\\
%&\Sxewtw =\summ\pim\frac{\Smxewz{1}}{\emt(1-\emt)},~~
%\hSxewtw=\frac{1}{n}\summ\sum_{i\in[m]}Z_i(\xwr_i-\bxwz{1})(\sqrt{\wn\wy}\,\ew_i(1)-\bewz{1}),
%\end{align*}
%where $\bewz{1}=\summ\sum_{i\in[m]}Z_i\sqrt{\wn\wy}\,\ew_i(1)$.

We will only prove the result for $z=1$, as the proof for $z=0$ is similar. Define $\h = \hgz{1} - \gz{1}$. By KKT condition, we have
\begin{equation}\label{KKTgamma}
\frac{1}{\nt}\sum_{i:Z_i = 1}   ( \xwr_{i} - \bxwz{1} )[  \yw_i - \bywz{1}-( \xwr_{i} - \bxwz{1} )^\T \hgz{1}  ] =\lambda_1  \subg,
\end{equation}
where $\subg$ is the sub-gradient of $\|\bgamma\|_1$ taking value at $\bgamma = \hgz{1}$, i.e.,
\begin{align*}
\subg \in \partial \|\bgamma\|_1 \bigg|_{\bgamma = \hgz{1}} ~~\text{with}~~ \left\{
\begin{array}{*{20}{l}}
\subg_j = \sign{(\hgz{1})_j}~\text{for}~j~\text{such that}~(\hgz{1})_j \neq 0\\
\subg_j \in [-1,1], ~~\text{otherwise.}
\end{array}
\right.
\end{align*}
Define
$$
\ew_i(1) = \yw_i(1) - \bar{Y}^\omega(1)-(\xwr_i-\bar{\xr}^\omega)^\T \gz{1}.
$$
Then,
%By the decomposition of $Y_i(1)$, we have
$$
\yw_i(1)-\bywz{1}=\{\xwr_i-\bxwz{1}\}^\T \gz{1}+\{\ew_i(1) - \bewz{1}\}.
$$
% \[Y_i(1) - \bymz{1} = \big(\xr_i - \bxmz{1}\big)^\t \gz{1} + \e_i(1) - \bemz{1}.\]
Plugging in to \eqref{KKTgamma}, we have
\begin{equation}\label{eq:KKT2}
    \hSxw \big\{ \gz{1} - \hgz{1} \big\} + \hSxewtw = \lambda_1 \subg.
\end{equation}
Multiplying both sides of \eqref{eq:KKT2}  by $ - \h^\T = \{ \gz{1} - \hgz{1}  \}^\T$, we have
\[ \h^\t \hSxw \h - \h^\t \hSxewtw  =  \lambda_1 (-\h)^\t \subg \leqslant \lambda_1 \big( \|\gz{1}\|_1 - \|\hgz{1}\|_1 \big), \]
where the last inequality holds because
\[ \{\gz{1} \}^\T \subg \leqslant \|\gz{1}\|_1\|\subg\|_\infty \leqslant \|\gz{1}\|_1 ~~\text{and}~~ \hgz{1}^\t \subg = \|\hgz{1}\|_1. \]
Rearranging and by H$\ddot{o}$lder inequality, we have
\begin{equation}\label{eqn::rearrang-g}
   \h^\t \hSxw \h \leqslant \lambda_1 \big\{  \|\gz{1}\|_1 - \|\hgz{1}\|_1 \big\} + \|\h\|_1 \cdot \|\hSxewtw\|_\infty.
\end{equation}
By the definition of $\h$ and several applications of the triangle inequality, we have
$$
\begin{aligned}
\|\gz{1}\|_1 - \|\hgz{1}\|_1 =&
\|\gzs{1}\|_1 - \|\hgzs{1}\|_1 + \|\gzsc{1}\|_1 - \|\hgzsc{1}\|_1\\
\leqslant& \|\h_{\mathcal{S}}\|_1 + \|\gzsc{1}\|_1 - \{\|\h_{\mathcal{S}^c}\|_1 -\|\gzsc{1}\|_1\} \\
=& \|\h_{\mathcal{S}}\|_1 - \|\h_{\mathcal{S}^c}\|_1 + 2 \|\gzsc{1}\|_1.
\end{aligned}
$$
Therefore, conditional on $\setXwe$ with $\pr(\setXwe) \rightarrow 1$ (Lemma~\ref{lem:concen-covxe-gamma}), we have
\begin{align*}
0 \leqslant \h^\t \hSxw \h & \leqslant \lambda_1\Big( \|\h_{\mathcal{S}}\|_1 - \|\h_{\mathcal{S}^c}\|_1 + 2 \|\gzsc{1}\|_1 + \eta \|\h\|_1 \Big)\\
& \leqslant  \lambda_1\Big\{ (\eta - 1)\|\h_{\mathcal{S}^c}\|_1 + (1 + \eta) \|\h_{\mathcal{S}}\|_1 + 2 \|\gzsc{1}\|_1 \Big\}.
\end{align*}
Then,
\begin{equation} \label{eq *3-g}
(1 - \eta) \|\h_{\mathcal{S}^c}\|_1 \leqslant  (1 + \eta) \|\h_{\mathcal{S}}\|_1 + 2 \|\gzsc{1}\|_1 \leqslant (1 + \eta) \|\h_{\mathcal{S}}\|_1 + 2 s  \lambda_1.
\end{equation}
Recall that $\xi>1$ and $0<\eta<(\xi-1)/(\xi+1)<1$, thus $(1-\eta)\xi-(1+\eta)>0$. To proceed, consider the following two cases:

(1) If $\|\h_{\mathcal{S}}\|_1\leqslant 2s\lambda_1/\{(1-\eta)\xi - (1+\eta)\}$, then by \eqref{eq *3-g}:
$$
\begin{aligned}
\|\h\|_1 =& \|\h_{\mathcal{S}}\|_1 + \|\h_{\mathcal{S}^c}\|_1 \\
\leqslant& \|\h_{\mathcal{S}}\|_1+ \dfrac{1+\eta}{1 - \eta}\|\h_{\mathcal{S}}\|_1 + \dfrac{2s\lambda_1}{1 - \eta} \\
=& \dfrac{1}{1 - \eta} \Big\{ 2\|\h_{\mathcal{S}}\|_1 +2s\lambda_1 \Big\}\\
\leqslant & \dfrac{2s\lambda_1}{1 - \eta} \Big\{ \dfrac{2}{(1-\eta)\xi - (1+\eta)} +1 \Big\}.
\end{aligned}
$$
Then,
\[\|\h\|_1 = \|\hgz{1} - \gz{1}\|_1 = \Op ( s \lambda_1 ).\]

(2) If $2s\lambda_1<\{(1-\eta)\xi - (1+\eta)\}\|\h_{\mathcal{S}}\|_1$, then by \eqref{eq *3-g}, we have
$$
\|\h_{\mathcal{S}^c}\|_1 \leqslant \frac{1+\eta}{1-\eta}\|\h_{\mathcal{S}}\|_1 +\frac{(1-\eta)\xi - (1+\eta)}{1-\eta}\|\h_{\mathcal{S}}\|_1=\xi\|\h_{\mathcal{S}}\|_1.
$$
By Condition~\ref{cond:re}, we have
\begin{equation}\label{eq *5-g}
\|\h\|_1 = \|\h_{\mathcal{S}}\|_1 + \|\h_{\mathcal{S}^c}\|_1 \leqslant (1+\xi) \|\h_{\mathcal{S}}\|_1 \leqslant (1 + \xi)  C  s  \| \Sxw \h \|_\infty.
\end{equation}
Taking the $l_\infty$-norm on both sides of the KKT condition \eqref{eq:KKT2}, we have, conditional on the event $\setXwe$,
\begin{equation}\label{eq *4-g}
\|\hSxw \h\|_\infty \leqslant \lambda_1 + \|\hSxewtw\|_\infty \leqslant (1+\eta) \lambda_1.
\end{equation}
By triangle inequality and H$\ddot{o}$lder inequality, and conditional on the event $\setXwe \cap \setXwXw$, we have
\begin{align*}
s\|  \Sxw \h\|_\infty &
\leqslant s\|\hSxw - \Sxw\|_\infty\|\h\|_1+s\|\hSxw \h\|_\infty\\
&\leqslant  s C\sqrt{(\log p)/n} \|\h\|_1 + s\|\hSxw \h\|_\infty\\
& \leqslant  s C\sqrt{(\log p)/n}\|\h\|_1 + s(1+ \eta) \lambda_1 \\
& = o(1) \| \h\|_1 + s(1+ \eta) \lambda_1,
\end{align*}
where the third inequality is because of inequality \eqref{eq *4-g} and the last equality is because of Condition \ref{cond:tuning}. Combining above inequality and \eqref{eq *5-g}, we obtain
$$
\|\h\|_1 = \|\hgz{1} - \gz{1}\|_1 = \Op(s \lambda_1).
$$
\end{proof}

% \section{Proofs of Lemmas}\label{sec::proof-lemma}
% \label{sec:lm}

% \subsection{Proof of Lemma~\ref{lem:expmean}}

% \subsection{Proof of Lemma~\ref{lem:expcov}}

%\subsection{Proof of Lemma~\ref{lem:concen-covxe-gamma}}
%\label{sec:proof-concen-covxe-gamma}

\subsection{Proof of Lemma~\ref{lem:shat-gamma}}
\label{sec:proof-shat-gamma}

\begin{proof}

We will only prove that $\|\hgz{1}\|_0 \leqslant {C} s$ holds in probability for some constant ${C}$; the proof for $\|\hgz{0}\|_0 \leqslant {C}s$ is similar.
First, we bound $\|\hSxw \big\{ \gz{1} - \hgz{1} \big\}\|_2^2$ from below.
Specifically, we consider the $j$th element of $\hSxw \big\{ \gz{1} - \hgz{1} \big\}$; that is, $\bb{e}^\t_j \hSxw \big\{ \gz{1} - \hgz{1} \big\}$, where $\bb{e}_j$ is a $p$-dimensional vector with one on its $j$th entry and zero on other entries.
By KKT condition, we have shown that (see \eqref{eq:KKT2})
\[ \hSxw \big\{ \gz{1} - \hgz{1} \big\} + \hSxewtw = \lambda_1 \subg.\]
For any $j \in \{1,2,\dots,p\}$ such that $\big( \hgz{1} \big)_j \neq 0$, we have
\[ \big|\bb{e}^\t_j \hSxw \big\{ \gz{1} - \hgz{1} \big\} + \bb{e}^\t_j \hSxewtw \big|= \lambda_1.\]
Conditional on $\setXwe = \Big\{ \|\hSxewtw\|_\infty \leqslant \eta  \lambda_1 \Big\}$ with $\pr(\setXwe) \rightarrow 1$ (Lemma \ref{lem:concen-covxe-gamma}), and by triangle inequality, we have,
\begin{align*}
\big|  \bb{e}^\t_j \hSxw \big\{ \gz{1} - \hgz{1} \big\} \big| \geqslant \lambda_1 - \big| \bb{e}^\t_j \hSxewtw \big| \geqslant \lambda_1 - \|\hSxewtw \|_\infty \geqslant (1 - \eta) \lambda_1.
\end{align*}
Then, summing up the square of the elements of $\hSxw \big\{ \gz{1} - \hgz{1} \big\}$, we have
\begin{align}
\|\hSxw \big\{ \gz{1} - \hgz{1} \big\}\|_2^2=&
\sum_{j=1}^p\big|  \bb{e}^\t_j \hSxw \big\{ \gz{1} - \hgz{1} \big\} \big|^2\nonumber\\
\geqslant& \sum_{j: (\hgz{1})_j \neq 0}\big|  \bb{e}^\t_j \hSxw \big\{ \gz{1} - \hgz{1} \big\} \big|^2\nonumber\\
\geqslant &(1 - \eta)^2 \lambda^2_1 \|\hgz{1}\|_0.\label{eq:bbelow-g}
\end{align}

Second, we bound $\|\hSxw \big\{ \gz{1} - \hgz{1} \big\}\|_2^2$ from above:
\begin{align}
 \|\hSxw \big\{ \gz{1} - \hgz{1} \big\}\|_2^2
\leqslant &\Lambda_{\max} \big( \hSxw \big) \cdot \big\| (\hSxw)^{1/2} \big\{ \gz{1} - \hgz{1} \big\} \big\|^2_2\nonumber\\
= &\Lambda_{\max} \big( \hSxw \big) \cdot  \big\{ \gz{1} - \hgz{1} \big\}^\t \hSxw \big\{ \gz{1} - \hgz{1} \big\}\nonumber\\
= &\Lambda_{\max}  \big( \hSxw \big) \cdot \h^\t \hSxw \h.\label{eq:babove0-g}
\end{align}
We deal with $\h^\t \hSxw \h$ and $\Lambda_{\max}( \hSxw )$ separately.
To bound $\h^\t \hSxw \h$, we have shown that (see \eqref{eqn::rearrang-g}),
$$ \h^\t \hSxw \h \leqslant \lambda_1 \big\{  \|\gz{1}\|_1 - \|\hgz{1}\|_1 \big\} + \|\h\|_1  \|\hSxewtw\|_\infty.
$$
Conditional on $\setXwe = \Big\{ \|\hSxewtw\|_\infty \leqslant \eta  \lambda_1 \Big\}$ with $\pr(\setXwe) \rightarrow 1$ (Lemma \ref{lem:concen-covxe-gamma}), and by triangle inequality, we have
\begin{equation*}
\h^\t \hSxw \h \leqslant \lambda_1(1 + \eta)  \|\h\|_1.
\end{equation*}
According to the proof of Lemma~\ref{lem:ghat}, with probability tending to one, there exists a constant $C$, such that
$$ \|\h\|_1 = \big\| \gz{1} - \hgz{1} \big\|_1  \leqslant C \lambda_1 s.
$$
Therefore, we have
\begin{equation}\label{eq:hSh-g}
\h^\t \hSxw \h \leqslant C  (1+\eta) \lambda^2_1 s.
\end{equation}
To bound $\Lambda_{\max} \big( \hSxw\big)$, we expand its expression as follows:
\begin{equation}\label{LS0-g}
\Lambda_{\max} \big( \hSxw\big)  = \max_{\bb{u}: \|\bb{u}\|_2 = 1 } \bb{u}^\t \hSxw \bb{u}.
\end{equation}
By expanding $\hSxw$, we have
\begin{align*}
\bb{u}^\t \hSxw \bb{u}
\leqslant& \frac{1}{n}\summ\sum_{i\in[m]}Z_i\bb{u}^\T(\xwr_i)(\xwr_i)^\T \bb{u}\\
\leqslant&\frac{1}{n}\summ\bb{u}^\T(\xwr_i)(\xwr_i)^\T \bb{u}\\
=&\bb{u}^\T\left\{\summ\pim\frac{\Smx}{\emt^2(1-\emt)} \right\}\bb{u}\\
\leqslant&\frac{1}{\propc^3}\bb{u}^\T\left\{\summ\pim\Smx\right\}\bb{u},
\end{align*}
where the last inequality is due to Condition \ref{cond:propensity}.
Plugging above inequality into \eqref{LS0-g}, and by Condition \ref{cond:largest-RE}, we have
\begin{equation}\label{eq:LS-g}
\Lambda_{\max} \big( \hSxw\big)
\leqslant \frac{1}{\propc^3}\max_{\bb{u}: \|\bb{u}\|_2 = 1}  \bb{u}^\t\summ \pim \Smx \bb{u}
\leqslant \frac{1}{\propc^3} \Lambda_{\max} \big( \Sx\big)
\leqslant \frac{C}{\propc^3}.
\end{equation}
Combining \eqref{eq:babove0-g}, \eqref{eq:hSh-g}, and \eqref{eq:LS-g}, the following inequality holds in probability:
\begin{equation}\label{eq:babove-g}
 \|\hSxw \big\{ \gz{1} - \hgz{1} \big\}\|_2^2 \leqslant
  C^2  (1+\eta)  \lambda^2_1 s/\propc^3.
\end{equation}
Finally, combining \eqref{eq:bbelow-g} and \eqref{eq:babove-g}, we have, with probability tending to one,
$$
\|\hgz{1}\|_0 \leqslant  \dfrac{C^2(1+\eta)}{\propc^3(1 - \eta)^2}  s=:\widetilde{C}s.
$$
\end{proof}

\section{Additional simulation results}
\label{sec:addsim}

\subsection{Comparison of methods}

\subsubsection{Lasso, ridge, and elastic net}

In this section, we examine the performance of the proposed methods in scenarios with highly correlated covariates and a dense coefficient vector. We set $s=200$, $p=400$, and $n=300$.
Using a setting similar to that in \citet{Yue2018}, we generate $\xr_i$ as follows:
$$
\begin{array}{lll}
\xr_{i t}=Z_1+\varepsilon_{i t}^x, & Z_1 \sim N(0,1), & t=1, 2, 3 \\
\xr_{i t}=Z_2+\varepsilon_{i t}^x, & Z_2 \sim N(0,1), & t=4, 5, 6, \\
...&&\\
\xr_{i t}=Z_{66}+\varepsilon_{i t}^x, & Z_{66} \sim N(0,1), & t=196, 197, 198,\\
\xr_{i t} \stackrel{i.i.d.}\sim N(0,1) && t=199, \ldots, 400,
\end{array}
$$
where $\varepsilon_{i t}^x \sim N(0,0.1)$.
We generate potential outcomes by $Y_i(1)=Y_i(0)=B_i/M+\xr_i^\T \boldsymbol{\beta}-2\xr_{bc,i}^\T \boldsymbol{\beta} + \e_i$, where the first $s$ elements of $\boldsymbol{\beta}$ are generated from $U[-0.1,0.1]$ and the remaining elements are zero. The setups for blocking are similar to those in Section \ref{sec:sim}.

We also consider ridge \citep{hoerl1970ridge} and elastic net \citep{zou2005regularization} methods for estimating $\gz{z}$. Therefore, there are four estimators, denoted as unadj, lasso, ridge, and enet. We employ 10-fold cross-validation to select the tuning parameters.

Tables \ref{tab:sim_corr_1} and \ref{tab:sim_corr_2} display the simulation results.
First, in most cases, all methods demonstrate negligible biases and desirable coverage probabilities, except for three instances where the elastic net produces approximately 87\% coverage. Second, when compared to the unadjusted estimator, the Lasso exhibits reductions ranging from 3\%--28\% in standard deviation (SD) and 4\%--25\% in interval length. Additionally, compared to the Lasso, the elastic net can achieve further reductions of 2\%--8\% in SD and 8\%--27\% in interval length while selecting more covariates for the working model.
Overall, in this simulation setting, while the proposed Lasso-adjusted ATE estimator may not exhibit a significant efficiency improvement, it continues to deliver reasonable performance.
Compared to the Lasso method, the ridge method does not exhibit any advantages, while the elastic net shows the potential for further enhancing precision.

\begin{table}[H]

\caption{\label{tab:sim_corr_1}Simulation results for highly correlated covariates.}
\centering
\begin{threeparttable}
\begin{tabular}[t]{>{\centering\arraybackslash}p{1.4cm}ccrrrrrr}
\toprule
\multicolumn{1}{c}{Scenario} & \multicolumn{1}{c}{Rerand.} & \multicolumn{1}{c}{Est.} & \multicolumn{1}{c}{Bias} & \multicolumn{1}{c}{SD} & \multicolumn{1}{c}{RMSE} & \multicolumn{1}{c}{CP} & \multicolumn{1}{c}{Length} & \multicolumn{1}{c}{$\hat{s}$}\\
\midrule
 & no & unadj & 0.2 (0.3) & 10.1 (0.2) & 10.1 (0.2) & 95.2 (0.7) & 40.0 (0.0) & 0\\

 & yes & unadj & 0.0 (0.3) & 9.1 (0.2) & 9.1 (0.2) & 95.6 (0.6) & 36.9 (0.0) & 0\\

 & no & lasso & 0.2 (0.3) & 7.7 (0.2) & 7.7 (0.2) & 94.6 (0.7) & 30.2 (0.0) & 82\\

 & yes & lasso & -0.1 (0.2) & 7.3 (0.2) & 7.3 (0.2) & 96.3 (0.6) & 30.1 (0.0) & 82\\

 & no & ridge & 0.2 (0.3) & 8.3 (0.2) & 8.3 (0.2) & 92.4 (0.8) & 28.9 (0.0) & 400\\

 & yes & ridge & -0.1 (0.2) & 7.6 (0.2) & 7.6 (0.2) & 94.5 (0.8) & 28.8 (0.0) & 400\\

 & no & enet & 0.2 (0.2) & 7.2 (0.2) & 7.2 (0.2) & 87.8 (1.0) & 22.1 (0.1) & 322\\

\multirow{-8}{1.4cm}{\centering\arraybackslash No block} & yes & enet & -0.2 (0.2) & 6.7 (0.1) & 6.7 (0.1) & 90.2 (0.9) & 22.0 (0.1) & 324\\
\midrule

 & no & unadj & -0.4 (0.3) & 10.7 (0.2) & 10.7 (0.2) & 93.8 (0.8) & 39.8 (0.0) & 0\\

 & yes & unadj & -0.2 (0.3) & 10.3 (0.2) & 10.3 (0.2) & 94.2 (0.7) & 39.3 (0.0) & 0\\

 & no & lasso & -0.3 (0.3) & 9.1 (0.2) & 9.1 (0.2) & 93.1 (0.8) & 33.6 (0.1) & 65\\

 & yes & lasso & -0.2 (0.3) & 8.6 (0.2) & 8.6 (0.2) & 95.2 (0.6) & 33.7 (0.1) & 64\\

 & no & ridge & -0.3 (0.3) & 8.8 (0.2) & 8.8 (0.2) & 90.1 (0.9) & 29.0 (0.0) & 400\\

 & yes & ridge & -0.1 (0.3) & 8.4 (0.2) & 8.4 (0.2) & 90.9 (0.9) & 28.9 (0.0) & 400\\

 & no & enet & -0.2 (0.2) & 8.3 (0.2) & 8.3 (0.2) & 88.4 (1.0) & 26.5 (0.1) & 331\\

\multirow{-8}{1.4cm}{\centering\arraybackslash Large, equal} & yes & enet & -0.1 (0.3) & 8.0 (0.2) & 8.0 (0.2) & 89.7 (1.0) & 26.5 (0.1) & 331\\
\midrule

 & no & unadj & 0.8 (0.3) & 10.1 (0.2) & 10.1 (0.2) & 95.7 (0.6) & 40.9 (0.0) & 0\\

 & yes & unadj & 0.4 (0.3) & 9.8 (0.2) & 9.8 (0.2) & 94.3 (0.7) & 39.0 (0.0) & 0\\

 & no & lasso & 0.1 (0.3) & 9.8 (0.2) & 9.8 (0.2) & 96.2 (0.6) & 39.4 (0.1) & 19\\

 & yes & lasso & -0.4 (0.3) & 9.6 (0.2) & 9.6 (0.2) & 95.3 (0.7) & 39.5 (0.1) & 21\\

 & no & ridge & 0.2 (0.3) & 9.6 (0.2) & 9.6 (0.2) & 95.8 (0.6) & 38.2 (0.1) & 400\\

 & yes & ridge & -0.2 (0.3) & 9.4 (0.2) & 9.4 (0.2) & 94.9 (0.7) & 38.3 (0.1) & 400\\

 & no & enet & -0.4 (0.3) & 9.4 (0.2) & 9.4 (0.2) & 94.2 (0.7) & 36.1 (0.1) & 353\\

\multirow{-8}{1.4cm}{\centering\arraybackslash Large, unequal} & yes & enet & -0.8 (0.3) & 9.2 (0.2) & 9.2 (0.2) & 93.9 (0.8) & 36.1 (0.1) & 348\\
\midrule

 & no & unadj & -0.7 (0.4) & 11.6 (0.3) & 11.6 (0.3) & 94.5 (0.7) & 45.1 (0.0) & 0\\

 & yes & unadj & 0.1 (0.3) & 11.0 (0.2) & 11.0 (0.2) & 95.7 (0.6) & 44.6 (0.0) & 0\\

 & no & lasso & -0.6 (0.3) & 10.0 (0.2) & 10.0 (0.2) & 94.9 (0.7) & 38.8 (0.1) & 45\\

 & yes & lasso & 0.1 (0.3) & 9.7 (0.2) & 9.7 (0.2) & 95.6 (0.6) & 38.8 (0.1) & 45\\

 & no & ridge & -0.5 (0.3) & 10.0 (0.2) & 10.0 (0.2) & 91.2 (0.8) & 34.4 (0.0) & 400\\

 & yes & ridge & 0.1 (0.3) & 9.5 (0.2) & 9.5 (0.2) & 92.6 (0.9) & 34.5 (0.0) & 400\\

 & no & enet & -0.5 (0.3) & 9.5 (0.2) & 9.5 (0.2) & 90.6 (0.9) & 32.1 (0.1) & 278\\

\multirow{-8}{1.4cm}{\centering\arraybackslash Small, equal} & yes & enet & 0.1 (0.3) & 9.2 (0.2) & 9.2 (0.2) & 91.1 (0.9) & 32.2 (0.1) & 278\\
\midrule

 & no & unadj & -0.1 (0.3) & 9.4 (0.2) & 9.4 (0.2) & 94.6 (0.7) & 36.3 (0.0) & 0\\

 & yes & unadj & 0.1 (0.3) & 9.2 (0.2) & 9.2 (0.2) & 94.9 (0.7) & 35.9 (0.0) & 0\\

 & no & lasso & 0.5 (0.3) & 8.7 (0.2) & 8.7 (0.2) & 94.8 (0.7) & 33.3 (0.1) & 24\\

 & yes & lasso & 0.5 (0.3) & 8.4 (0.2) & 8.4 (0.2) & 95.4 (0.7) & 33.4 (0.1) & 22\\

 & no & ridge & 0.5 (0.3) & 8.7 (0.2) & 8.8 (0.2) & 93.7 (0.8) & 32.3 (0.1) & 400\\

 & yes & ridge & 0.5 (0.3) & 8.5 (0.2) & 8.5 (0.2) & 94.3 (0.8) & 32.4 (0.0) & 400\\

 & no & enet & 0.8 (0.3) & 8.4 (0.2) & 8.5 (0.2) & 93.1 (0.8) & 30.5 (0.1) & 312\\

\multirow{-8}{1.4cm}{\centering\arraybackslash Small, unequal} & yes & enet & 0.8 (0.3) & 8.1 (0.2) & 8.2 (0.2) & 94.0 (0.8) & 30.7 (0.1) & 321\\
\bottomrule
\end{tabular}
\begin{tablenotes}
\item Note: The numbers in brackets are the corresponding standard errors estimated using the bootstrap with 500 replications. Bias, SD, RMSE, CP, Length, and their standard errors are multiplied by 100.
\end{tablenotes}
\end{threeparttable}
\end{table}

\begin{table}[H]

\caption{\label{tab:sim_corr_2}Simulation results for highly correlated covariates.}
\centering
\begin{threeparttable}
\begin{tabular}[t]{>{\centering\arraybackslash}p{1.4cm}ccrrrrrr}
\toprule
\multicolumn{1}{c}{Scenario} & \multicolumn{1}{c}{Rerand.} & \multicolumn{1}{c}{Est.} & \multicolumn{1}{c}{Bias} & \multicolumn{1}{c}{SD} & \multicolumn{1}{c}{RMSE} & \multicolumn{1}{c}{CP} & \multicolumn{1}{c}{Length} & \multicolumn{1}{c}{$\hat{s}$}\\
\midrule
 & no & unadj & 0.4 (0.3) & 10.4 (0.2) & 10.4 (0.2) & 94.5 (0.7) & 39.6 (0.0) & 0\\

 & yes & unadj & -0.1 (0.3) & 10.0 (0.2) & 10.0 (0.2) & 94.6 (0.7) & 39.0 (0.0) & 0\\

 & no & lasso & 0.2 (0.3) & 9.0 (0.2) & 9.0 (0.2) & 94.9 (0.7) & 34.6 (0.1) & 45\\

 & yes & lasso & -0.2 (0.3) & 8.8 (0.2) & 8.8 (0.2) & 95.7 (0.6) & 34.6 (0.1) & 45\\

 & no & ridge & 0.3 (0.3) & 8.9 (0.2) & 8.9 (0.2) & 91.7 (0.9) & 30.5 (0.0) & 400\\

 & yes & ridge & -0.1 (0.3) & 8.6 (0.2) & 8.6 (0.2) & 92.7 (0.8) & 30.5 (0.0) & 400\\

 & no & enet & 0.2 (0.3) & 8.5 (0.2) & 8.5 (0.2) & 90.4 (0.9) & 28.4 (0.1) & 308\\

\multirow{-8}{1.4cm}{\centering\arraybackslash Hybrid, equal} & yes & enet & -0.2 (0.3) & 8.2 (0.2) & 8.2 (0.2) & 91.3 (0.9) & 28.4 (0.1) & 306\\
\midrule

 & no & unadj & 0.0 (0.3) & 10.6 (0.2) & 10.6 (0.2) & 93.6 (0.7) & 40.5 (0.0) & 0\\

 & yes & unadj & 0.6 (0.3) & 10.1 (0.2) & 10.1 (0.2) & 94.4 (0.8) & 38.8 (0.0) & 0\\

 & no & lasso & 0.4 (0.3) & 10.1 (0.2) & 10.1 (0.2) & 93.5 (0.8) & 38.7 (0.1) & 25\\

 & yes & lasso & 1.0 (0.3) & 9.6 (0.2) & 9.7 (0.2) & 95.1 (0.7) & 38.7 (0.1) & 25\\

 & no & ridge & 0.4 (0.3) & 10.2 (0.2) & 10.2 (0.2) & 92.7 (0.9) & 37.7 (0.1) & 400\\

 & yes & ridge & 1.0 (0.3) & 9.7 (0.2) & 9.8 (0.2) & 94.8 (0.7) & 37.6 (0.1) & 400\\

 & no & enet & 0.7 (0.3) & 9.9 (0.2) & 9.9 (0.2) & 91.4 (0.9) & 35.6 (0.1) & 310\\

\multirow{-8}{1.4cm}{\centering\arraybackslash Hybrid, unequal} & yes & enet & 1.3 (0.3) & 9.4 (0.2) & 9.5 (0.2) & 93.3 (0.8) & 35.5 (0.1) & 317\\
\midrule

 & no & unadj & -0.3 (0.3) & 11.2 (0.2) & 11.2 (0.2) & 96.1 (0.6) & 44.7 (0.1) & 0\\

 & yes & unadj & 0.2 (0.3) & 10.8 (0.2) & 10.8 (0.2) & 93.8 (0.7) & 42.0 (0.1) & 0\\

 & no & lasso & -0.3 (0.3) & 10.2 (0.2) & 10.2 (0.2) & 94.3 (0.8) & 39.7 (0.1) & 28\\

 & yes & lasso & 0.3 (0.3) & 10.2 (0.2) & 10.2 (0.2) & 93.9 (0.8) & 39.7 (0.1) & 28\\

 & no & ridge & -0.3 (0.3) & 9.9 (0.2) & 9.9 (0.2) & 91.2 (0.9) & 35.2 (0.1) & 400\\

 & yes & ridge & 0.2 (0.3) & 9.7 (0.2) & 9.7 (0.2) & 91.8 (0.8) & 35.3 (0.1) & 400\\

 & no & enet & -0.2 (0.3) & 9.7 (0.2) & 9.7 (0.2) & 90.5 (0.9) & 33.1 (0.1) & 298\\

\multirow{-8}{1.4cm}{\centering\arraybackslash Triplet, equal} & yes & enet & 0.3 (0.3) & 9.6 (0.2) & 9.6 (0.2) & 90.6 (0.9) & 32.9 (0.2) & 296\\
\midrule

 & no & unadj & -0.0 (0.3) & 11.1 (0.3) & 11.1 (0.3) & 94.8 (0.7) & 43.7 (0.1) & 0\\

 & yes & unadj & -0.7 (0.3) & 10.8 (0.3) & 10.8 (0.3) & 93.5 (0.8) & 42.1 (0.1) & 0\\

 & no & lasso & 1.0 (0.3) & 10.6 (0.3) & 10.6 (0.3) & 94.0 (0.8) & 40.3 (0.1) & 18\\

 & yes & lasso & 0.2 (0.3) & 10.4 (0.3) & 10.4 (0.3) & 94.3 (0.8) & 40.3 (0.1) & 17\\

 & no & ridge & 0.6 (0.3) & 10.5 (0.3) & 10.5 (0.3) & 93.7 (0.8) & 39.9 (0.1) & 400\\

 & yes & ridge & -0.0 (0.3) & 10.3 (0.2) & 10.3 (0.2) & 94.4 (0.7) & 39.8 (0.1) & 400\\

 & no & enet & 1.3 (0.3) & 10.2 (0.2) & 10.3 (0.2) & 92.7 (0.8) & 36.8 (0.1) & 321\\

\multirow{-8}{1.4cm}{\centering\arraybackslash Triplet, unequal} & yes & enet & 0.6 (0.3) & 10.0 (0.2) & 10.0 (0.2) & 92.9 (0.8) & 36.7 (0.1) & 311\\
\bottomrule
\end{tabular}
\begin{tablenotes}
\item Note: The numbers in brackets are the corresponding standard errors estimated using the bootstrap with 500 replications. Bias, SD, RMSE, CP, Length, and their standard errors are multiplied by 100.
\end{tablenotes}
\end{threeparttable}
\end{table}

\subsubsection{OLS, Lasso, and Lasso + OLS}

In this section, we explore the bias of the Lasso-adjusted ATE estimator in dense scenarios. We set $s=40,200,400$, $p=400$, and $n=300$. We generate potential outcomes by $Y_i(1)=Y_i(0)=B_i/M+\xr_i^\T \boldsymbol{\beta}-2\xr_{bc,i}^\T \boldsymbol{\beta} + \e_i$, where the first $s$ elements of $\boldsymbol{\beta}$ are generated from $U[-0.1,0.1]$ and the remaining elements are zero. The other setups are similar to those in Section \ref{sec:sim}.

Furthermore, we use OLS and Lasso + OLS methods for estimating $\gz{z}$. For OLS, we adjust solely for the first 5 covariates. For Lasso + OLS, we employ Lasso for covariate selection and then refit the model using OLS. This simulation is repeated 20,000 times.

Figure \ref{fig:sim_bias} illustrates the bias of all methods in various scenarios. Tables \ref{tab:sim_bias_40_1} to \ref{tab:sim_bias_400_2} show the detailed summary statistics. First, the biases of all methods are negligible compared to the standard deviation. In cases where propensity scores are equal across blocks, the biases of all methods are similar. However, when propensity scores differ across blocks, the biases of Lasso and Lasso + OLS are larger than those arising from unpenalized regression. Additionally, these biases tend to increase with the value of sparsity $s$. Second, the coverage probabilities of all methods are close to the confidence level. Third, Lasso + OLS achieves the best efficiency with reductions of up to 43\% in standard deviation and up to 50\% in mean confidence interval length compared to the unadjusted estimator.

\begin{figure}[p]
\centering
\includegraphics[width=1\textwidth]{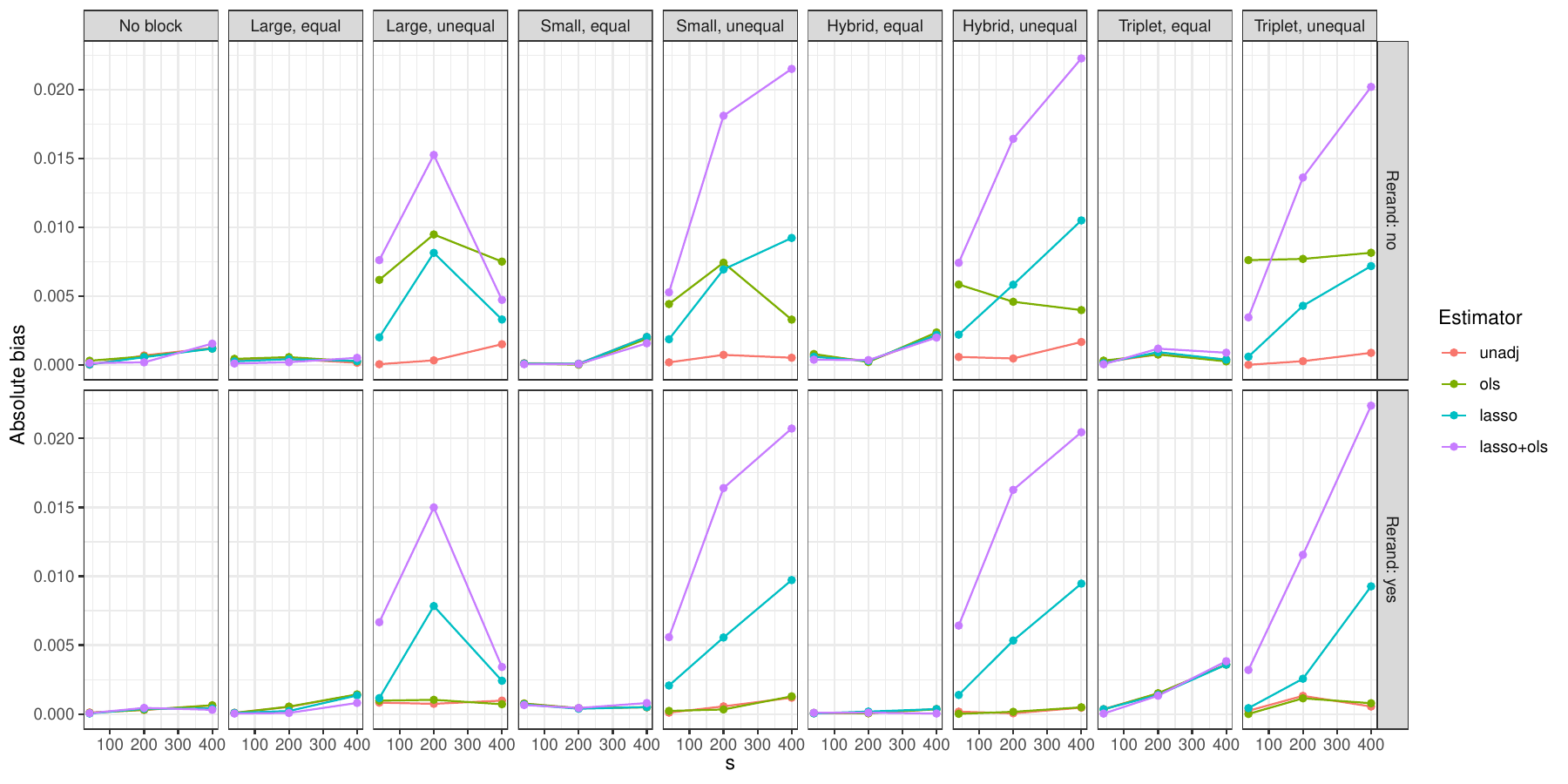}
\caption{Absolute bias of unadjusted, OLS, Lasso, and Lasso + OLS estimator, when $s=40,200,400$.}
\label{fig:sim_bias}
\end{figure}

\clearpage

\begin{table}[p]

\caption{\label{tab:sim_bias_40_1}Simulation results for OLS, Lasso, and Lasso + OLS, when $s=40$.}
\centering
\begin{threeparttable}
\begin{tabular}[t]{>{\centering\arraybackslash}p{1.4cm}ccrrrrr}
\toprule
\multicolumn{1}{c}{Scenario} & \multicolumn{1}{c}{Rerand.} & \multicolumn{1}{c}{Est.} & \multicolumn{1}{c}{Bias} & \multicolumn{1}{c}{SD} & \multicolumn{1}{c}{RMSE} & \multicolumn{1}{c}{CP} & \multicolumn{1}{c}{Length}\\
\midrule
 & no & unadj & -0.002 (0.036) & 5.1 (0.0) & 5.1 (0.0) & 94.9 (0.2) & 20.0 (0.0)\\

 & yes & unadj & -0.010 (0.036) & 4.9 (0.0) & 4.9 (0.0) & 94.7 (0.2) & 18.9 (0.0)\\

 & no & ols & -0.030 (0.035) & 4.9 (0.0) & 4.9 (0.0) & 94.6 (0.2) & 19.0 (0.0)\\

 & yes & ols & -0.009 (0.034) & 4.9 (0.0) & 4.9 (0.0) & 94.9 (0.2) & 19.0 (0.0)\\

 & no & lasso & 0.004 (0.023) & 3.3 (0.0) & 3.3 (0.0) & 94.4 (0.2) & 12.7 (0.0)\\

 & yes & lasso & -0.004 (0.022) & 3.2 (0.0) & 3.2 (0.0) & 95.1 (0.1) & 12.7 (0.0)\\

 & no & lasso+ols & 0.014 (0.020) & 3.0 (0.0) & 3.0 (0.0) & 90.8 (0.2) & 10.1 (0.0)\\

\multirow{-8}{1.4cm}{\centering\arraybackslash No block} & yes & lasso+ols & 0.006 (0.021) & 2.9 (0.0) & 2.9 (0.0) & 91.4 (0.2) & 10.1 (0.0)\\
\midrule

 & no & unadj & -0.036 (0.031) & 4.4 (0.0) & 4.4 (0.0) & 95.0 (0.1) & 17.4 (0.0)\\

 & yes & unadj & 0.009 (0.029) & 4.1 (0.0) & 4.1 (0.0) & 94.6 (0.2) & 15.9 (0.0)\\

 & no & ols & -0.044 (0.027) & 4.2 (0.0) & 4.2 (0.0) & 94.4 (0.2) & 16.1 (0.0)\\

 & yes & ols & 0.007 (0.029) & 4.1 (0.0) & 4.1 (0.0) & 94.8 (0.2) & 16.1 (0.0)\\

 & no & lasso & -0.025 (0.027) & 3.8 (0.0) & 3.8 (0.0) & 94.9 (0.2) & 14.8 (0.0)\\

 & yes & lasso & 0.008 (0.024) & 3.5 (0.0) & 3.5 (0.0) & 96.2 (0.1) & 14.8 (0.0)\\

 & no & lasso+ols & -0.010 (0.022) & 3.3 (0.0) & 3.3 (0.0) & 93.6 (0.2) & 12.1 (0.0)\\

\multirow{-8}{1.4cm}{\centering\arraybackslash Large, equal} & yes & lasso+ols & 0.003 (0.022) & 3.1 (0.0) & 3.1 (0.0) & 94.5 (0.1) & 12.1 (0.0)\\
\midrule

 & no & unadj & 0.005 (0.035) & 4.7 (0.0) & 4.7 (0.0) & 95.1 (0.1) & 18.5 (0.0)\\

 & yes & unadj & 0.082 (0.033) & 4.6 (0.0) & 4.6 (0.0) & 94.5 (0.2) & 17.9 (0.0)\\

 & no & ols & 0.617 (0.033) & 4.7 (0.0) & 4.8 (0.0) & 94.6 (0.2) & 18.4 (0.0)\\

 & yes & ols & 0.097 (0.034) & 4.6 (0.0) & 4.6 (0.0) & 95.3 (0.2) & 18.4 (0.0)\\

 & no & lasso & -0.200 (0.033) & 4.5 (0.0) & 4.5 (0.0) & 95.0 (0.1) & 17.7 (0.0)\\

 & yes & lasso & -0.115 (0.033) & 4.5 (0.0) & 4.5 (0.0) & 95.2 (0.1) & 17.8 (0.0)\\

 & no & lasso+ols & -0.761 (0.031) & 4.3 (0.0) & 4.4 (0.0) & 93.1 (0.2) & 16.2 (0.0)\\

\multirow{-8}{1.4cm}{\centering\arraybackslash Large, unequal} & yes & lasso+ols & -0.666 (0.032) & 4.2 (0.0) & 4.3 (0.0) & 93.7 (0.2) & 16.2 (0.0)\\
\midrule

 & no & unadj & -0.013 (0.036) & 4.8 (0.0) & 4.8 (0.0) & 95.1 (0.2) & 18.8 (0.0)\\

 & yes & unadj & 0.076 (0.033) & 4.7 (0.0) & 4.7 (0.0) & 94.6 (0.2) & 18.3 (0.0)\\

 & no & ols & -0.010 (0.034) & 4.7 (0.0) & 4.7 (0.0) & 94.8 (0.2) & 18.4 (0.0)\\

 & yes & ols & 0.076 (0.034) & 4.7 (0.0) & 4.7 (0.0) & 94.9 (0.2) & 18.4 (0.0)\\

 & no & lasso & -0.010 (0.030) & 4.4 (0.0) & 4.4 (0.0) & 95.1 (0.2) & 17.1 (0.0)\\

 & yes & lasso & 0.069 (0.028) & 4.3 (0.0) & 4.3 (0.0) & 95.6 (0.1) & 17.1 (0.0)\\

 & no & lasso+ols & 0.005 (0.026) & 3.9 (0.0) & 3.9 (0.0) & 93.8 (0.2) & 14.4 (0.0)\\

\multirow{-8}{1.4cm}{\centering\arraybackslash Small, equal} & yes & lasso+ols & 0.065 (0.026) & 3.7 (0.0) & 3.7 (0.0) & 94.8 (0.2) & 14.4 (0.0)\\
\midrule

 & no & unadj & -0.018 (0.030) & 4.0 (0.0) & 4.0 (0.0) & 95.0 (0.2) & 16.0 (0.0)\\

 & yes & unadj & 0.010 (0.029) & 4.0 (0.0) & 4.0 (0.0) & 94.3 (0.2) & 15.5 (0.0)\\

 & no & ols & 0.442 (0.028) & 4.1 (0.0) & 4.1 (0.0) & 94.6 (0.2) & 15.8 (0.0)\\

 & yes & ols & 0.022 (0.028) & 4.0 (0.0) & 4.0 (0.0) & 94.9 (0.2) & 15.8 (0.0)\\

 & no & lasso & 0.186 (0.027) & 3.8 (0.0) & 3.8 (0.0) & 94.8 (0.2) & 14.8 (0.0)\\

 & yes & lasso & 0.207 (0.027) & 3.7 (0.0) & 3.8 (0.0) & 95.0 (0.1) & 14.8 (0.0)\\

 & no & lasso+ols & 0.527 (0.025) & 3.5 (0.0) & 3.6 (0.0) & 93.2 (0.2) & 13.1 (0.0)\\

\multirow{-8}{1.4cm}{\centering\arraybackslash Small, unequal} & yes & lasso+ols & 0.558 (0.025) & 3.5 (0.0) & 3.5 (0.0) & 93.3 (0.2) & 13.1 (0.0)\\
\bottomrule
\end{tabular}
\begin{tablenotes}
\item Note: The numbers in brackets are the corresponding standard errors estimated using the bootstrap with 500 replications. Bias, SD, RMSE, CP, Length, and their standard errors are multiplied by 100.
\end{tablenotes}
\end{threeparttable}
\end{table}

\begin{table}[p]

\caption{\label{tab:sim_bias_40_2}Simulation results for OLS, Lasso, and Lasso + OLS, when $s=40$.}
\centering
\begin{threeparttable}
\begin{tabular}[t]{>{\centering\arraybackslash}p{1.4cm}ccrrrrr}
\toprule
\multicolumn{1}{c}{Scenario} & \multicolumn{1}{c}{Rerand.} & \multicolumn{1}{c}{Est.} & \multicolumn{1}{c}{Bias} & \multicolumn{1}{c}{SD} & \multicolumn{1}{c}{RMSE} & \multicolumn{1}{c}{CP} & \multicolumn{1}{c}{Length}\\
\midrule
 & no & unadj & -0.070 (0.032) & 4.7 (0.0) & 4.7 (0.0) & 94.9 (0.2) & 18.3 (0.0)\\

 & yes & unadj & -0.005 (0.031) & 4.4 (0.0) & 4.4 (0.0) & 94.6 (0.2) & 16.9 (0.0)\\

 & no & ols & -0.079 (0.032) & 4.4 (0.0) & 4.4 (0.0) & 94.6 (0.2) & 17.0 (0.0)\\

 & yes & ols & -0.008 (0.032) & 4.4 (0.0) & 4.4 (0.0) & 94.8 (0.2) & 17.0 (0.0)\\

 & no & lasso & -0.059 (0.027) & 4.0 (0.0) & 4.0 (0.0) & 94.8 (0.1) & 15.6 (0.0)\\

 & yes & lasso & -0.004 (0.025) & 3.8 (0.0) & 3.8 (0.0) & 96.0 (0.1) & 15.6 (0.0)\\

 & no & lasso+ols & -0.038 (0.025) & 3.3 (0.0) & 3.3 (0.0) & 93.1 (0.2) & 12.0 (0.0)\\

\multirow{-8}{1.4cm}{\centering\arraybackslash Hybrid, equal} & yes & lasso+ols & -0.009 (0.023) & 3.2 (0.0) & 3.2 (0.0) & 93.8 (0.2) & 12.0 (0.0)\\
\midrule

 & no & unadj & 0.058 (0.033) & 4.6 (0.0) & 4.6 (0.0) & 94.8 (0.2) & 18.0 (0.0)\\

 & yes & unadj & -0.017 (0.031) & 4.4 (0.0) & 4.4 (0.0) & 94.8 (0.1) & 17.2 (0.0)\\

 & no & ols & 0.584 (0.032) & 4.5 (0.0) & 4.6 (0.0) & 94.4 (0.2) & 17.5 (0.0)\\

 & yes & ols & -0.001 (0.032) & 4.4 (0.0) & 4.4 (0.0) & 95.2 (0.2) & 17.5 (0.0)\\

 & no & lasso & 0.220 (0.029) & 4.1 (0.0) & 4.1 (0.0) & 94.9 (0.2) & 16.2 (0.0)\\

 & yes & lasso & 0.138 (0.029) & 3.9 (0.0) & 3.9 (0.0) & 96.1 (0.1) & 16.2 (0.0)\\

 & no & lasso+ols & 0.741 (0.026) & 3.7 (0.0) & 3.8 (0.0) & 92.8 (0.2) & 13.9 (0.0)\\

\multirow{-8}{1.4cm}{\centering\arraybackslash Hybrid, unequal} & yes & lasso+ols & 0.642 (0.025) & 3.5 (0.0) & 3.6 (0.0) & 94.3 (0.2) & 13.9 (0.0)\\
\midrule

 & no & unadj & -0.019 (0.033) & 4.4 (0.0) & 4.4 (0.0) & 94.6 (0.2) & 17.3 (0.0)\\

 & yes & unadj & 0.035 (0.028) & 4.2 (0.0) & 4.2 (0.0) & 93.6 (0.2) & 15.7 (0.0)\\

 & no & ols & 0.031 (0.031) & 4.3 (0.0) & 4.3 (0.0) & 93.8 (0.2) & 16.3 (0.0)\\

 & yes & ols & 0.032 (0.030) & 4.2 (0.0) & 4.2 (0.0) & 94.7 (0.2) & 16.3 (0.0)\\

 & no & lasso & -0.013 (0.030) & 4.2 (0.0) & 4.2 (0.0) & 94.2 (0.2) & 16.1 (0.0)\\

 & yes & lasso & 0.035 (0.026) & 3.9 (0.0) & 3.9 (0.0) & 95.8 (0.1) & 16.1 (0.0)\\

 & no & lasso+ols & 0.005 (0.028) & 3.9 (0.0) & 3.9 (0.0) & 92.5 (0.2) & 14.2 (0.0)\\

\multirow{-8}{1.4cm}{\centering\arraybackslash Triplet, equal} & yes & lasso+ols & 0.003 (0.027) & 3.6 (0.0) & 3.6 (0.0) & 94.5 (0.2) & 14.2 (0.0)\\
\midrule

 & no & unadj & 0.001 (0.034) & 4.8 (0.0) & 4.8 (0.0) & 94.7 (0.2) & 19.0 (0.0)\\

 & yes & unadj & -0.023 (0.032) & 4.3 (0.0) & 4.3 (0.0) & 92.2 (0.2) & 15.7 (0.0)\\

 & no & ols & 0.761 (0.030) & 4.4 (0.0) & 4.4 (0.0) & 94.1 (0.2) & 16.9 (0.0)\\

 & yes & ols & -0.000 (0.029) & 4.3 (0.0) & 4.3 (0.0) & 94.7 (0.2) & 16.8 (0.0)\\

 & no & lasso & 0.060 (0.032) & 4.6 (0.0) & 4.6 (0.0) & 94.6 (0.2) & 18.0 (0.0)\\

 & yes & lasso & 0.042 (0.030) & 4.1 (0.0) & 4.1 (0.0) & 96.7 (0.1) & 18.0 (0.0)\\

 & no & lasso+ols & 0.345 (0.033) & 4.6 (0.0) & 4.6 (0.0) & 92.5 (0.2) & 16.6 (0.0)\\

\multirow{-8}{1.4cm}{\centering\arraybackslash Triplet, unequal} & yes & lasso+ols & 0.319 (0.032) & 4.3 (0.0) & 4.3 (0.0) & 94.5 (0.2) & 16.6 (0.0)\\
\bottomrule
\end{tabular}
\begin{tablenotes}
\item Note: The numbers in brackets are the corresponding standard errors estimated using the bootstrap with 500 replications. Bias, SD, RMSE, CP, Length, and their standard errors are multiplied by 100.
\end{tablenotes}
\end{threeparttable}
\end{table}

\begin{table}[p]

\caption{\label{tab:sim_bias_200_1}Simulation results for OLS, Lasso, and Lasso + OLS, when $s=200$.}
\centering
\begin{threeparttable}
\begin{tabular}[t]{>{\centering\arraybackslash}p{1.4cm}ccrrrrr}
\toprule
\multicolumn{1}{c}{Scenario} & \multicolumn{1}{c}{Rerand.} & \multicolumn{1}{c}{Est.} & \multicolumn{1}{c}{Bias} & \multicolumn{1}{c}{SD} & \multicolumn{1}{c}{RMSE} & \multicolumn{1}{c}{CP} & \multicolumn{1}{c}{Length}\\
\midrule
 & no & unadj & 0.069 (0.079) & 11.3 (0.1) & 11.3 (0.1) & 94.8 (0.2) & 44.0 (0.0)\\

 & yes & unadj & -0.029 (0.080) & 11.2 (0.1) & 11.2 (0.1) & 94.7 (0.2) & 43.7 (0.0)\\

 & no & ols & 0.062 (0.084) & 11.5 (0.1) & 11.5 (0.1) & 94.4 (0.2) & 44.1 (0.0)\\

 & yes & ols & -0.030 (0.083) & 11.2 (0.1) & 11.2 (0.1) & 94.9 (0.2) & 44.1 (0.0)\\

 & no & lasso & 0.056 (0.066) & 9.3 (0.0) & 9.3 (0.0) & 94.6 (0.2) & 35.8 (0.0)\\

 & yes & lasso & -0.038 (0.064) & 9.2 (0.0) & 9.2 (0.0) & 94.8 (0.2) & 35.8 (0.0)\\

 & no & lasso+ols & 0.019 (0.058) & 8.6 (0.0) & 8.6 (0.0) & 90.3 (0.2) & 28.8 (0.0)\\

\multirow{-8}{1.4cm}{\centering\arraybackslash No block} & yes & lasso+ols & -0.044 (0.062) & 8.7 (0.0) & 8.7 (0.0) & 90.1 (0.2) & 28.8 (0.0)\\
\midrule

 & no & unadj & -0.043 (0.078) & 11.1 (0.1) & 11.1 (0.1) & 95.0 (0.2) & 43.8 (0.0)\\

 & yes & unadj & -0.052 (0.074) & 11.1 (0.1) & 11.1 (0.1) & 94.6 (0.2) & 43.1 (0.0)\\

 & no & ols & -0.056 (0.082) & 11.1 (0.1) & 11.1 (0.1) & 94.9 (0.2) & 43.5 (0.0)\\

 & yes & ols & -0.053 (0.081) & 11.1 (0.1) & 11.1 (0.1) & 94.8 (0.2) & 43.5 (0.0)\\

 & no & lasso & -0.042 (0.067) & 9.4 (0.0) & 9.4 (0.0) & 94.8 (0.2) & 37.0 (0.0)\\

 & yes & lasso & -0.022 (0.067) & 9.5 (0.0) & 9.5 (0.0) & 94.9 (0.1) & 37.0 (0.0)\\

 & no & lasso+ols & -0.020 (0.063) & 8.6 (0.0) & 8.6 (0.0) & 91.0 (0.2) & 29.4 (0.0)\\

\multirow{-8}{1.4cm}{\centering\arraybackslash Large, equal} & yes & lasso+ols & 0.008 (0.062) & 8.7 (0.0) & 8.7 (0.0) & 90.7 (0.2) & 29.4 (0.0)\\
\midrule

 & no & unadj & 0.033 (0.085) & 12.7 (0.1) & 12.7 (0.1) & 94.8 (0.2) & 49.7 (0.0)\\

 & yes & unadj & 0.074 (0.089) & 12.5 (0.1) & 12.5 (0.1) & 94.5 (0.2) & 48.4 (0.0)\\

 & no & ols & 0.947 (0.092) & 12.6 (0.1) & 12.6 (0.1) & 94.5 (0.2) & 48.9 (0.0)\\

 & yes & ols & 0.103 (0.096) & 12.5 (0.1) & 12.5 (0.1) & 94.8 (0.2) & 48.9 (0.0)\\

 & no & lasso & -0.814 (0.084) & 11.9 (0.1) & 12.0 (0.1) & 94.8 (0.2) & 46.8 (0.0)\\

 & yes & lasso & -0.784 (0.085) & 11.8 (0.1) & 11.8 (0.1) & 95.2 (0.2) & 46.8 (0.0)\\

 & no & lasso+ols & -1.524 (0.085) & 11.6 (0.1) & 11.7 (0.1) & 93.5 (0.2) & 43.4 (0.0)\\

\multirow{-8}{1.4cm}{\centering\arraybackslash Large, unequal} & yes & lasso+ols & -1.500 (0.081) & 11.4 (0.1) & 11.5 (0.1) & 94.0 (0.2) & 43.4 (0.0)\\
\midrule

 & no & unadj & 0.009 (0.079) & 10.6 (0.1) & 10.6 (0.1) & 95.0 (0.2) & 41.8 (0.0)\\

 & yes & unadj & 0.041 (0.079) & 10.7 (0.1) & 10.7 (0.1) & 94.8 (0.2) & 41.5 (0.0)\\

 & no & ols & -0.002 (0.080) & 10.8 (0.1) & 10.8 (0.1) & 94.8 (0.2) & 41.9 (0.0)\\

 & yes & ols & 0.043 (0.077) & 10.7 (0.0) & 10.7 (0.0) & 95.1 (0.1) & 41.9 (0.0)\\

 & no & lasso & 0.008 (0.071) & 10.3 (0.1) & 10.3 (0.1) & 95.0 (0.2) & 40.5 (0.0)\\

 & yes & lasso & 0.038 (0.070) & 10.4 (0.0) & 10.4 (0.0) & 94.9 (0.2) & 40.5 (0.0)\\

 & no & lasso+ols & 0.006 (0.069) & 10.0 (0.0) & 10.0 (0.0) & 93.5 (0.2) & 37.2 (0.0)\\

\multirow{-8}{1.4cm}{\centering\arraybackslash Small, equal} & yes & lasso+ols & 0.043 (0.074) & 10.0 (0.0) & 10.0 (0.0) & 93.3 (0.2) & 37.1 (0.0)\\
\midrule

 & no & unadj & 0.073 (0.067) & 10.2 (0.0) & 10.2 (0.0) & 94.9 (0.2) & 40.0 (0.0)\\

 & yes & unadj & -0.055 (0.070) & 10.1 (0.0) & 10.1 (0.0) & 94.5 (0.2) & 38.9 (0.0)\\

 & no & ols & 0.742 (0.073) & 10.2 (0.0) & 10.2 (0.0) & 94.6 (0.2) & 39.5 (0.0)\\

 & yes & ols & -0.033 (0.070) & 10.1 (0.0) & 10.1 (0.0) & 94.8 (0.2) & 39.4 (0.0)\\

 & no & lasso & 0.693 (0.069) & 10.0 (0.0) & 10.0 (0.0) & 94.8 (0.2) & 38.7 (0.0)\\

 & yes & lasso & 0.555 (0.070) & 9.9 (0.0) & 9.9 (0.0) & 95.0 (0.2) & 38.7 (0.0)\\

 & no & lasso+ols & 1.809 (0.074) & 9.9 (0.0) & 10.1 (0.1) & 92.8 (0.2) & 36.4 (0.0)\\

\multirow{-8}{1.4cm}{\centering\arraybackslash Small, unequal} & yes & lasso+ols & 1.640 (0.072) & 9.8 (0.1) & 9.9 (0.1) & 92.8 (0.2) & 36.3 (0.0)\\
\bottomrule
\end{tabular}
\begin{tablenotes}
\item Note: The numbers in brackets are the corresponding standard errors estimated using the bootstrap with 500 replications. Bias, SD, RMSE, CP, Length, and their standard errors are multiplied by 100.
\end{tablenotes}
\end{threeparttable}
\end{table}

\begin{table}[p]

\caption{\label{tab:sim_bias_200_2}Simulation results for OLS, Lasso, and Lasso + OLS, when $s=200$.}
\centering
\begin{threeparttable}
\begin{tabular}[t]{>{\centering\arraybackslash}p{1.4cm}ccrrrrr}
\toprule
\multicolumn{1}{c}{Scenario} & \multicolumn{1}{c}{Rerand.} & \multicolumn{1}{c}{Est.} & \multicolumn{1}{c}{Bias} & \multicolumn{1}{c}{SD} & \multicolumn{1}{c}{RMSE} & \multicolumn{1}{c}{CP} & \multicolumn{1}{c}{Length}\\
\midrule
 & no & unadj & 0.027 (0.068) & 10.0 (0.0) & 10.0 (0.0) & 94.9 (0.1) & 39.1 (0.0)\\

 & yes & unadj & -0.007 (0.073) & 9.8 (0.1) & 9.8 (0.1) & 94.5 (0.2) & 37.7 (0.0)\\

 & no & ols & 0.021 (0.071) & 9.8 (0.0) & 9.8 (0.0) & 94.7 (0.2) & 38.0 (0.0)\\

 & yes & ols & -0.005 (0.067) & 9.8 (0.0) & 9.8 (0.0) & 94.8 (0.2) & 38.0 (0.0)\\

 & no & lasso & 0.030 (0.065) & 9.4 (0.0) & 9.4 (0.0) & 94.9 (0.2) & 36.7 (0.0)\\

 & yes & lasso & -0.018 (0.066) & 9.3 (0.0) & 9.3 (0.0) & 95.0 (0.2) & 36.7 (0.0)\\

 & no & lasso+ols & 0.035 (0.063) & 9.1 (0.0) & 9.1 (0.0) & 91.9 (0.2) & 31.8 (0.0)\\

\multirow{-8}{1.4cm}{\centering\arraybackslash Hybrid, equal} & yes & lasso+ols & -0.010 (0.065) & 9.1 (0.0) & 9.1 (0.0) & 91.8 (0.2) & 31.7 (0.0)\\
\midrule

 & no & unadj & 0.047 (0.077) & 10.9 (0.1) & 10.9 (0.1) & 94.7 (0.2) & 42.6 (0.0)\\

 & yes & unadj & 0.005 (0.075) & 10.8 (0.1) & 10.8 (0.1) & 94.6 (0.2) & 42.1 (0.0)\\

 & no & ols & 0.458 (0.081) & 11.1 (0.1) & 11.1 (0.1) & 94.5 (0.2) & 42.5 (0.0)\\

 & yes & ols & 0.016 (0.075) & 10.8 (0.1) & 10.8 (0.1) & 94.8 (0.2) & 42.5 (0.0)\\

 & no & lasso & 0.582 (0.069) & 10.2 (0.0) & 10.2 (0.0) & 95.2 (0.1) & 40.9 (0.0)\\

 & yes & lasso & 0.533 (0.071) & 10.1 (0.1) & 10.1 (0.1) & 95.6 (0.1) & 40.9 (0.0)\\

 & no & lasso+ols & 1.641 (0.071) & 9.9 (0.0) & 10.0 (0.0) & 92.5 (0.2) & 37.0 (0.0)\\

\multirow{-8}{1.4cm}{\centering\arraybackslash Hybrid, unequal} & yes & lasso+ols & 1.627 (0.069) & 9.8 (0.1) & 10.0 (0.1) & 92.9 (0.2) & 36.9 (0.0)\\
\midrule

 & no & unadj & -0.077 (0.074) & 10.7 (0.1) & 10.7 (0.1) & 94.6 (0.2) & 41.5 (0.0)\\

 & yes & unadj & 0.149 (0.072) & 10.6 (0.0) & 10.6 (0.0) & 94.2 (0.2) & 40.3 (0.0)\\

 & no & ols & -0.076 (0.077) & 10.8 (0.1) & 10.8 (0.1) & 93.8 (0.2) & 40.9 (0.0)\\

 & yes & ols & 0.151 (0.074) & 10.6 (0.1) & 10.6 (0.1) & 94.6 (0.2) & 40.8 (0.0)\\

 & no & lasso & -0.093 (0.075) & 10.5 (0.1) & 10.5 (0.1) & 94.0 (0.2) & 40.3 (0.0)\\

 & yes & lasso & 0.138 (0.078) & 10.5 (0.1) & 10.5 (0.1) & 94.4 (0.2) & 40.2 (0.0)\\

 & no & lasso+ols & -0.118 (0.073) & 10.6 (0.1) & 10.6 (0.1) & 92.0 (0.2) & 38.3 (0.0)\\

\multirow{-8}{1.4cm}{\centering\arraybackslash Triplet, equal} & yes & lasso+ols & 0.133 (0.075) & 10.6 (0.1) & 10.6 (0.1) & 92.2 (0.2) & 38.3 (0.0)\\
\midrule

 & no & unadj & 0.027 (0.074) & 10.7 (0.1) & 10.7 (0.1) & 94.7 (0.2) & 41.9 (0.0)\\

 & yes & unadj & -0.132 (0.070) & 10.2 (0.0) & 10.2 (0.0) & 94.1 (0.2) & 38.8 (0.0)\\

 & no & ols & 0.769 (0.071) & 10.3 (0.0) & 10.3 (0.0) & 94.2 (0.2) & 39.6 (0.0)\\

 & yes & ols & -0.114 (0.069) & 10.1 (0.1) & 10.1 (0.1) & 94.6 (0.2) & 39.6 (0.0)\\

 & no & lasso & 0.429 (0.076) & 10.5 (0.1) & 10.5 (0.1) & 94.3 (0.2) & 40.9 (0.0)\\

 & yes & lasso & 0.256 (0.071) & 10.0 (0.1) & 10.0 (0.1) & 95.7 (0.1) & 40.9 (0.0)\\

 & no & lasso+ols & 1.360 (0.077) & 10.6 (0.1) & 10.7 (0.1) & 92.6 (0.2) & 39.2 (0.0)\\

\multirow{-8}{1.4cm}{\centering\arraybackslash Triplet, unequal} & yes & lasso+ols & 1.155 (0.071) & 10.2 (0.1) & 10.3 (0.1) & 93.7 (0.2) & 39.2 (0.0)\\
\bottomrule
\end{tabular}
\begin{tablenotes}
\item Note: The numbers in brackets are the corresponding standard errors estimated using the bootstrap with 500 replications. Bias, SD, RMSE, CP, Length, and their standard errors are multiplied by 100.
\end{tablenotes}
\end{threeparttable}
\end{table}

\begin{table}[p]

\caption{\label{tab:sim_bias_400_1}Simulation results for OLS, Lasso, and Lasso + OLS, when $s=400$.}
\centering
\begin{threeparttable}
\begin{tabular}[t]{>{\centering\arraybackslash}p{1.4cm}ccrrrrr}
\toprule
\multicolumn{1}{c}{Scenario} & \multicolumn{1}{c}{Rerand.} & \multicolumn{1}{c}{Est.} & \multicolumn{1}{c}{Bias} & \multicolumn{1}{c}{SD} & \multicolumn{1}{c}{RMSE} & \multicolumn{1}{c}{CP} & \multicolumn{1}{c}{Length}\\
\midrule
 & no & unadj & -0.123 (0.113) & 16.2 (0.1) & 16.2 (0.1) & 95.2 (0.1) & 63.6 (0.0)\\

 & yes & unadj & 0.062 (0.117) & 16.3 (0.1) & 16.3 (0.1) & 94.6 (0.2) & 63.3 (0.0)\\

 & no & ols & -0.118 (0.111) & 16.5 (0.1) & 16.5 (0.1) & 95.0 (0.2) & 63.9 (0.0)\\

 & yes & ols & 0.063 (0.116) & 16.4 (0.1) & 16.4 (0.1) & 94.8 (0.2) & 63.9 (0.0)\\

 & no & lasso & -0.118 (0.096) & 14.6 (0.1) & 14.6 (0.1) & 95.4 (0.1) & 57.6 (0.0)\\

 & yes & lasso & 0.040 (0.107) & 14.7 (0.1) & 14.7 (0.1) & 95.0 (0.1) & 57.6 (0.0)\\

 & no & lasso+ols & -0.155 (0.100) & 14.2 (0.1) & 14.2 (0.1) & 91.4 (0.2) & 49.0 (0.0)\\

\multirow{-8}{1.4cm}{\centering\arraybackslash No block} & yes & lasso+ols & 0.030 (0.097) & 14.3 (0.1) & 14.3 (0.1) & 91.3 (0.2) & 49.0 (0.0)\\
\midrule

 & no & unadj & -0.016 (0.094) & 13.5 (0.1) & 13.5 (0.1) & 94.8 (0.2) & 52.8 (0.0)\\

 & yes & unadj & 0.140 (0.094) & 13.3 (0.1) & 13.3 (0.1) & 94.8 (0.2) & 51.8 (0.0)\\

 & no & ols & -0.026 (0.100) & 13.5 (0.1) & 13.5 (0.1) & 94.6 (0.2) & 52.3 (0.0)\\

 & yes & ols & 0.142 (0.093) & 13.3 (0.1) & 13.3 (0.1) & 95.0 (0.1) & 52.3 (0.0)\\

 & no & lasso & -0.029 (0.093) & 13.0 (0.1) & 13.0 (0.1) & 94.9 (0.2) & 50.8 (0.0)\\

 & yes & lasso & 0.135 (0.085) & 12.8 (0.1) & 12.8 (0.1) & 95.0 (0.2) & 50.8 (0.0)\\

 & no & lasso+ols & -0.051 (0.093) & 12.8 (0.1) & 12.8 (0.1) & 92.3 (0.2) & 45.5 (0.0)\\

\multirow{-8}{1.4cm}{\centering\arraybackslash Large, equal} & yes & lasso+ols & 0.080 (0.088) & 12.6 (0.1) & 12.6 (0.1) & 92.6 (0.2) & 45.5 (0.0)\\
\midrule

 & no & unadj & -0.150 (0.111) & 15.5 (0.1) & 15.5 (0.1) & 94.9 (0.2) & 61.0 (0.0)\\

 & yes & unadj & -0.097 (0.108) & 15.5 (0.1) & 15.5 (0.1) & 94.3 (0.2) & 59.6 (0.0)\\

 & no & ols & 0.750 (0.114) & 15.5 (0.1) & 15.5 (0.1) & 94.6 (0.2) & 60.3 (0.0)\\

 & yes & ols & -0.071 (0.115) & 15.5 (0.1) & 15.5 (0.1) & 94.6 (0.2) & 60.3 (0.0)\\

 & no & lasso & -0.330 (0.108) & 15.1 (0.1) & 15.1 (0.1) & 94.9 (0.2) & 59.2 (0.0)\\

 & yes & lasso & -0.241 (0.106) & 15.0 (0.1) & 15.0 (0.1) & 94.9 (0.2) & 59.2 (0.0)\\

 & no & lasso+ols & -0.472 (0.108) & 14.9 (0.1) & 14.9 (0.1) & 94.0 (0.2) & 56.7 (0.0)\\

\multirow{-8}{1.4cm}{\centering\arraybackslash Large, unequal} & yes & lasso+ols & -0.341 (0.103) & 14.9 (0.1) & 14.9 (0.1) & 94.0 (0.2) & 56.7 (0.0)\\
\midrule

 & no & unadj & -0.204 (0.112) & 15.8 (0.1) & 15.8 (0.1) & 95.0 (0.1) & 62.7 (0.0)\\

 & yes & unadj & 0.048 (0.118) & 15.9 (0.1) & 15.9 (0.1) & 94.7 (0.2) & 61.9 (0.0)\\

 & no & ols & -0.190 (0.110) & 15.9 (0.1) & 15.9 (0.1) & 94.8 (0.2) & 62.5 (0.0)\\

 & yes & ols & 0.049 (0.117) & 15.9 (0.1) & 15.9 (0.1) & 94.9 (0.2) & 62.5 (0.0)\\

 & no & lasso & -0.203 (0.110) & 14.9 (0.1) & 14.9 (0.1) & 94.9 (0.2) & 58.8 (0.0)\\

 & yes & lasso & 0.050 (0.102) & 15.1 (0.1) & 15.1 (0.1) & 95.0 (0.2) & 58.8 (0.0)\\

 & no & lasso+ols & -0.157 (0.104) & 14.6 (0.1) & 14.6 (0.1) & 91.5 (0.2) & 50.7 (0.0)\\

\multirow{-8}{1.4cm}{\centering\arraybackslash Small, equal} & yes & lasso+ols & 0.079 (0.110) & 14.6 (0.1) & 14.6 (0.1) & 91.5 (0.2) & 50.8 (0.0)\\
\midrule

 & no & unadj & 0.052 (0.114) & 16.3 (0.1) & 16.3 (0.1) & 94.9 (0.2) & 64.0 (0.0)\\

 & yes & unadj & 0.118 (0.116) & 16.4 (0.1) & 16.4 (0.1) & 94.2 (0.2) & 63.2 (0.0)\\

 & no & ols & 0.329 (0.118) & 16.5 (0.1) & 16.5 (0.1) & 94.8 (0.2) & 64.0 (0.0)\\

 & yes & ols & 0.128 (0.118) & 16.4 (0.1) & 16.4 (0.1) & 94.5 (0.2) & 64.0 (0.0)\\

 & no & lasso & 0.922 (0.110) & 15.4 (0.1) & 15.5 (0.1) & 95.0 (0.2) & 60.4 (0.0)\\

 & yes & lasso & 0.972 (0.109) & 15.6 (0.1) & 15.6 (0.1) & 94.3 (0.2) & 60.4 (0.0)\\

 & no & lasso+ols & 2.149 (0.104) & 15.2 (0.1) & 15.3 (0.1) & 92.3 (0.2) & 55.1 (0.0)\\

\multirow{-8}{1.4cm}{\centering\arraybackslash Small, unequal} & yes & lasso+ols & 2.072 (0.111) & 15.4 (0.1) & 15.5 (0.1) & 92.0 (0.2) & 55.1 (0.0)\\
\bottomrule
\end{tabular}
\begin{tablenotes}
\item Note: The numbers in brackets are the corresponding standard errors estimated using the bootstrap with 500 replications. Bias, SD, RMSE, CP, Length, and their standard errors are multiplied by 100.
\end{tablenotes}
\end{threeparttable}
\end{table}

\begin{table}[p]

\caption{\label{tab:sim_bias_400_2}Simulation results for OLS, Lasso, and Lasso + OLS, when $s=400$.}
\centering
\begin{threeparttable}
\begin{tabular}[t]{>{\centering\arraybackslash}p{1.4cm}ccrrrrr}
\toprule
\multicolumn{1}{c}{Scenario} & \multicolumn{1}{c}{Rerand.} & \multicolumn{1}{c}{Est.} & \multicolumn{1}{c}{Bias} & \multicolumn{1}{c}{SD} & \multicolumn{1}{c}{RMSE} & \multicolumn{1}{c}{CP} & \multicolumn{1}{c}{Length}\\
\midrule
 & no & unadj & -0.210 (0.097) & 14.0 (0.1) & 14.0 (0.1) & 94.7 (0.2) & 54.9 (0.0)\\

 & yes & unadj & -0.036 (0.096) & 14.1 (0.1) & 14.1 (0.1) & 94.5 (0.2) & 54.1 (0.0)\\

 & no & ols & -0.236 (0.100) & 14.1 (0.1) & 14.1 (0.1) & 94.4 (0.2) & 54.6 (0.0)\\

 & yes & ols & -0.035 (0.098) & 14.1 (0.1) & 14.1 (0.1) & 94.7 (0.2) & 54.6 (0.0)\\

 & no & lasso & -0.218 (0.102) & 13.7 (0.1) & 13.7 (0.1) & 94.7 (0.2) & 53.3 (0.0)\\

 & yes & lasso & -0.036 (0.101) & 13.7 (0.1) & 13.7 (0.1) & 94.8 (0.2) & 53.3 (0.0)\\

 & no & lasso+ols & -0.199 (0.098) & 13.6 (0.1) & 13.6 (0.1) & 92.6 (0.2) & 49.0 (0.0)\\

\multirow{-8}{1.4cm}{\centering\arraybackslash Hybrid, equal} & yes & lasso+ols & -0.003 (0.093) & 13.5 (0.1) & 13.5 (0.1) & 92.7 (0.2) & 49.0 (0.0)\\
\midrule

 & no & unadj & 0.167 (0.119) & 17.3 (0.1) & 17.3 (0.1) & 94.9 (0.2) & 67.5 (0.0)\\

 & yes & unadj & 0.046 (0.117) & 16.9 (0.1) & 16.9 (0.1) & 94.5 (0.2) & 65.6 (0.0)\\

 & no & ols & 0.398 (0.128) & 17.2 (0.1) & 17.2 (0.1) & 94.6 (0.2) & 66.4 (0.0)\\

 & yes & ols & 0.049 (0.121) & 16.9 (0.1) & 16.9 (0.1) & 94.7 (0.1) & 66.3 (0.0)\\

 & no & lasso & 1.049 (0.120) & 16.6 (0.1) & 16.6 (0.1) & 94.2 (0.2) & 63.9 (0.0)\\

 & yes & lasso & 0.946 (0.112) & 16.3 (0.1) & 16.3 (0.1) & 94.5 (0.2) & 63.8 (0.0)\\

 & no & lasso+ols & 2.225 (0.118) & 16.5 (0.1) & 16.7 (0.1) & 91.0 (0.2) & 57.8 (0.1)\\

\multirow{-8}{1.4cm}{\centering\arraybackslash Hybrid, unequal} & yes & lasso+ols & 2.045 (0.115) & 16.2 (0.1) & 16.3 (0.1) & 91.3 (0.2) & 57.8 (0.1)\\
\midrule

 & no & unadj & -0.040 (0.104) & 15.4 (0.1) & 15.4 (0.1) & 94.8 (0.2) & 59.9 (0.0)\\

 & yes & unadj & -0.359 (0.100) & 15.1 (0.1) & 15.1 (0.1) & 94.1 (0.2) & 57.6 (0.0)\\

 & no & ols & 0.026 (0.110) & 15.5 (0.1) & 15.5 (0.1) & 94.0 (0.2) & 58.5 (0.0)\\

 & yes & ols & -0.359 (0.111) & 15.1 (0.1) & 15.1 (0.1) & 94.6 (0.2) & 58.5 (0.0)\\

 & no & lasso & -0.038 (0.111) & 15.3 (0.1) & 15.3 (0.1) & 94.3 (0.2) & 58.7 (0.0)\\

 & yes & lasso & -0.358 (0.097) & 15.0 (0.1) & 15.0 (0.1) & 94.7 (0.2) & 58.7 (0.0)\\

 & no & lasso+ols & -0.088 (0.115) & 15.5 (0.1) & 15.5 (0.1) & 93.0 (0.2) & 57.1 (0.0)\\

\multirow{-8}{1.4cm}{\centering\arraybackslash Triplet, equal} & yes & lasso+ols & -0.383 (0.110) & 15.2 (0.1) & 15.2 (0.1) & 93.2 (0.2) & 57.1 (0.0)\\
\midrule

 & no & unadj & -0.088 (0.118) & 16.6 (0.1) & 16.6 (0.1) & 94.5 (0.2) & 64.9 (0.0)\\

 & yes & unadj & 0.054 (0.119) & 16.6 (0.1) & 16.6 (0.1) & 94.0 (0.2) & 63.5 (0.0)\\

 & no & ols & 0.814 (0.118) & 16.8 (0.1) & 16.8 (0.1) & 94.2 (0.2) & 64.3 (0.0)\\

 & yes & ols & 0.078 (0.122) & 16.6 (0.1) & 16.6 (0.1) & 94.3 (0.2) & 64.3 (0.0)\\

 & no & lasso & 0.718 (0.112) & 16.3 (0.1) & 16.3 (0.1) & 93.9 (0.2) & 62.2 (0.0)\\

 & yes & lasso & 0.927 (0.112) & 16.2 (0.1) & 16.3 (0.1) & 94.0 (0.2) & 62.1 (0.0)\\

 & no & lasso+ols & 2.018 (0.117) & 16.3 (0.1) & 16.4 (0.1) & 91.6 (0.2) & 58.6 (0.1)\\

\multirow{-8}{1.4cm}{\centering\arraybackslash Triplet, unequal} & yes & lasso+ols & 2.238 (0.116) & 16.3 (0.1) & 16.5 (0.1) & 91.3 (0.2) & 58.7 (0.1)\\
\bottomrule
\end{tabular}
\begin{tablenotes}
\item Note: The numbers in brackets are the corresponding standard errors estimated using the bootstrap with 500 replications. Bias, SD, RMSE, CP, Length, and their standard errors are multiplied by 100.
\end{tablenotes}
\end{threeparttable}
\end{table}

\subsubsection{Choice of tuning parameter}

In this section, we examine the performance of the Lasso-adjusted ATE estimator for different choices of the tuning parameter $\lambda$. We consider a scenario where $s=40$ is not very small. We also set $p=400$ and $n=300$. We generate potential outcomes by $Y_i(1)=Y_i(0)=B_i/M+\xr_i^\T \boldsymbol{\beta}-2\xr_{bc,i}^\T \boldsymbol{\beta} + \e_i$, where the first $s$ elements of $\boldsymbol{\beta}$ are generated from $U[-0.1,0.1]$ and the remaining elements are zero. The other setups closely resemble those described in Section \ref{sec:sim}.

For each data set, we utilize 100 values of $\lambda$ determined by the R function glmnet() as tuning parameters for the Lasso. Additionally, we employ the ``1se'' criterion in the R function cv.glmnet() to select the optimal tuning parameter for the Lasso. For each $\lambda$, we conduct 1000 simulations, perform inference, and calculate the corresponding RMSE, interval length, and coverage probability.

Figure \ref{fig:sim_lambda_40} displays the results. The RMSE and coverage probability are not sensitive to the choice of the tuning parameter $\lambda$ around the value selected by the ``1se'' criterion, while the interval length is more sensitive to the choice of $\lambda$.  Using the ``1se'' criterion to select $\lambda$ could ensure nominal coverage and reasonable efficiency.

\begin{figure}[p]
\centering
\includegraphics[width=1\textwidth]{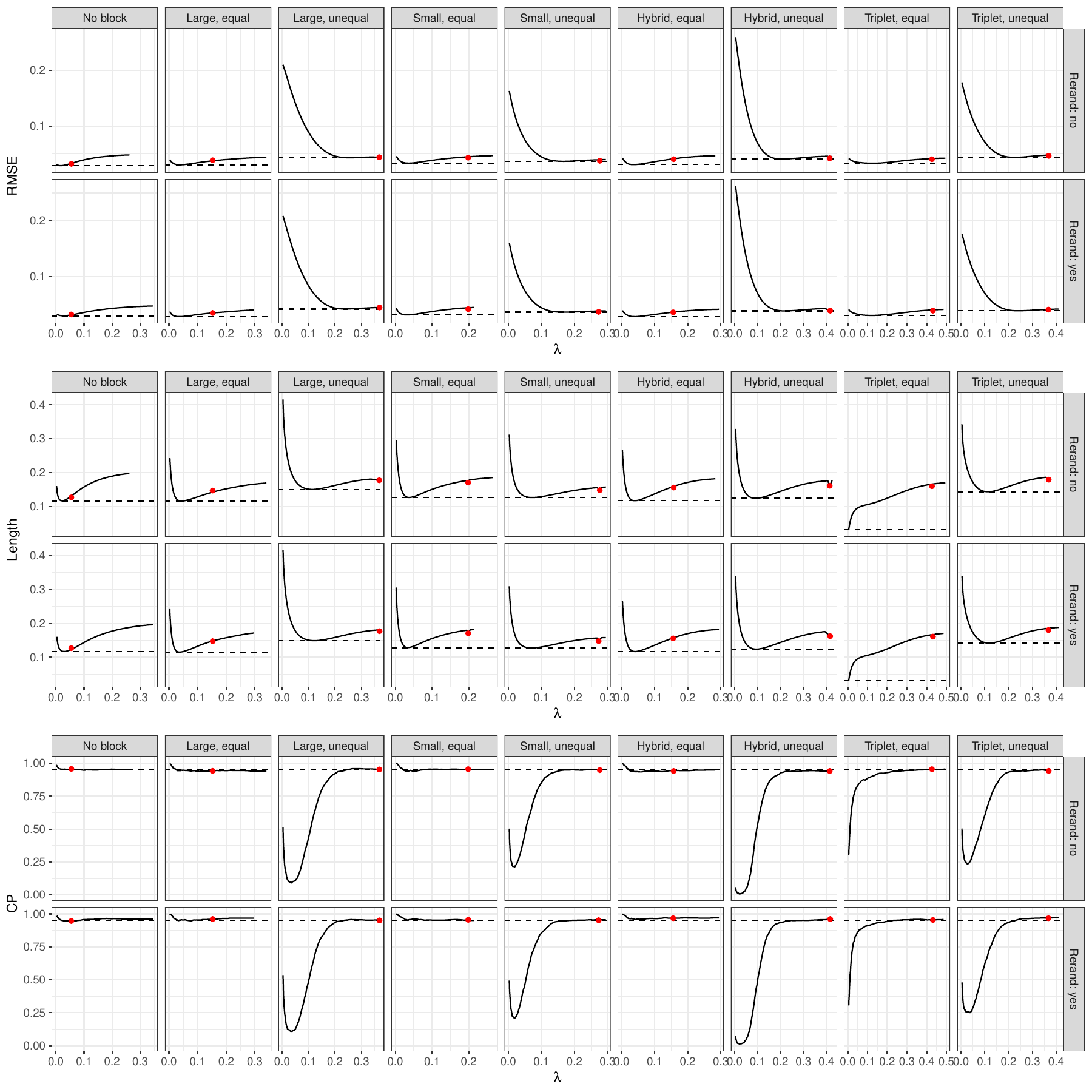}
\caption{Root Mean Squared Error (RMSE) of the Lasso-adjusted ATE estimator for various $\lambda$ values. The red point denotes the $\lambda$ chosen by the "1se" criterion.}
\label{fig:sim_lambda_40}
\end{figure}

\clearpage

\subsubsection{Lasso with forced adjustment for covariates in $\mathcal{K}$}

\zk{In this section, we consider the same settings as those in Section \ref{sec:sim}, while examining another approach to improve the finite sample performance of Lasso.
Specifically, we force Lasso to adjust all covariates in $\mathcal{K}$ by setting their corresponding $l_1$ penalties as 0.
Table~\ref{tab:sim_force} shows the results. 
The Lasso with forced adjustment does reduce 3\%--10\% RMSE compared to the Lasso alone. However, it is not as efficient as the Lasso with rerandomization in the design stage, which reduces 9\%--33\% RMSE compared to the Lasso alone.
}

\begin{table}[H]
\centering
\caption{\label{tab:sim_force}Simulation results for forcing Lasso to adjust all covariates in $\mathcal{K}$ by setting their corresponding $l_1$ penalties as 0.}
\centering
\begin{threeparttable}
\begin{tabular}[t]{>{\centering\arraybackslash}p{1.4cm}ccrrrrr}
\toprule
\multicolumn{1}{c}{Scenario} & \multicolumn{1}{c}{Rerand.} & \multicolumn{1}{c}{Est.} & \multicolumn{1}{c}{Bias} & \multicolumn{1}{c}{SD} & \multicolumn{1}{c}{RMSE} & \multicolumn{1}{c}{CP} & \multicolumn{1}{c}{Length}\\
\midrule
 & no & lasso & 0.1 (0.4) & 11.5 (0.2) & 11.5 (0.2) & 99.8 (0.1) & 65.9 (0.1)\\

 & yes & lasso & -0.1 (0.3) & 10.4 (0.2) & 10.4 (0.2) & 99.6 (0.2) & 66.0 (0.1)\\

 & no & lasso (force) & 0.2 (0.3) & 11.0 (0.2) & 11.0 (0.2) & 99.5 (0.2) & 64.3 (0.1)\\

\multirow{-4}{1.4cm}{\centering\arraybackslash No block} & yes & lasso (force) & -0.2 (0.3) & 11.1 (0.3) & 11.1 (0.3) & 99.4 (0.2) & 64.4 (0.1)\\
\midrule

 & no & lasso & 0.7 (0.5) & 16.4 (0.4) & 16.5 (0.4) & 98.2 (0.4) & 81.6 (0.1)\\

 & yes & lasso & 0.5 (0.5) & 14.4 (0.3) & 14.4 (0.3) & 99.9 (0.1) & 81.7 (0.1)\\

 & no & lasso (force) & 0.5 (0.5) & 15.5 (0.3) & 15.5 (0.3) & 99.0 (0.3) & 79.1 (0.1)\\

\multirow{-4}{1.4cm}{\centering\arraybackslash Large, equal} & yes & lasso (force) & 0.6 (0.5) & 15.2 (0.3) & 15.2 (0.3) & 99.5 (0.2) & 79.2 (0.1)\\
\midrule

 & no & lasso & 1.3 (0.8) & 25.1 (0.6) & 25.1 (0.6) & 97.1 (0.5) & 115.2 (0.3)\\

 & yes & lasso & 0.1 (0.6) & 20.1 (0.5) & 20.1 (0.5) & 99.4 (0.2) & 114.8 (0.2)\\

 & no & lasso (force) & 0.9 (0.7) & 23.1 (0.5) & 23.1 (0.5) & 98.1 (0.4) & 109.9 (0.2)\\

\multirow{-4}{1.4cm}{\centering\arraybackslash Large, unequal} & yes & lasso (force) & 0.1 (0.7) & 22.9 (0.5) & 22.9 (0.5) & 97.7 (0.5) & 109.4 (0.2)\\
\midrule

 & no & lasso & 0.2 (0.4) & 14.4 (0.3) & 14.4 (0.3) & 98.8 (0.3) & 73.0 (0.1)\\

 & yes & lasso & 0.5 (0.4) & 11.4 (0.2) & 11.4 (0.2) & 100.0 (0.0) & 73.2 (0.1)\\

 & no & lasso (force) & 0.1 (0.4) & 14.0 (0.3) & 14.0 (0.3) & 98.9 (0.3) & 72.0 (0.1)\\

\multirow{-4}{1.4cm}{\centering\arraybackslash Small, equal} & yes & lasso (force) & 0.4 (0.4) & 13.6 (0.3) & 13.6 (0.3) & 99.1 (0.3) & 72.1 (0.1)\\
\midrule

 & no & lasso & 0.9 (0.5) & 15.1 (0.3) & 15.1 (0.3) & 97.2 (0.5) & 66.7 (0.1)\\

 & yes & lasso & 0.8 (0.4) & 11.8 (0.3) & 11.8 (0.3) & 99.3 (0.3) & 66.8 (0.1)\\

 & no & lasso (force) & 0.7 (0.4) & 13.9 (0.3) & 13.9 (0.3) & 97.3 (0.5) & 63.2 (0.1)\\

\multirow{-4}{1.4cm}{\centering\arraybackslash Small, unequal} & yes & lasso (force) & 0.4 (0.5) & 13.7 (0.3) & 13.7 (0.3) & 96.9 (0.6) & 63.2 (0.1)\\
\midrule

 & no & lasso & 0.3 (0.4) & 13.5 (0.3) & 13.5 (0.3) & 99.7 (0.2) & 81.9 (0.1)\\

 & yes & lasso & 0.2 (0.3) & 11.4 (0.2) & 11.4 (0.2) & 100.0 (0.0) & 81.7 (0.1)\\

 & no & lasso (force) & 0.1 (0.4) & 12.8 (0.3) & 12.8 (0.3) & 99.9 (0.1) & 81.3 (0.1)\\

\multirow{-4}{1.4cm}{\centering\arraybackslash Hybrid, equal} & yes & lasso (force) & 0.2 (0.4) & 12.8 (0.3) & 12.8 (0.3) & 99.7 (0.2) & 81.2 (0.1)\\
\midrule

 & no & lasso & 2.6 (0.5) & 16.9 (0.4) & 17.1 (0.4) & 97.0 (0.5) & 79.4 (0.1)\\

 & yes & lasso & 1.1 (0.4) & 13.3 (0.3) & 13.4 (0.3) & 99.5 (0.2) & 79.4 (0.1)\\

 & no & lasso (force) & 2.7 (0.5) & 15.9 (0.4) & 16.2 (0.3) & 98.3 (0.4) & 77.6 (0.2)\\

\multirow{-4}{1.4cm}{\centering\arraybackslash Hybrid, unequal} & yes & lasso (force) & 0.8 (0.5) & 15.4 (0.4) & 15.4 (0.4) & 98.4 (0.4) & 77.6 (0.1)\\
\midrule

 & no & lasso & 0.9 (0.4) & 13.3 (0.4) & 13.3 (0.4) & 97.7 (0.5) & 59.8 (0.2)\\

 & yes & lasso & 0.4 (0.3) & 9.6 (0.3) & 9.6 (0.3) & 99.9 (0.1) & 59.7 (0.2)\\

 & no & lasso (force) & 1.1 (0.4) & 13.0 (0.3) & 13.0 (0.3) & 97.0 (0.5) & 56.0 (0.2)\\

\multirow{-4}{1.4cm}{\centering\arraybackslash Triplet, equal} & yes & lasso (force) & 0.5 (0.4) & 12.5 (0.3) & 12.5 (0.3) & 97.6 (0.5) & 55.7 (0.2)\\
\midrule

 & no & lasso & 1.7 (0.5) & 16.5 (0.3) & 16.5 (0.3) & 96.0 (0.6) & 66.8 (0.2)\\

 & yes & lasso & 1.2 (0.4) & 11.1 (0.3) & 11.1 (0.3) & 99.7 (0.2) & 67.3 (0.2)\\

 & no & lasso (force) & 1.4 (0.5) & 14.8 (0.3) & 14.9 (0.3) & 95.6 (0.7) & 60.8 (0.2)\\

\multirow{-4}{1.4cm}{\centering\arraybackslash Triplet, unequal} & yes & lasso (force) & 0.4 (0.5) & 14.9 (0.3) & 14.9 (0.3) & 95.3 (0.6) & 60.9 (0.2)\\
\bottomrule
\end{tabular}
\begin{tablenotes}
\item Note: The numbers in brackets are the corresponding standard errors estimated using the bootstrap with 500 replications. Bias, SD, RMSE, CP, Length, and their standard errors are multiplied by 100.
\end{tablenotes}
\end{threeparttable}
\end{table}

\subsection{Performance under various setups}

\subsubsection{Block-specific coefficients}

In this section, we consider block-specific coefficients for the relationship between potential outcomes and covariates in the data-generating process. Specifically, the potential outcomes are generated by the following equation: $Y_i(z)=(B_i/M)^{2z+1}+\xr_i^\T \boldsymbol{\beta}_{[m]}(z)-2\xr_{bc,i}^\T \boldsymbol{\beta}_{[m]}(z) + \e_i(z)$, $i=1,...,n$, $z = 0, 1$, $i\in [m]$. Here, the first $s$ elements of $\boldsymbol{\beta}_{[m]}(z)$ are generated from $U(0,1)$ and the remaining elements are zero. We set $n=300$. The remaining configurations are the same as those detailed in Section \ref{sec:sim}. Table \ref{tab:sim_heter} shows the simulation results. All methods yield consistent and rational results, supporting the same conclusion presented in the main text.

\begin{table}[p]

\caption{\label{tab:sim_heter}Simulation results for block-specific coefficients.}
\centering
\begin{threeparttable}
\begin{tabular}[t]{>{\centering\arraybackslash}p{1.4cm}ccrrrrr}
\toprule
\multicolumn{1}{c}{Scenario} & \multicolumn{1}{c}{Rerand.} & \multicolumn{1}{c}{Est.} & \multicolumn{1}{c}{Bias} & \multicolumn{1}{c}{SD} & \multicolumn{1}{c}{RMSE} & \multicolumn{1}{c}{CP} & \multicolumn{1}{c}{Length}\\
\midrule
 & no & unadj & -0.5 (1.0) & 31.6 (0.8) & 31.6 (0.8) & 96.6 (0.6) & 135.6 (0.1)\\

 & yes & unadj & -0.1 (0.5) & 16.8 (0.4) & 16.8 (0.4) & 98.4 (0.4) & 82.9 (0.1)\\

 & no & lasso & 0.2 (0.4) & 12.5 (0.3) & 12.5 (0.3) & 99.6 (0.2) & 73.7 (0.1)\\

\multirow{-4}{1.4cm}{\centering\arraybackslash No block} & yes & lasso & -0.4 (0.3) & 11.1 (0.3) & 11.2 (0.3) & 99.9 (0.1) & 74.0 (0.1)\\
\midrule

 & no & unadj & -1.4 (0.8) & 27.9 (0.6) & 27.9 (0.6) & 96.0 (0.6) & 118.1 (0.1)\\

 & yes & unadj & 0.7 (0.7) & 20.7 (0.5) & 20.7 (0.5) & 97.5 (0.5) & 89.1 (0.1)\\

 & no & lasso & -1.3 (0.4) & 11.9 (0.3) & 12.0 (0.3) & 99.0 (0.3) & 62.6 (0.1)\\

\multirow{-4}{1.4cm}{\centering\arraybackslash Large, equal} & yes & lasso & 0.0 (0.3) & 10.5 (0.2) & 10.5 (0.2) & 99.8 (0.1) & 62.5 (0.1)\\
\midrule

 & no & unadj & 1.0 (1.2) & 41.7 (1.0) & 41.7 (1.0) & 97.1 (0.5) & 186.8 (0.4)\\

 & yes & unadj & -0.9 (1.0) & 31.6 (0.6) & 31.6 (0.6) & 96.5 (0.6) & 134.7 (0.3)\\

 & no & lasso & 2.1 (0.8) & 25.5 (0.6) & 25.6 (0.6) & 97.6 (0.5) & 123.4 (0.3)\\

\multirow{-4}{1.4cm}{\centering\arraybackslash Large, unequal} & yes & lasso & -0.0 (0.7) & 22.2 (0.5) & 22.2 (0.5) & 99.9 (0.1) & 123.5 (0.3)\\
\midrule

 & no & unadj & 1.0 (1.1) & 34.9 (0.8) & 34.9 (0.8) & 95.5 (0.6) & 142.7 (0.1)\\

 & yes & unadj & 0.1 (0.7) & 21.2 (0.5) & 21.2 (0.5) & 96.9 (0.5) & 92.6 (0.1)\\

 & no & lasso & 0.8 (0.5) & 16.5 (0.4) & 16.6 (0.4) & 98.5 (0.4) & 77.8 (0.1)\\

\multirow{-4}{1.4cm}{\centering\arraybackslash Small, equal} & yes & lasso & 0.0 (0.4) & 14.0 (0.3) & 14.0 (0.3) & 99.4 (0.2) & 77.7 (0.1)\\
\midrule

 & no & unadj & -1.2 (1.2) & 36.6 (0.8) & 36.6 (0.8) & 94.6 (0.7) & 149.0 (0.2)\\

 & yes & unadj & 0.6 (0.8) & 24.3 (0.6) & 24.3 (0.6) & 94.7 (0.7) & 100.0 (0.2)\\

 & no & lasso & 1.3 (0.6) & 18.9 (0.4) & 19.0 (0.4) & 96.9 (0.5) & 84.2 (0.1)\\

\multirow{-4}{1.4cm}{\centering\arraybackslash Small, unequal} & yes & lasso & 1.3 (0.5) & 15.0 (0.4) & 15.0 (0.4) & 99.4 (0.2) & 84.2 (0.2)\\
\midrule

 & no & unadj & 0.7 (1.2) & 38.2 (0.8) & 38.2 (0.8) & 97.1 (0.5) & 162.3 (0.1)\\

 & yes & unadj & -0.7 (0.8) & 23.7 (0.5) & 23.7 (0.5) & 97.3 (0.5) & 104.8 (0.1)\\

 & no & lasso & 0.2 (0.5) & 16.2 (0.3) & 16.2 (0.3) & 98.7 (0.4) & 83.5 (0.1)\\

\multirow{-4}{1.4cm}{\centering\arraybackslash Hybrid, equal} & yes & lasso & -0.1 (0.4) & 14.1 (0.3) & 14.1 (0.3) & 99.2 (0.3) & 83.4 (0.1)\\
\midrule

 & no & unadj & 1.4 (1.3) & 40.2 (0.9) & 40.3 (0.9) & 95.9 (0.6) & 166.0 (0.2)\\

 & yes & unadj & 1.6 (0.8) & 27.7 (0.6) & 27.7 (0.6) & 95.5 (0.7) & 114.3 (0.2)\\

 & no & lasso & 2.7 (0.7) & 20.3 (0.4) & 20.5 (0.4) & 97.8 (0.4) & 93.4 (0.2)\\

\multirow{-4}{1.4cm}{\centering\arraybackslash Hybrid, unequal} & yes & lasso & 2.2 (0.5) & 17.1 (0.4) & 17.2 (0.4) & 99.4 (0.3) & 93.4 (0.2)\\
\midrule

 & no & unadj & -1.6 (1.3) & 41.2 (0.9) & 41.2 (0.9) & 94.2 (0.7) & 158.2 (0.3)\\

 & yes & unadj & -0.1 (0.8) & 25.8 (0.6) & 25.8 (0.6) & 94.1 (0.8) & 102.1 (0.4)\\

 & no & lasso & -1.0 (0.6) & 20.1 (0.5) & 20.1 (0.5) & 94.7 (0.7) & 80.0 (0.2)\\

\multirow{-4}{1.4cm}{\centering\arraybackslash Triplet, equal} & yes & lasso & 0.1 (0.5) & 15.0 (0.4) & 15.0 (0.4) & 98.9 (0.3) & 79.9 (0.2)\\
\midrule

 & no & unadj & 0.7 (1.2) & 37.3 (0.8) & 37.2 (0.8) & 94.4 (0.7) & 147.2 (0.3)\\

 & yes & unadj & 1.0 (0.7) & 22.2 (0.5) & 22.2 (0.5) & 93.7 (0.7) & 87.8 (0.5)\\

 & no & lasso & 1.7 (0.7) & 23.8 (0.6) & 23.9 (0.6) & 95.7 (0.7) & 96.2 (0.3)\\

\multirow{-4}{1.4cm}{\centering\arraybackslash Triplet, unequal} & yes & lasso & 1.4 (0.5) & 16.5 (0.4) & 16.5 (0.4) & 99.5 (0.2) & 96.3 (0.3)\\
\bottomrule
\end{tabular}
\begin{tablenotes}
\item Note: The numbers in brackets are the corresponding standard errors estimated using the bootstrap with 500 replications. Bias, SD, RMSE, CP, Length, and their standard errors are multiplied by 100.
\end{tablenotes}
\end{threeparttable}
\end{table}

\subsubsection{Constant treatment effect}

In this section, we investigate the finite sample performance of the proposed variance estimator. Our theoretical analysis indicates that the proposed variance estimator is generally conservative. It is consistent when the treatment effects within each coarse block remain constant, and the average treatment effects across fine blocks also remain constant. These conditions are satisfied when $Y_i(0) = Y_i(1)$ for all $i$. Therefore, we set $Y_i(0) = Y_i(1)$ for all $i$, while maintaining the other settings consistent with those in Section \ref{sec:sim}.
Table \ref{tab:sim_varest} displays the results. The empirical coverage probabilities of all methods closely align with the nominal level of 95\%. In some instances, when rerandomization is employed, the coverage may slightly exceed 95\%, reaching around 97\%.

% In this section, we investigate the finite sample performance of the proposed variance estimator. Our theoretical analysis suggests that the proposed variance estimator is generally conservative. Specifically, it achieves consistency when the treatment effects within each coarse block remain constant, and the average treatment effects across fine blocks also remain constant. These conditions are satisfied when $Y_i(0) = Y_i(1)$ for all $i$. Therefore, we set $Y_i(0) = Y_i(1)$ for all $i$, while keeping the other settings consistent with those in Section \ref{sec:sim}.
% Table \ref{tab:sim_varest} displays the results. The empirical coverage probabilities of all methods closely align with the nominal level of 95\%. In some instances, when rerandomization is employed, the coverage may slightly exceed 95\%, reaching around 97\%.

\begin{table}[p]

\caption{\label{tab:sim_varest}Simulation results in scenarios where $Y_i(1) = Y_i(0)$ for all $i$.}
\centering
\begin{threeparttable}
\begin{tabular}[t]{>{\centering\arraybackslash}p{1.4cm}ccrrrrr}
\toprule
\multicolumn{1}{c}{Scenario} & \multicolumn{1}{c}{Rerand.} & \multicolumn{1}{c}{Est.} & \multicolumn{1}{c}{Bias} & \multicolumn{1}{c}{SD} & \multicolumn{1}{c}{RMSE} & \multicolumn{1}{c}{CP} & \multicolumn{1}{c}{Length}\\
\midrule
 & no & unadj & -0.1 (1.2) & 38.8 (0.9) & 38.7 (0.9) & 94.7 (0.7) & 149.6 (0.0)\\

 & yes & unadj & -0.4 (0.7) & 22.7 (0.5) & 22.7 (0.5) & 94.1 (0.8) & 85.1 (0.0)\\

 & no & lasso & 0.5 (0.6) & 17.9 (0.4) & 17.9 (0.4) & 94.3 (0.8) & 69.6 (0.0)\\

\multirow{-4}{1.4cm}{\centering\arraybackslash No block} & yes & lasso & -0.2 (0.5) & 16.8 (0.4) & 16.8 (0.4) & 97.2 (0.5) & 69.7 (0.0)\\
\midrule

 & no & unadj & -0.7 (1.0) & 32.2 (0.7) & 32.2 (0.7) & 95.1 (0.7) & 129.4 (0.0)\\

 & yes & unadj & 0.9 (0.7) & 22.6 (0.5) & 22.6 (0.5) & 94.6 (0.7) & 85.7 (0.0)\\

 & no & lasso & -0.1 (0.5) & 16.3 (0.4) & 16.3 (0.4) & 94.3 (0.7) & 62.1 (0.0)\\

\multirow{-4}{1.4cm}{\centering\arraybackslash Large, equal} & yes & lasso & 0.8 (0.5) & 15.1 (0.3) & 15.1 (0.3) & 95.7 (0.7) & 62.1 (0.0)\\
\midrule

 & no & unadj & -0.6 (1.0) & 34.7 (0.8) & 34.7 (0.8) & 94.0 (0.8) & 135.2 (0.2)\\

 & yes & unadj & -0.6 (1.0) & 31.6 (0.7) & 31.6 (0.7) & 94.2 (0.8) & 119.7 (0.1)\\

 & no & lasso & 0.1 (0.6) & 19.0 (0.4) & 19.0 (0.4) & 94.4 (0.8) & 74.1 (0.2)\\

\multirow{-4}{1.4cm}{\centering\arraybackslash Large, unequal} & yes & lasso & -0.3 (0.5) & 16.9 (0.3) & 16.9 (0.3) & 97.3 (0.5) & 74.6 (0.2)\\
\midrule

 & no & unadj & 0.3 (1.3) & 41.3 (0.9) & 41.3 (0.9) & 95.2 (0.7) & 163.4 (0.1)\\

 & yes & unadj & -0.1 (0.9) & 27.0 (0.6) & 27.0 (0.6) & 94.6 (0.7) & 104.5 (0.1)\\

 & no & lasso & -0.4 (0.5) & 16.7 (0.4) & 16.7 (0.4) & 95.5 (0.7) & 67.9 (0.1)\\

\multirow{-4}{1.4cm}{\centering\arraybackslash Small, equal} & yes & lasso & -0.5 (0.5) & 16.1 (0.4) & 16.1 (0.4) & 95.9 (0.6) & 67.8 (0.1)\\
\midrule

 & no & unadj & -0.5 (1.3) & 40.6 (0.9) & 40.6 (0.9) & 96.2 (0.6) & 165.1 (0.2)\\

 & yes & unadj & 0.5 (0.8) & 27.5 (0.6) & 27.5 (0.6) & 92.6 (0.8) & 101.1 (0.3)\\

 & no & lasso & -0.5 (0.6) & 19.3 (0.4) & 19.3 (0.4) & 96.0 (0.6) & 76.1 (0.1)\\

\multirow{-4}{1.4cm}{\centering\arraybackslash Small, unequal} & yes & lasso & 0.2 (0.5) & 17.0 (0.4) & 17.0 (0.4) & 97.3 (0.5) & 76.0 (0.1)\\
\midrule

 & no & unadj & 0.6 (1.3) & 39.6 (0.9) & 39.6 (0.9) & 93.6 (0.8) & 153.7 (0.1)\\

 & yes & unadj & 0.4 (0.9) & 25.5 (0.6) & 25.5 (0.6) & 94.9 (0.7) & 97.6 (0.1)\\

 & no & lasso & 0.3 (0.5) & 17.7 (0.4) & 17.7 (0.4) & 95.5 (0.7) & 68.8 (0.1)\\

\multirow{-4}{1.4cm}{\centering\arraybackslash Hybrid, equal} & yes & lasso & 0.4 (0.5) & 16.7 (0.4) & 16.7 (0.4) & 95.8 (0.6) & 68.8 (0.1)\\
\midrule

 & no & unadj & 1.6 (1.6) & 49.3 (1.0) & 49.3 (1.0) & 96.2 (0.6) & 200.7 (0.2)\\

 & yes & unadj & 0.8 (1.1) & 38.1 (0.9) & 38.0 (0.9) & 94.5 (0.7) & 144.8 (0.2)\\

 & no & lasso & 1.8 (0.8) & 25.4 (0.6) & 25.5 (0.6) & 94.4 (0.7) & 98.6 (0.2)\\

\multirow{-4}{1.4cm}{\centering\arraybackslash Hybrid, unequal} & yes & lasso & 1.0 (0.7) & 22.7 (0.4) & 22.7 (0.5) & 96.6 (0.6) & 98.8 (0.2)\\
\midrule

 & no & unadj & 0.8 (0.9) & 31.6 (0.7) & 31.6 (0.7) & 94.4 (0.7) & 121.4 (0.2)\\

 & yes & unadj & -0.1 (0.7) & 21.8 (0.4) & 21.7 (0.4) & 92.7 (0.8) & 81.8 (0.3)\\

 & no & lasso & 0.2 (0.4) & 13.9 (0.3) & 13.9 (0.3) & 93.7 (0.7) & 51.4 (0.2)\\

\multirow{-4}{1.4cm}{\centering\arraybackslash Triplet, equal} & yes & lasso & 0.0 (0.3) & 10.2 (0.2) & 10.1 (0.2) & 98.4 (0.4) & 51.6 (0.2)\\
\midrule

 & no & unadj & 1.3 (1.2) & 40.0 (1.0) & 40.0 (1.0) & 94.6 (0.8) & 158.4 (0.2)\\

 & yes & unadj & 0.6 (0.7) & 22.7 (0.5) & 22.7 (0.5) & 91.5 (0.9) & 82.9 (0.5)\\

 & no & lasso & 3.2 (0.6) & 20.7 (0.5) & 20.9 (0.5) & 95.5 (0.7) & 81.4 (0.3)\\

\multirow{-4}{1.4cm}{\centering\arraybackslash Triplet, unequal} & yes & lasso & 1.8 (0.5) & 14.0 (0.3) & 14.1 (0.3) & 99.2 (0.3) & 81.3 (0.3)\\
\bottomrule
\end{tabular}
\begin{tablenotes}
\item Note: The numbers in brackets are the corresponding standard errors estimated using the bootstrap with 500 replications. Bias, SD, RMSE, CP, Length, and their standard errors are multiplied by 100.
\end{tablenotes}
\end{threeparttable}
\end{table}

\subsubsection{Larger sample size}

In this section, we consider a larger sample size $n=600$, while keeping the other settings same as those in Section \ref{sec:sim}.
Figure~\ref{fig:sim_dist600} and Table~\ref{tab:sim_600} show the simulation results for $n=600$. The findings align closely with those discussed in the main text.

\begin{figure}[p]
\centering
\includegraphics[scale=0.6]{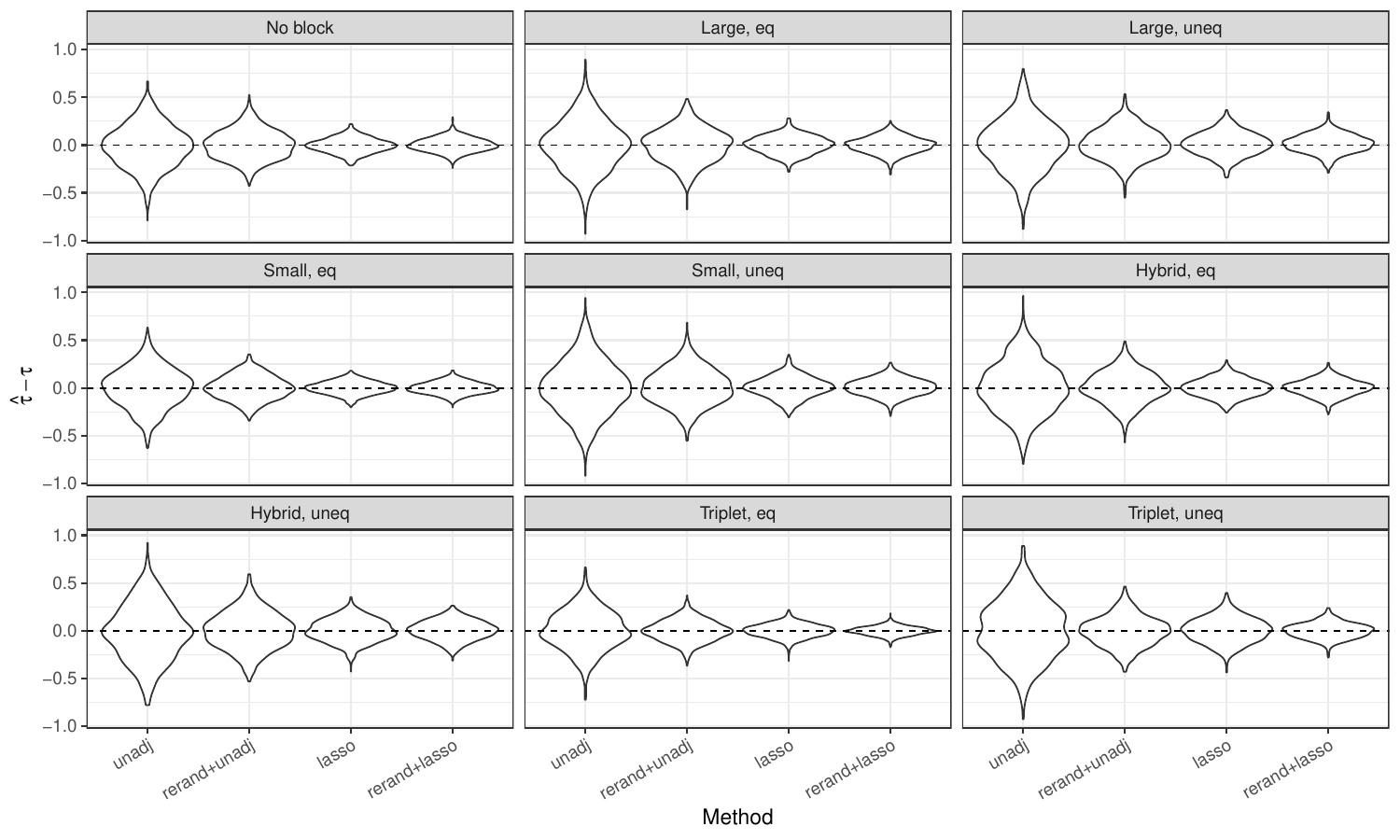}
\caption{Distributions of the average treatment effect estimators minus the true value of the average treatment effect for different scenarios when $n=600$.}
\label{fig:sim_dist600}
\end{figure}

\begin{table}[p]

\caption{\label{tab:sim_600}Simulation results for different scenarios when $n=600$.}
\centering
\begin{threeparttable}
\begin{tabular}[t]{>{\centering\arraybackslash}p{1.4cm}ccrrrrr}
\toprule
\multicolumn{1}{c}{Scenario} & \multicolumn{1}{c}{Rerand.} & \multicolumn{1}{c}{Est.} & \multicolumn{1}{c}{Bias} & \multicolumn{1}{c}{SD} & \multicolumn{1}{c}{RMSE} & \multicolumn{1}{c}{CP} & \multicolumn{1}{c}{Length}\\
\midrule
 & no & unadj & -1.0 (0.7) & 21.8 (0.5) & 21.8 (0.5) & 97.0 (0.6) & 93.7 (0.1)\\

 & yes & unadj & 0.5 (0.5) & 15.1 (0.3) & 15.1 (0.3) & 97.5 (0.5) & 68.4 (0.0)\\

 & no & lasso & -0.4 (0.3) & 7.7 (0.2) & 7.7 (0.2) & 100.0 (0.0) & 47.6 (0.0)\\

\multirow{-4}{1.4cm}{\centering\arraybackslash No block} & yes & lasso & 0.2 (0.2) & 7.2 (0.2) & 7.2 (0.2) & 99.6 (0.2) & 47.5 (0.0)\\
\midrule

 & no & unadj & 0.3 (0.8) & 25.7 (0.6) & 25.7 (0.6) & 96.8 (0.5) & 109.3 (0.0)\\

 & yes & unadj & -0.4 (0.5) & 17.6 (0.4) & 17.6 (0.4) & 96.9 (0.5) & 76.7 (0.0)\\

 & no & lasso & 0.4 (0.3) & 9.3 (0.2) & 9.3 (0.2) & 99.2 (0.3) & 52.3 (0.0)\\

\multirow{-4}{1.4cm}{\centering\arraybackslash Large, equal} & yes & lasso & 0.0 (0.3) & 8.5 (0.2) & 8.5 (0.2) & 99.6 (0.2) & 52.4 (0.0)\\
\midrule

 & no & unadj & 0.3 (0.9) & 26.5 (0.6) & 26.4 (0.6) & 94.1 (0.8) & 105.7 (0.1)\\

 & yes & unadj & 1.4 (0.5) & 16.2 (0.3) & 16.2 (0.3) & 97.7 (0.5) & 70.9 (0.1)\\

 & no & lasso & 1.3 (0.4) & 12.0 (0.3) & 12.1 (0.3) & 97.4 (0.5) & 55.6 (0.1)\\

\multirow{-4}{1.4cm}{\centering\arraybackslash Large, unequal} & yes & lasso & 1.6 (0.3) & 9.6 (0.2) & 9.7 (0.2) & 99.2 (0.3) & 55.7 (0.1)\\
\midrule

 & no & unadj & 0.0 (0.7) & 20.5 (0.4) & 20.5 (0.4) & 95.9 (0.6) & 83.1 (0.0)\\

 & yes & unadj & 0.2 (0.4) & 12.3 (0.2) & 12.3 (0.2) & 97.6 (0.5) & 53.0 (0.0)\\

 & no & lasso & 0.1 (0.2) & 6.7 (0.1) & 6.7 (0.1) & 99.8 (0.1) & 37.1 (0.0)\\

\multirow{-4}{1.4cm}{\centering\arraybackslash Small, equal} & yes & lasso & -0.2 (0.2) & 6.2 (0.1) & 6.2 (0.1) & 99.9 (0.1) & 37.1 (0.0)\\
\midrule

 & no & unadj & 0.3 (0.8) & 27.4 (0.6) & 27.4 (0.6) & 96.2 (0.7) & 115.5 (0.1)\\

 & yes & unadj & 1.2 (0.6) & 18.7 (0.4) & 18.7 (0.4) & 97.0 (0.5) & 80.7 (0.1)\\

 & no & lasso & 0.8 (0.3) & 10.7 (0.2) & 10.8 (0.2) & 98.8 (0.3) & 55.9 (0.1)\\

\multirow{-4}{1.4cm}{\centering\arraybackslash Small, unequal} & yes & lasso & 0.7 (0.3) & 8.7 (0.2) & 8.8 (0.2) & 99.9 (0.1) & 55.8 (0.1)\\
\midrule

 & no & unadj & 0.6 (0.9) & 26.6 (0.6) & 26.6 (0.6) & 95.9 (0.6) & 108.6 (0.0)\\

 & yes & unadj & 0.3 (0.5) & 16.4 (0.4) & 16.4 (0.4) & 96.6 (0.6) & 69.5 (0.1)\\

 & no & lasso & 0.3 (0.3) & 9.3 (0.2) & 9.3 (0.2) & 99.5 (0.2) & 48.4 (0.0)\\

\multirow{-4}{1.4cm}{\centering\arraybackslash Hybrid, equal} & yes & lasso & 0.1 (0.3) & 8.4 (0.2) & 8.4 (0.2) & 99.4 (0.2) & 48.5 (0.0)\\
\midrule

 & no & unadj & -1.2 (0.9) & 29.0 (0.6) & 29.0 (0.6) & 96.2 (0.6) & 117.6 (0.1)\\

 & yes & unadj & 0.2 (0.6) & 19.4 (0.4) & 19.4 (0.4) & 94.7 (0.7) & 79.9 (0.1)\\

 & no & lasso & 0.3 (0.4) & 11.9 (0.3) & 11.9 (0.3) & 97.7 (0.5) & 55.5 (0.1)\\

\multirow{-4}{1.4cm}{\centering\arraybackslash Hybrid, unequal} & yes & lasso & 0.6 (0.3) & 10.0 (0.2) & 10.0 (0.2) & 99.9 (0.1) & 55.6 (0.1)\\
\midrule

 & no & unadj & 0.3 (0.7) & 20.6 (0.5) & 20.5 (0.5) & 95.7 (0.7) & 82.4 (0.1)\\

 & yes & unadj & -0.2 (0.4) & 11.1 (0.2) & 11.1 (0.2) & 95.3 (0.7) & 46.8 (0.2)\\

 & no & lasso & -0.0 (0.2) & 7.1 (0.2) & 7.1 (0.2) & 98.2 (0.4) & 34.4 (0.1)\\

\multirow{-4}{1.4cm}{\centering\arraybackslash Triplet, equal} & yes & lasso & -0.2 (0.1) & 4.6 (0.1) & 4.6 (0.1) & 100.0 (0.0) & 34.6 (0.1)\\
\midrule

 & no & unadj & 0.7 (1.0) & 30.2 (0.6) & 30.2 (0.6) & 95.2 (0.7) & 118.6 (0.1)\\

 & yes & unadj & -0.4 (0.5) & 15.4 (0.3) & 15.4 (0.3) & 93.3 (0.8) & 59.8 (0.3)\\

 & no & lasso & 0.6 (0.4) & 13.0 (0.3) & 13.0 (0.3) & 96.5 (0.6) & 53.7 (0.1)\\

\multirow{-4}{1.4cm}{\centering\arraybackslash Triplet, unequal} & yes & lasso & 0.1 (0.2) & 8.0 (0.2) & 8.0 (0.2) & 100.0 (0.0) & 53.8 (0.1)\\
\bottomrule
\end{tabular}
\begin{tablenotes}
\item Note: The numbers in brackets are the corresponding standard errors estimated using the bootstrap with 500 replications. Bias, SD, RMSE, CP, Length, and their standard errors are multiplied by 100.
\end{tablenotes}
\end{threeparttable}
\end{table}

\subsubsection{Two semi-synthetic data sets}

% semi-synthetic 1
To evaluate the repeated sampling properties of different estimators \zhu{on the data set in Section \ref{sec:rd_ok}}, which depend on all the potential outcomes, we generate synthetic data based on the experimental data. We use Lasso to fit two sparse linear models for the treatment and control groups, respectively, and impute the unobserved potential outcomes using the fitted models. We use the same blocking and propensity scores for each block as those in the original experiment and implement either stratified randomization or stratified rerandomization in the design stage. For rerandomization, we use the following covariates: age, high school grades, and whether the students correctly answered two tests about the scholarship formula. We replicate the simulation 1000 times on this semi-synthetic data set.
Table \ref{tab:rd} shows the results.
Rerandomization is the preferable strategy, and the Lasso-adjusted estimator is superior to the unadjusted one. Specifically, the combination of rerandomization and Lasso adjustment reduces the standard deviation by 30\% and decreases the mean confidence interval lengths by 19\%.

% semi-synthetic 2
To further examine the repeated sampling properties of the Lasso-adjusted estimator \zhu{on the data set in Section \ref{sec:rd_fish}}, we use the same approach
\zhu{as the last paragraph} to impute the unobserved potential outcomes. We use the same blocking and propensity scores for each block as those in the matched data set and implement stratified randomization 1000 times on this semi-synthetic data set.
Table \ref{tab:rd} shows the results.
The distribution of the Lasso-adjusted estimator is more concentrated than that of the unadjusted estimator. Moreover, the Lasso-adjusted estimator reduces the standard deviation by 32\% and decreases the mean confidence interval length by 9\%. Thus, the Lasso-adjusted estimator improves both the estimation and inference efficiencies.

\begin{table}[p]

\caption{\label{tab:rd}Simulation results for two semi-synthetic data sets}
\centering
\begin{threeparttable}
\begin{tabular}[t]{ccrrrrr}
\toprule
\multicolumn{1}{c}{Rerand.} & \multicolumn{1}{c}{Est.} & \multicolumn{1}{c}{Bias} & \multicolumn{1}{c}{SD} & \multicolumn{1}{c}{RMSE} & \multicolumn{1}{c}{CP} & \multicolumn{1}{c}{Length}\\
\midrule
\addlinespace[0.3em]
\multicolumn{7}{c}{\textbf{OK data}}\\
\hspace{1em}no & unadj & -0.12 (0.15) & 4.61 (0.11) & 4.61 (0.11) & 97.40 (0.51) & 21.03 (0.03)\\
\hspace{1em}yes & unadj & 0.26 (0.14) & 4.38 (0.10) & 4.39 (0.10) & 97.90 (0.43) & 20.43 (0.03)\\
\hspace{1em}no & lasso & -0.31 (0.11) & 3.42 (0.08) & 3.43 (0.08) & 98.40 (0.41) & 17.25 (0.04)\\
\hspace{1em}yes & lasso & -0.01 (0.10) & 3.21 (0.07) & 3.21 (0.07) & 99.40 (0.23) & 17.14 (0.04)\\
\midrule
\addlinespace[0.3em]
\multicolumn{7}{c}{\textbf{Fish data}}\\
\hspace{1em}- & unadj & 0.12 (0.14) & 4.31 (0.09) & 4.31 (0.09) & 98.30 (0.41) & 19.94 (0.07)\\
\hspace{1em}- & lasso & 0.15 (0.11) & 3.68 (0.08) & 3.68 (0.08) & 98.80 (0.33) & 18.02 (0.07)\\
\bottomrule
\end{tabular}
\begin{tablenotes}
\item Note: The numbers in brackets are the corresponding standard errors estimated using the bootstrap with 500 replications. Bias, SD, RMSE, CP, Length, and their standard errors are multiplied by 100.
\end{tablenotes}
\end{threeparttable}
\end{table}

\clearpage

\subsubsection{Power issue}

\zk{In this section, we consider the same settings as those in Section \ref{sec:sim}, except we set different values for $\tau$ to examine the power of various methods.
Figure~\ref{fig:sim_power} shows the results.
The proposed Lasso-adjusted estimator may lose power when the true $\tau$ is small.
However, when the true $\tau$ is moderately large, the proposed Lasso-adjusted estimator demonstrates increased power compared to the unadjusted estimator.
}

\begin{figure}[H]
\centering
\includegraphics[scale=0.8]{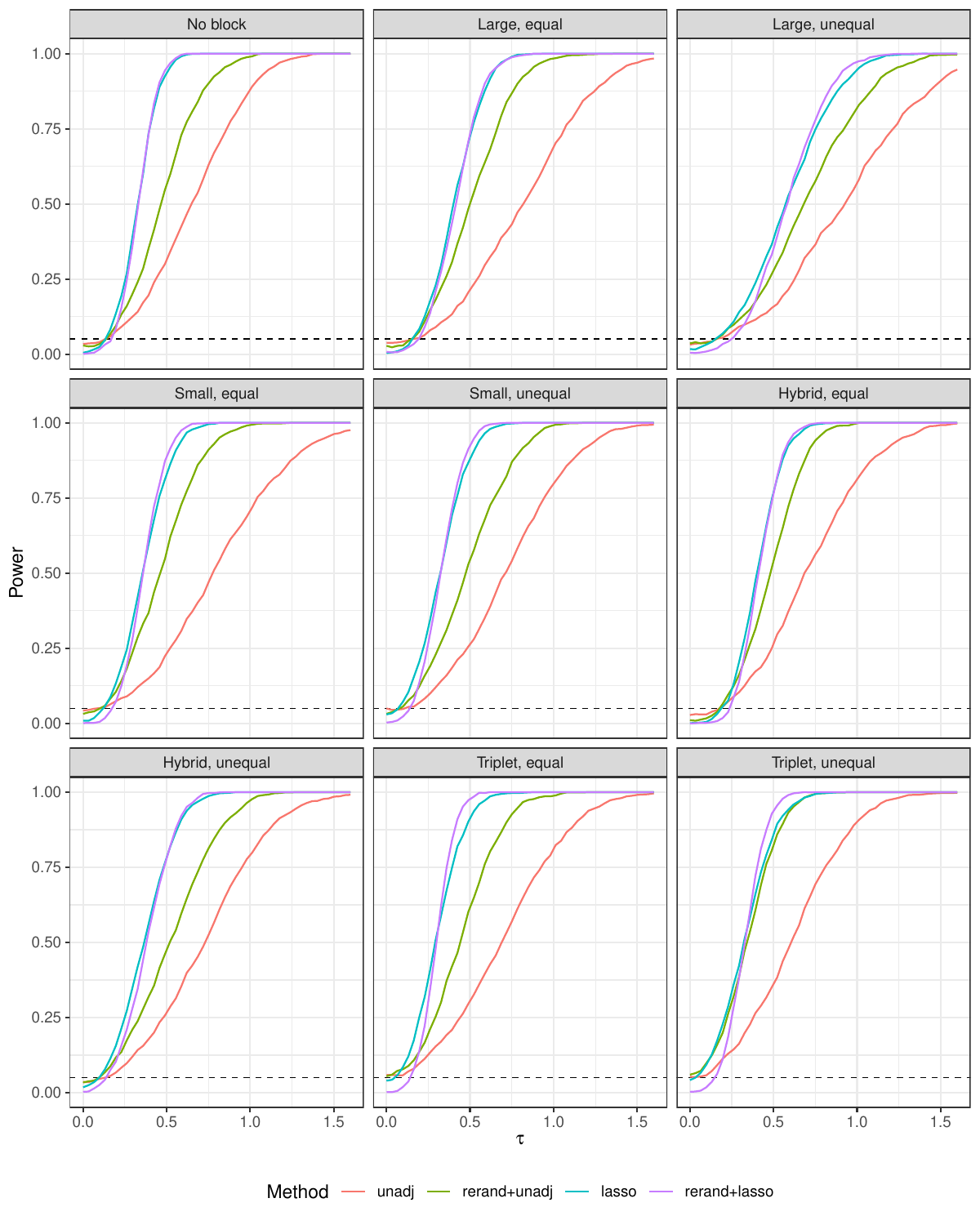}
\caption{Power comparison of different methods. The horizontal dashed line represents the significance level of 0.05.}
\label{fig:sim_power}
\end{figure}

% \bibliographystyle{agsm}
% \bibliography{causal}

\end{document}